\documentclass[english]{article}
\usepackage[T1]{fontenc}
\usepackage[latin9]{inputenc}
\usepackage{geometry}
\usepackage{array}
\usepackage{verbatim}
\usepackage{float}
\usepackage{mathrsfs}
\usepackage{mathtools}
\usepackage{bm}
\usepackage{amsmath}
\usepackage{amssymb}
\usepackage{graphicx}

\makeatletter

\providecommand{\tabularnewline}{\\}
\floatstyle{ruled}
\newfloat{algorithm}{tbp}{loa}
\providecommand{\algorithmname}{Algorithm}
\floatname{algorithm}{\protect\algorithmname}

 \usepackage{color}
\usepackage{gnuplottex}

\usepackage{tikz, pgf, pgfplots}
\usetikzlibrary{automata,positioning}
\usetikzlibrary{arrows,backgrounds,plotmarks}

\usepackage{cite}
\usepackage{amsthm}\usepackage{dsfont}\usepackage{array}\usepackage{mathrsfs}

\usepackage{hyperref}

\definecolor{ejc}{RGB}{255,0,0}
\definecolor{yxc}{RGB}{0,0,200}

\newcommand{\mcr}{\mathsf{MCR}}

\makeatother

\usepackage{babel}
\begin{document}
\theoremstyle{plain}\newtheorem{lem}{\textbf{Lemma}}\newtheorem{theorem}{\textbf{Theorem}}\newtheorem{fact}{\textbf{Fact}}\newtheorem{corollary}{\textbf{Corollary}}\newtheorem{assumption}{\textbf{Assumption}}\newtheorem{example}{\textbf{Example}}\newtheorem{definition}{\textbf{Definition}}\newtheorem{prop}{\textbf{Proposition}}

\theoremstyle{definition}

\theoremstyle{remark}\newtheorem{remark}{\textbf{Remark}}

\date{September 2016; ~Revised November 2016}

\title{The Projected Power Method:  An Efficient Algorithm for \\ Joint Alignment from Pairwise Differences}

\author{Yuxin Chen \thanks{Department of Statistics, Stanford
    University, Stanford, CA 94305, U.S.A.} \thanks{Department of Electrical Engineering, Princeton
    University, Princeton, NJ 08544, U.S.A.}  \and Emmanuel J. Cand\`es
  \footnotemark[1] \thanks{Department of Mathematics, Stanford
    University, Stanford, CA 94305, U.S.A.} }

\maketitle

\begin{abstract}

  Various applications involve assigning discrete label values to a
  collection of objects based on some pairwise noisy data. Due to the
  discrete---and hence nonconvex---structure of the problem, computing the optimal assignment (e.g.~maximum likelihood assignment) 
  becomes intractable at first
  sight.  This paper makes progress towards efficient computation 
   by focusing on a concrete joint alignment problem---that is,
  the problem of recovering $n$ discrete variables
  $x_i \in \{1,\cdots, m\}$, $1\leq i\leq n$ given noisy observations
  of their modulo differences $\{x_i - x_j~\mathsf{mod}~m\}$.  We
  propose a low-complexity and model-free procedure, which operates in a lifted
  space by representing distinct label values in orthogonal
  directions, and which attempts to optimize quadratic functions over
  hypercubes.  Starting with a first guess computed via a spectral
  method, the algorithm successively refines the iterates via
  projected power iterations. We prove that for a broad class of statistical models, 
  the proposed projected
  power method makes no error---and hence converges to the maximum likelihood estimate---in a
  suitable regime.  Numerical experiments have been carried out on
  both synthetic and real data to demonstrate the practicality of our
  algorithm. We expect this algorithmic framework to be effective for
  a broad range of discrete assignment problems.

\end{abstract}

\section{Introduction\label{sec:Intro}}

\subsection{Nonconvex optimization\label{sec:Nonconvex}}

Nonconvex optimization permeates almost all fields of science and engineering applications. For instance, consider the structured recovery problem where one wishes to recover some structured inputs $\bm{x}=[x_i]_{1\leq i\leq n}$ from noisy samples $\bm{y}$. The recovery procedure often involves solving some optimization problem (e.g.~maximum likelihood estimation)
\begin{eqnarray}
 \text{maximize}_{\bm{z}\in \mathbb{R}^n} \quad \ell(\bm{z};\bm{y}) \qquad  \text{subject to} \quad {\bm{z}\in \mathcal{S}}, \label{eq:opt}
\end{eqnarray}
where the objective function $\ell(\bm{z};\bm{y})$ measures how well a candidate $\bm{z}$ fits the samples. Unfortunately, this program (\ref{eq:opt}) may be highly nonconvex, depending on the choices of the goodness-of-fit measure $\ell(\cdot)$ as well as the feasible set $\mathcal{S}$. 
In contrast to convex optimization that has become the cornerstone of modern algorithm design, nonconvex problems are in general daunting to solve. 
Part of the challenges arises from the existence of (possibly exponentially many) local stationary points; in fact, oftentimes even checking local optimality for a feasible point proves NP-hard. 

Despite the general intractability, recent years have seen progress on
nonconvex procedures for several classes of problems, including
low-rank matrix recovery
\cite{keshavan2010few,keshavan2010matrix,jain2013low,sun2015guaranteed,chen2015fast,zheng2015convergent,tu2015low,ma2017implicit,zhao2015nonconvex,park2016provable,yi2016fast},
phase retrieval
\cite{netrapalli2013phase,candes2015phase,shechtman2014gespar,chen2015solving,cai2015optimal,sun2016geometric,chi2016kaczmarz,zhang2016provable,ma2017implicit,wang2016solving,zhang2016reshaped},
dictionary learning \cite{sun2015complete,sun2015nonconvex}, blind
deconvolution \cite{li2016rapid,ma2017implicit}, empirical risk minimization
\cite{mei2016landscape}, to name just a few.  For example, we have
learned that several problems of this kind provably enjoy benign
geometric structure when the sample complexity is sufficiently large,
in the sense that all local stationary points (except for the global
optimum) become saddle points and are not difficult to escape
\cite{sun2015nonconvex,sun2015complete,ge2016matrix,bhojanapalli2016global,lee2016gradient}. For
the problem of solving certain random systems of quadratic equations,
this phenomenon arises as long as the number of equations or sample
size exceeds the order of $n\log^3n$, with $n$ denoting the number of unknowns
\cite{sun2016geometric}.\footnote{This geometric property alone is not
  sufficient to ensure rapid convergence of an algorithm.} We have
also learned that it is possible to minimize certain non-convex random
functionals---closely associated with the famous phase retrieval
problem---even when there may be multiple local minima
\cite{candes2015phase,chen2015solving}. In such problems, one can find
a reasonably large basin of attraction around the global solution, in
which a first-order method converges geometrically fast. More
importantly, the existence of such a basin is often guaranteed even in
the most challenging regime with minimal sample complexity. Take the phase retrieval problem as an
example:  this basin exists as soon as the sample size is about the
order of $n$ \cite{chen2015solving}.  This motivates the development
of an efficient two-stage paradigm that consists of a
carefully-designed initialization scheme to enter the basin, followed
by an iterative refinement procedure that is expected to converge
within a logarithmic number of iterations
\cite{candes2015phase,chen2015solving}; see also
\cite{keshavan2010few,keshavan2010matrix} for related ideas in matrix completion.

In the present work, we extend the knowledge of nonconvex optimization
by studying a class of assignment problems in which each $x_i$ is
represented on a finite alphabet, as detailed in the next subsection.
Unlike the aforementioned problems like phase retrieval which are
inherently continuous in nature, in this work we are preoccupied with
an input space that is discrete and already nonconvex to start
with. We would like to contribute to understanding what is possible to
solve in this setting.

\subsection{A joint alignment problem\label{sec:Setup}}

This paper primarily focuses on the following joint discrete alignment problem. 
Consider a collection of $n$ variables $\left\{ x_{i}\right\} _{1\leq i\leq n}$,
where each variable can take $m$ different possible values, namely,
$x_{i}\in[m]:=\left\{ 1,\cdots,m\right\} $. Imagine we obtain a set of
pairwise difference samples $\left\{ y_{i,j}\mid(i,j)\in\Omega\right\} $
over some symmetric\footnote{We say $\Omega$ is symmetric if  $(i,j)\in \Omega$ implies $(j,i)\in \Omega$ for any $i$ and $j$.} index set $\Omega\subseteq\left[n\right]\times\left[n\right]$,
where $y_{i,j}$ is a noisy measurement of the modulo difference of the incident variables
\begin{equation}
  y_{i,j} ~\leftarrow~ x_{i}-x_{j}~\mathsf{mod}~m, \qquad (i,j) \in \Omega.
\end{equation}
For example, one might obtain a set of data $\{y_{i,j}\}$ where only 50\% of them are consistent with the truth $x_i - x_j~\mathsf{mod}~m$.  
The goal is to simultaneously recover all
$\left\{ x_{i}\right\} $ based on the measurements
$\left\{ y_{i,j}\right\} $, up to some unrecoverable global offset.\footnote{Specifically, it is impossible to distinguish
  the $m$ sets of inputs $\left\{ x_{i}\right\} _{1\leq i\leq n}$,
  $\left\{ x_{i}-1\right\} _{1\leq i\leq n}$, $\cdots$, and
  $\left\{ x_{i}-m+1\right\} _{1\leq i\leq n}$ even if we obtain perfect measurements of all pairwise
  differences $\{x_i-x_j~\mathsf{mod}~m: 1\leq i,j \leq n\}$. }

To tackle this problem, one is often led to the following program
\begin{eqnarray}
  \text{maximize}_{\left\{ z_{i}\right\} } &  & \sum_{ (i,j) \in \Omega} \ell\left(z_{i},z_{j};\text{ }y_{i,j}\right)  \label{eq:MLE}\\
  \text{subject to}\text{ } &  & z_{i}\in\left\{ 1,\cdots,m\right\} ,
  \quad i=1,\cdots,n, \nonumber 
\end{eqnarray}
where $\ell\left(z_{i},z_{j};\text{ }y_{i,j}\right)$ is some function
that evaluates how consistent the observed sample $y_{i,j}$
corresponds to the candidate solution $(z_{i},z_{j})$. For instance,
one possibility for $\ell$ may be
\begin{eqnarray}
  \ell\left(z_{i},z_{j};\text{ }y_{i,j}\right) =  \begin{cases}
  1,\qquad & \text{if }z_{i}-z_{j}=y_{i,j}~\mathsf{mod}~m,\\
  0, & \text{else},
  \end{cases}
  \label{eq:LL-consistency}
\end{eqnarray}
under which  the program (\ref{eq:MLE}) seeks a solution that maximizes the agreement between the paiwise observations and the recovery.   
Throughout the rest of the paper, we set $\ell(z_{i},z_{j};\text{ }y_{i,j})\equiv0$
whenever $(i,j)\notin\Omega$.

This joint alignment problem finds applications in multiple domains.  To begin with, the binary case (i.e.~$m=2$) deserves special attention, as it reduces to a graph partitioning problem. For instance, in a community detection scenario in which one wishes to partition all users into two clusters,  the variables $\{x_i\}$ to recover indicate the cluster assignments for each user, while $y_{i,j}$ represents the friendship between two users $i$ and $j$ (e.g.~\cite{karrer2011stochastic,coja2010graph,abbe2014exact,chaudhuri2012spectral,mossel2014consistency,jalali2011clustering,oymak2011finding,chen2012clustering,hajek2015achieving,chen2016community,javanmard2016phase,chin2015stochastic,globerson2015hard}). This allows to model, for example, the haplotype phasing problem arising in computational genomics \cite{si2014haplotype,chen2016community}. 
Another example is the problem of water-fat separation in magnetic resonance imaging (more precisely, in Dixon imaging). A crucial step is to determine, at each image pixel $i$, the phasor (associated with the field inhomogeneity) out of two possible candidates, represented by $x_i=1$ and $x_i=2$, respectively. The task takes as input some pre-computed pairwise cost functions $-\ell(x_i,x_j)$, which provides information about whether $x_i=x_j$ or not at pixels $i$ and $j$; see \cite{zhang2016resolving,hernando2010robust,berglund2011two} for details.

Moving beyond the binary case, this problem is motivated by the need
of jointly aligning multiple images/shapes/pictures that arises in
various fields. Imagine a sequence of images of the same physical
instance (e.g.~a building, a molecule), where each $x_i$ represents
the orientation of the camera when taking the $i$th image. A variety
of computer vision tasks (e.g.~3D reconstruction from multiple scenes)
or structural biology applications (e.g.~cryo-electron microscopy)
rely upon joint alignment of these images; or equivalently, joint
recovery of the camera orientations associated with each image.
Practically, it is often easier to estimate the relative camera
orientation between a pair of images using raw features
\cite{huang2013fine,wang2013exact,bandeira2014multireference}. The
problem then boils down to this: how to jointly aggregate such
pairwise information in order to improve the collection of camera pose
estimates?

\subsection{Our contributions}

In this work, we propose to solve the problem (\ref{eq:MLE}) via a novel model-free nonconvex procedure. Informally, the procedure starts by lifting each variable $x_i\in [m]$ to higher dimensions such that distinct values are represented in orthogonal directions, and then encodes the goodness-of-fit measure $\ell(x_i, x_j; y_{i,j})$ for each $(i,j)$ by an $m\times m$ matrix. This way of representation allows to recast (\ref{eq:MLE}) as a constrained quadratic program or, equivalently, a constrained principal component analysis (PCA) problem.  We then attempt optimization by means of projected power iterations, following an initial guess obtained via suitable low-rank factorization. 
This procedure proves effective for a broad family of statistical models, and might be interesting for many other Boolean assignment problems beyond joint alignment.

\section{Algorithm: projected power method}

In this section, we present a nonconvex procedure to solve the nonconvex problem (\ref{eq:MLE}), which entails a series of projected power iterations over a higher-dimensional space. In what follows, this algorithm will be termed a {\em projected power method} (PPM).

\subsection{Matrix representation\label{sec:matrix-rep}}

The formulation (\ref{eq:MLE}) admits an alternative matrix representation that is often more amenable to computation. To begin
with, each state $z_{i}\in\left\{ 1,\cdots,m\right\} $ can be represented
by a binary-valued vector $\bm{z}_{i}\in\left\{ 0,1\right\} ^{m}$
such that
\begin{equation}
  z_{i}=j \quad \Longleftrightarrow \quad 
  \bm{z}_{i} = \bm{e}_{j} \in \mathbb{R}^{m},
  \label{eq:state-vector}
\end{equation}
where $\bm{e}_{1}=\left[\begin{array}{c}
1\\
0\\
\vdots\\
0
\end{array}\right]$, $\bm{e}_{2}=\left[\begin{array}{c}
0\\
1\\
\vdots\\
0
\end{array}\right]$, $\cdots$, $\bm{e}_{m}=\left[\begin{array}{c}
0\\
0\\
\vdots\\
1
\end{array}\right]$ 
are the canonical basis vectors. In addition, for each pair $(i,j)$,
one can introduce an input matrix $\bm{L}_{i,j}\in\mathbb{R}^{m\times m}$
to encode $\ell(z_{i},z_{j};y_{i,j})$ given all possible input combinations
of $\left(z_{i},z_{j}\right)$; that is,
\begin{eqnarray}
  \left(\bm{L}_{i,j}\right)_{\alpha,\beta}  :=  \ell\left(z_{i}=\alpha,z_{j}=\beta;
  ~y_{i,j}\right),\qquad1\leq\alpha,\beta\leq m.
  \label{eq:Input-Matrix}
\end{eqnarray}
Take the choice (\ref{eq:LL-consistency}) of $\ell$ for example:  
\begin{equation}
  \left(\bm{L}_{i,j}\right)_{\alpha,\beta} = \begin{cases}
  1,\quad & \text{if }\alpha-\beta=y_{i,j} ~\mathsf{mod}~m,\\
  0, & \text{else},
  \end{cases}
  \quad\forall(i,j)\in\Omega;
  \label{eq:L-random-corruption}
\end{equation}
in words, $\bm{L}_{i,j}$ is a cyclic permutation matrix obtained by
circularly shifting the identity matrix $\bm{I}_{m} \in \mathbb{R}^{m\times m}$ by $y_{i,j}$
positions.
By convention, we take $\bm{L}_{i,i}\equiv\bm{0}$ for all $1\leq i\leq n$ and $\bm{L}_{i,j}\equiv\bm{0}$
for all $(i,j)\notin\Omega$. 

The preceding notation enables the quadratic form representation
\[
  \ell\left(z_{i},z_{j};y_{i,j}\right)= \bm{z}_{i}^{\top}\bm{L}_{i,j} \bm{z}_{j}.
\]
For notational simplicity, we stack all the $\bm{z}_{i}$'s and the
$\bm{L}_{i,j}$'s into a concatenated vector and matrix
\begin{equation}
  \bm{z}=\left[  \begin{array}{c}
  \bm{z}_{1}\\
  \vdots\\
  \bm{z}_{n}
  \end{array} \right] \in \mathbb{R}^{nm} \quad\text{and}\quad \bm{L} = \left[\begin{array}{ccc}
  \bm{L}_{1,1} & \cdots & \bm{L}_{1,n}\\
  \vdots & \ddots & \vdots\\
  \bm{L}_{n,1} & \cdots & \bm{L}_{n,n}
  \end{array}  \right]
  \in \mathbb{R}^{nm \times nm}
  \label{eq:state-vector-matrix}
\end{equation}
respectively, representing the states and log-likelihoods altogether.
As a consequence, our problem can be succinctly recast as a 
constrained quadratic program: 
\begin{eqnarray}
  \text{maximize}_{\bm{z}} &  & \bm{z}^{\top} \bm{L} \bm{z}     \label{eq:MLE-matrix}\\
  \text{subject to} &  & \bm{z}_{i}\in\left\{ \bm{e}_{1},\cdots,\bm{e}_{m}\right\} ,\quad i=1,\cdots,n.\nonumber 
\end{eqnarray}
This representation is appealing due to the simplicity of the objective function regardless of the landscape of $\ell(\cdot,\cdot)$, which allows one to focus on quadratic optimization rather than optimizing the (possibly complicated) function $\ell(\cdot,\cdot)$ directly.

There are other families of $\bm{L}$ that also lead to the problem (\ref{eq:MLE}).
We single out a simple, yet, important family obtained by enforcing
global scaling and offset of $\bm{L}$. Specifically, the solution
to (\ref{eq:MLE-matrix}) remains unchanged
if each $\bm{L}_{i,j}$ is replaced by\footnote{This is because
  $ \bm{z}_i^{\top} (a \bm{L}_{i,j} + b\bm{1} \bm{1}^{\top}) \bm{z}_j   = a \bm{z}_i^{\top} \bm{L}_{i,j}  \bm{z}_j + b$ given that $\bm{1}^\top \bm{z}_i = \bm{1}^\top \bm{z}_j = 1$.} 
\begin{equation}
  \bm{L}_{i,j}  \leftarrow a\bm{L}_{i,j}+b\cdot\bm{1}\bm{1}^{\top},  \qquad (i,j)\in\Omega
\end{equation}
for some numerical values $a>0$ and $b\in\mathbb{R}$. 
Another important instance in this family is the debiased
version of $\bm{L}$---denoted by $\bm{L}^{\mathrm{debias}}$---defined
as follows
\begin{eqnarray}
  \bm{L}_{i,j}^{\mathrm{debias}}  =  \bm{L}_{i,j}-\frac{\bm{1}^{\top}\bm{L}_{i,j}\bm{1}}{m^{2}}\cdot\bm{1}\bm{1}^{\top},\quad1\leq i,j\leq n,
  \label{eq:debiased-MLE}
\end{eqnarray}
which essentially removes the empirical average of ${\bm{L}_{i,j}}$
in each block.

\subsection{Algorithm \label{sec:Algorithm}}

One can interpret the quadratic program (\ref{eq:MLE-matrix})
as finding the principal component of $\bm{L}$ subject to certain
structural constraints. This motivates us to tackle this constrained
PCA problem by means of a power method, with the assistance of appropriate
regularization to enforce the structural constraints. More precisely,
we consider the following procedure, which starts from a suitable
initialization $\bm{z}^{(0)}$ and follows the update rule 
\begin{equation}
  \bm{z}^{(t+1)}=\mathcal{P}_{\Delta^{n}}\left(\mu_{t}\bm{L}\bm{z}^{(t)}\right),\quad \forall t\geq 0
  \label{eq:iterative-step}
\end{equation}
for some scaling parameter $\mu_{t}\in \mathbb{R}_+ \cup \{ \infty \}$. Here, $\mathcal{P}_{\Delta^{n}}$
represents block-wise projection onto the standard simplex, namely,
for any vector $\bm{z}=\left[ \bm{z}_i \right]_{1\leq i\leq n}\in\mathbb{R}^{nm}$,
\begin{equation}
  \mathcal{P}_{\Delta^{n}}\left(\boldsymbol{z}\right) := \left[ \begin{array}{c}
  \mathcal{P}_{\Delta}\left(\bm{z}_{1}\right)\\
  \vdots\\
  \mathcal{P}_{\Delta}\left(\bm{z}_{n}\right)
  \end{array} \right],
  \label{eq:block-wise-proj}
\end{equation}
where $\mathcal{P}_{\Delta}\left(\bm{z}_{i}\right)$ is the projection
of $\bm{z}_{i} \in \mathbb{R}^m$ onto the standard simplex 
\begin{equation}
  \Delta := \left\{ \bm{s}\in\mathbb{R}^{m}\mid\bm{1}^{\top}\bm{s}=1;\text{ }\bm{s}\text{ is non-negative}\right\} .\label{eq:simplex}
\end{equation}
In particular, when $\mu_{t}=\infty$, $\mathcal{P}_{\Delta}\left(\mu_{t}\bm{z}_{i}\right)$
reduces to a \emph{rounding procedure}. Specifically, if the largest entry of $\bm{z}_{i}$ is strictly larger than its second largest entry, then one has
\begin{equation}
	\lim_{\mu_t \rightarrow \infty}\mathcal{P}_{\Delta}\left(\mu_{t}\bm{z}_{i}\right) = \bm{e}_j
\end{equation}
with $j$ denoting the index of the largest entry of $\bm{z}_{i}$; see Fact \ref{fact:Proj-separation} for a justification.   

The key advantage of the PPM is its computational
efficiency: the most expensive step in each iteration lies in matrix
multiplication, which can be completed in nearly linear time,
 i.e.~in time  $O\left(|\Omega|m\log m\right)$.\footnote{Here and throughout, the standard notion $f(n)=O\left(g(n)\right)$
or $f(n)\lesssim g(n)$ mean there exists a constant $c>0$ such that
$|f(n)|\leq cg(n)$; $f(n)=o\left(g(n)\right)$ means $\underset{n\rightarrow\infty}{\lim}f(n)/g(n)=0$;
$f(n)\gtrsim g(n)$ means there exists a constant $c>0$ such that $|f(n)|\geq cg(n)$;
and $f(n)\asymp g(n)$ means there exist constants $c_{1}, c_{2}>0$
such that $c_{1}g(n)\leq |f(n)| \leq c_{2}g(n)$.} This arises
  from the fact that each block $\bm{L}_{i,j}$ is circulant, so  we
  can compute a matrix-vector product using at most two $m$-point
  FFTs. The projection step $\mathcal{P}_{\Delta}$ can be performed in $O\left(m\log m\right)$ flops via a
sorting-based algorithm (e.g.~\cite[Figure 1]{duchi2008efficient}), and hence $\mathcal{P}_{\Delta^n}(\cdot)$ is much cheaper than the matrix-vector multiplication $\bm{L} \bm{z}^{(t)}$ given that $|\Omega| \gg n$ occurs in most applications.

One important step towards guaranteeing rapid convergence is to identify
a decent initial guess $\bm{z}^{(0)}$. This is accomplished by low-rank factorization as follows
\begin{enumerate}
\item Compute the best rank-$m$ approximation of the input matrix $\bm{L}$,
namely,
\begin{equation}
  \hat{\bm{L}}:=\arg\min_{\bm{M}:\text{ rank} (\bm{M}) \leq m}  \Vert \bm{M}-\bm{L} \Vert _{\mathrm{F}},
  \label{eq:defn-Lm}
\end{equation}
where $\|\cdot\|$ represents the Frobenius norm;

\item Pick a random column $\hat{\bm{z}}$ of $\hat{\bm{L}}$ and set the
initial guess as $\bm{z}^{(0)}=\mathcal{P}_{\Delta^{n}}\left(\mu_{0}\hat{\bm{z}}\right)$. 

\end{enumerate}
\begin{remark}Alternatively, one can take $\hat{\bm{L}}$ to be the
best rank-$(m-1)$ approximation of the debiased input matrix $\bm{L}^{\mathrm{debias}}$
defined in (\ref{eq:debiased-MLE}), which can be computed in a slightly faster manner.
\end{remark}

\begin{remark}
A natural question arises as to whether the algorithm works with an arbitrary initial point. This question has been studied by \cite{bandeira2016low} for the more special stochastic block models, which shows that under some (suboptimal) conditions, all second-order critical points correspond to the truth and hence an arbitrary initialization works.  However, the condition presented therein is much more stringent than the optimal threshold  \cite{abbe2014exact,mossel2014consistency}. Moreover, it is unclear whether a local algorithm like the PPM can achieve optimal computation time without proper initialization. All of this would be interesting for future investigation. 
\end{remark}

The main motivation comes from the (approximate) low-rank structure
of the input matrix $\bm{L}$. As we shall shortly see, in many scenarios the data matrix 
is approximately of rank $m$ if the samples are noise-free. Therefore, a low-rank approximation
of $\bm{L}$ serves as a denoised version of the data, which is expected to reveal much information about the truth. 

 The low-rank factorization step  can be performed efficiently via the method of orthogonal iteration
(also called block power method) \cite[Section 7.3.2]{golub2012matrix}. Each power iteration consists of a matrix product of the form $\bm{L}\bm{U}$ as well as a QR decomposition of some matrix $\bm{V}$,  where $\bm{U}, \bm{V} \in \mathbb{R}^{nm\times m}$. The matrix product can be computed in $O( |\Omega| m^2 \log m )$ flops with the assistance of $m$-point FFTs, whereas the QR decomposition takes time $O(nm^3)$. In summary, each power iteration runs in time  $O( |\Omega| m^2 \log m + nm^3 )$. Consequently, the matrix product constitutes the main computational cost when $m \lesssim (|\Omega| \log m) / n$, while the QR decomposition becomes the  bottleneck when $m \gg |\Omega| \log m / n$.

It is noteworthy that both the initialization and the 
refinement we propose are model-free, which do not make any
assumptions on the data model. The whole algorithm is summarized in
Algorithm \ref{alg:PP}. There is of course the question of what
sequence $\{\mu_t\}$ to use, which we defer to Section
\ref{sec:theory}.

\begin{algorithm*}[t]
\caption{ Projected power method.\label{alg:PP}}
\label{Algorithm:ProjectedPower}
\begin{tabular}{>{\raggedright}p{1\textwidth}}
\textbf{Input}: 
the input matrix $\bm{L}=\left[\bm{L}_{i,j}\right]_{1\leq i,j\leq n}$; the
scaling factors $\left\{ \mu_{t}\right\} _{t\geq0}$. \vspace{0.7em}
\tabularnewline
\textbf{Initialize} $\boldsymbol{z}^{(0)}$ to be $\mathcal{P}_{\Delta^{n}}\left(\mu_{0}\hat{\bm{z}}\right)$
as defined in (\ref{eq:block-wise-proj}), where $\hat{\bm{z}}$ is
a random column of the best rank-$m$ approximation $\hat{\bm{L}}$
of $\bm{L}$. \vspace{0.3em}
\tabularnewline
\textbf{Loop: for $t=0,1,\cdots,T-1$ do} 
\begin{equation}
  \bm{z}^{(t+1)}=\mathcal{P}_{\Delta^{n}} \big( \mu_{t}\bm{L}\bm{z}^{(t)} \big),
  \label{eq:iterative-step-PP}
\end{equation}
where $\mathcal{P}_{\Delta^{n}}\left(\cdot\right)$ is as defined in
(\ref{eq:block-wise-proj}).
\tabularnewline
\textbf{Output }$\{ \hat{x}_{i} \} _{1\leq i\leq n}$, where $\hat{x}_{i}$
is the index of the largest entry of the block $\boldsymbol{z}_{i}^{(T)}$.
\tabularnewline
\end{tabular}
\end{algorithm*}

The proposed two-step algorithm, which is based on proper initialization followed by successive projection onto product of simplices, is  a new paradigm for solving a class of discrete optimization problems. As we will detail in the next section, it is provably effective for a family of statistical models. On the other hand, we 
remark that there exist many algorithms of a similar flavor to tackle other generalized eigenproblems, including but not limited to sparse PCA \cite{journee2010generalized,hein2010inverse,yuan2013truncated}, water-fat separation \cite{zhang2016resolving}, the hidden clique problem \cite{deshpande2015finding}, phase synchronization \cite{boumal2016nonconvex,liu2016estimation}, cone-constrained PCA \cite{deshpande2014cone}, and automatic network analysis \cite{wang2012automated}. These algorithms are variants of the projected power method, which 
 combine proper power iterations with additional procedures to promote sparsity or enforce other feasibility constraints. For instance, Deshpande et al.~\cite{deshpande2014cone} show that under some simple models, cone-constrained PCA can be efficiently computed using a generalized projected power method, provided that the cone constraint is convex.  The current work adds a new instance to this growing family of nonconvex methods.    

\section{Statistical models and main results}
\label{sec:theory}

This section explores the performance guarantees of the projected power method. 
We assume that $\Omega$ is obtained via random sampling at an observation
rate $p_{\mathrm{obs}}$ so that each $(i,j), ~i > j$ is included in
$\Omega$ independently with probability $p_{\mathrm{obs}}$, and
$\Omega$ is assumed to be independent of the measurement noise. In addition, we assume that the samples $\{y_{i,j} \mid i>j\}$ are independently generated. While the
independence noise assumption may not hold in reality, 
it serves as a starting point for us to develop a quantitative theoretical understanding for the effectiveness of the projected power method.
 This is also a common assumption in
  the literature (e.g.~\cite{singer2011angular,wang2013exact,
    huang2013consistent, chen2014near, 
    liu2016estimation,boumal2016nonconvex,PachauriKS13}).

With the above assumptions in mind, the MLE is exactly given by (\ref{eq:MLE}), with  $\ell\left(z_{i},z_{j};\text{ }y_{i,j}\right)$ representing the
log-likelihood (or some other equivalent function) of the candidate
solution $(z_{i},z_{j})$ given the outcome $y_{i,j}$.  Our key finding is that the PPM is not only  much more practical than computing the MLE directly\footnote{Finding the MLE here is an NP hard problem, and in general cannot be solved within polynomial time. Practically, one might attempt to compute it via convex relaxation (e.g.~\cite{huang2013consistent,bandeira2014multireference,chen2014near}), which is much more expensive than the PPM. }, but also  capable of achieving nearly identical statistical accuracy as the
MLE in a variety of scenarios.

Before proceeding to our results, we find it convenient to introduce a block sparsity
metric. Specifically, the {\em block sparsity} of a vector $\bm{h}=\left\{ \bm{h}_{i}\right\} _{1\leq i\leq n}$
is defined and denoted by
\[
  \|\bm{h}\|_{*,0}:=\sum _{i=1}^{n}\mathbb{I}\left\{ \bm{h}_{i} \neq 0\right\} ,
\]
where $\mathbb{I}\left\{ \cdot\right\} $ is the indicator function.
Since one can only hope to recover $\bm{x}$ up to some global offset,
we define the {\em misclassification rate} as the normalized block sparsity of the estimation error modulo the global shift
%
\begin{equation}
  \mcr(\hat{\bm{x}},\bm{x}) :=  \frac{1}{n} \min_{0\leq l<m}\left\Vert \hat{\bm{x}}-\mathsf{shift}_{l} (\bm{x}) \right\Vert _{*,0}.
   \label{eq:mcr}
\end{equation}
Here, $\mathsf{shift}_{l}\left(\bm{x}\right):=\left[\mathsf{shift}_{l}(\bm{x}_{i})\right]_{1\leq i\leq n}\in\mathbb{R}^{mn}$,
where $\mathsf{shift}_{l}(\bm{x}_{i})\in\mathbb{R}^{m}$ is obtained
by circularly shifting the entries of $\bm{x}_{i}\in\mathbb{R}^{m}$
by $l$ positions. Additionally, we let
$\log\left(\cdot\right)$ represent the natural logarithm throughout this paper.

\subsection{Random corruption model}

While our goal is to accommodate a general class of noise models, it is helpful to start with a concrete and simple example---termed a \emph{random
  corruption model}---such that
\begin{equation}
  y_{i,j}=
  \begin{cases}
    x_{i}-x_{j}\text{ }\mathsf{mod}\text{ }m,\qquad & \text{with probability \ensuremath{\pi_{0}}},\\
    \mathsf{Unif}\left(m\right), & \text{else},
  \end{cases}
  \qquad(i,j)\in\Omega,
  \label{eq:random-outlier}
\end{equation}
with $\mathsf{Unif}\left(m\right)$ being the uniform distribution over
$\left\{ 0,1,\cdots,m-1\right\} $. We will term the parameter $\pi_{0}$ 
the \emph{non-corruption rate}, since with probability $1-\pi_{0}$ the
observation behaves like a random noise carrying no information
whatsoever. Under this single-parameter model,  one can
write
\begin{eqnarray}
  \ell\left(z_{i},z_{j};\text{ }y_{i,j}\right) =  \begin{cases}
  \log\left(\pi_{0}+\frac{1-\pi_{0}}{m}\right),\qquad & \text{if }z_{i}-z_{j}=y_{i,j}~\mathsf{mod}~m,\\
  \log\left(\frac{1-\pi_{0}}{m}\right), & \text{else},
  \end{cases}
  \qquad   (i,j) \in \Omega.
  \label{eq:LL-random-noise}
\end{eqnarray}
Apart from its mathematical simplicity, the random corruption model somehow corresponds to the worst-case situation
since the uniform noise enjoys the highest entropy among all distributions
over a fixed range, thus forming a reasonable benchmark for practitioners.

 
Additionally, while Algorithm \ref{alg:PP}
can certainly be implemented using the formulation
\begin{equation}
  \left(\bm{L}_{i,j}\right)_{\alpha,\beta} = \begin{cases}
  \log\big( \pi_0 + \frac{1-\pi_0}{m} \big),\quad & \text{if }\alpha-\beta=y_{i,j} ~\mathsf{mod}~m,\\
  \log\big( \frac{1-\pi_0}{m} \big), & \text{else},
  \end{cases}
  \qquad  (i,j)\in\Omega,
  \label{eq:L-random-corruption-LL}
\end{equation}
we recommend taking (\ref{eq:L-random-corruption}) as the input matrix
in this case. It is easy to verify that (\ref{eq:L-random-corruption-LL}) and (\ref{eq:L-random-corruption}) are equivalent up to global scaling and offset, but  (\ref{eq:L-random-corruption}) is parameter free
and hence practically more appealing. 

We show that the PPM is guaranteed
to work even when the non-corruption rate $\pi_{0}$ is vanishingly
small, which corresponds to the scenario where almost all acquired
measurements behave like random noise. A formal statement is this:

\begin{theorem}\label{thm:ExactRecovery-outlier}
Consider the random corruption model (\ref{eq:random-outlier}) and the input matrix $\bm{L}$
given in (\ref{eq:L-random-corruption}). Fix $m>0$, and suppose
$p_{\mathrm{obs}}>c_{1}\log n/n$ and $\mu_{t}>c_{2}/\sigma_{2}\left(\bm{L}\right)$
for some sufficiently large constants $c_{1},c_{2}>0$. Then there
exists some absolute constant $0<\rho<1$ such that with probability
approaching one as $n$ scales, the iterates of Algorithm \ref{alg:PP} obey 
\begin{equation}
  \mcr(\bm{z}^{(t)},\bm{x}) \leq  0.49\rho^{t},\qquad\forall t\geq0,
  \label{eq:error-contraction-main}
\end{equation}
provided that the non-corruption rate $\pi_{0}$ exceeds\footnote{Theorem \ref{thm:ExactRecovery-outlier} continues to hold if we replace
1.01 with any other constant in (1,$\infty$).}
\begin{equation}
  \pi_{0}>2\sqrt{\frac{1.01\log n}{mnp_{\mathrm{obs}}}}.
  \label{eq:IT-sufficient-dense}
\end{equation}
\end{theorem}

\begin{remark}
\label{remark:sigma2}
Here and throughout,
$\sigma_{i}\left(\bm{L}\right)$ is the $i$th largest singular value
of $\bm{L}$. In fact, one can often replace $\sigma_{2}\left(\bm{L}\right)$
with $\sigma_{i}(\bm{L})$ for other $2\leq i< m$. But $\sigma_{1}\left(\bm{L}\right)$
is usually not a good choice unless we employ the debiased version of
$\bm{L}$ instead, because $\sigma_{1}\left(\bm{L}\right)$ typically
corresponds to the ``direct current'' component of $\bm{L}$ which
could be excessively large. In addition, we note that $\sigma_i(\bm{L})~(i\leq m)$ have been computed during spectral initialization and, as a result, will not result in extra computational cost.
\end{remark}

\begin{remark}
	As will be seen in Section \ref{sec:Iterative-stage-general}, a stronger version of error contraction arises such that
	\begin{equation} \label{eq:error-contraction-remark}
		\mathsf{MCR}(\bm{z}^{(t+1)}, \bm{x}) \leq \rho \hspace{0.1em} \mathsf{MCR}(\bm{z}^{(t)}, \bm{x}), \qquad \text{if }\mathsf{MCR}(\bm{z}^{(t)}, \bm{x}) \leq 0.49.
	\end{equation}
	This is a uniform result in the sense that (\ref{eq:error-contraction-remark}) occurs simultaneously for all $\bm{z}^{(t)}\in \Delta^n$ obeying $\mathsf{MCR}(\bm{z}^{(t)}, \bm{x}) \leq 0.49$, regardless of the preceding iterates $\{\bm{z}^{(0)}, \cdots, \bm{z}^{(t-1)} \}$ and the statistical dependency between $\bm{z}^{(t)}$ and $\{y_{i,j}\}$. In particular, if $\mu_t = \infty$, one has $\bm{z}^{(t)}\in \{ \bm{e}_1,\cdots, \bm{e}_m \}^n $, and hence  $\{ \bm{z}^{(t)} \}$ forms a sequence of feasible iterates with increasing accuracy. In this case, the iterates become accurate whenever  $\mathsf{MCR}(\bm{z}^{(t)}, \bm{x}) < 1/n$. 
\end{remark}

\begin{remark}
  The contraction rate $\rho$ can actually be as small as
  $O \left(1/(\pi_{0}^{2}np_{\mathrm{obs}}) \right)$, which is at most
  $O \left(1/\log n \right)$ if $m$ is fixed and if the condition
  (\ref{eq:IT-sufficient-dense}) holds. 
\end{remark}

According to Theorem \ref{thm:ExactRecovery-outlier}, convergence to
the ground truth can be expected in at most $O\left(\log n\right)$
iterations. This together with the per-iteration cost (which is on
the order of $|\Omega|$ since $\bm{L}_{i,j}$ is a cyclic permutation
matrix) shows that  the computational complexity of the iterative stage
is at most $O(|\Omega| \log n)$. This is nearly optimal since even
reading all the data and likelihood values take time about the order
of $|\Omega|$. All of this happens as soon as the corruption rate
does not exceed $1-O\left(\sqrt{\frac{\log n}{mnp_{\mathrm{obs}}}}\right)$,
uncovering the remarkable ability of the PPM to
tolerate and correct dense input errors.

As we shall see later in Section \ref{sec:Iterative-stage-general}, Theorem \ref{thm:ExactRecovery-outlier} holds as long as the algorithm starts with any  initial guess $\bm{z}^{(0)}$  obeying  
\begin{equation}
	\mcr ( \bm{z}^{(0)}, \bm{x} ) \leq  0.49, 
	\label{eq:init-condition}
\end{equation}
irrespective of whether $\bm{z}^{(0)}$ is independent of
the data $\{y_{i,j}\}$ or not.  Therefore, it often suffices to run the power method for a constant number of iterations during the initialization stage, which can be completed in  $O(|\Omega|)$ flops when $m$ is fixed.
The broader implication is that  Algorithm \ref{alg:PP}
remains successful if one adopts other initialization that can enter
the basin of attraction.

Finally, our result is sharp: to be sure, the error-correction
capability of the projected power method is statistically optimal, as
revealed by the following converse result.

\begin{theorem}
\label{thm:lower-bound-RCM}
Consider the random corruption
model (\ref{eq:random-outlier}) with any fixed $m>0$, and suppose
$p_{\mathrm{obs}}>c_{1}\log n/n$ for some sufficiently large constant
$c_{1}>0$. If \footnote{Theorem \ref{thm:lower-bound-RCM} continues to hold if we replace
0.99 with any other constant between 0 and 1.} 
\begin{equation}
  \pi_{0}<2\sqrt{\frac{0.99\log n}{mnp_{\mathrm{obs}}}},
  \label{eq:IT-necessary-dense}
\end{equation}
then the minimax probability of error 
\[
  \inf_{\hat{\bm{x}}}\max_{x_{i} \in [m], 1\leq i\leq n}
  \mathbb{P}\left( \mcr \left(\hat{\bm{x}},\bm{x}\right) > 0 \mid \bm{x} \right) \rightarrow 1,
  \qquad\text{as }n  \rightarrow  \infty,
\]
where the infimum is taken over all estimators and $\bm{x}$ is the
vector representation of $\left\{ x_{i}\right\}$  as before. 
\end{theorem}

As mentioned before, the binary case $m=2$ bears some similarity with the community detection problem in the presence of two communities. Arguably the most popular model for community detection is the stochastic block model (SBM), where any two vertices within the same cluster (resp.~across different clusters) are connected by an edge with probability $p$ (resp.~$q$). The asymptotic limits for both exact and partial recovery have been extensively studied \cite{massoulie2014community, abbe2014exact, mossel2014consistency,hajek2015achieving, abbe2015community, chin2015stochastic, brito2016recovery,guedon2015community}. We note, however, that the primary focus of community detection lies in the  sparse regime $p,q\asymp 1/n$ or logarithmic sparse regime (i.e.~$p,q\asymp \log n/n$), which is in contrast to the joint alignment problem in which the measurements are often considerably denser. There are, however, a few theoretical results that cover the dense regime, e.g.~\cite{mossel2014consistency}. To facilitate comparison, consider the case where $p = \frac{1+\pi_0}{2}$ and $q = \frac{1-\pi_0}{2}$ for some $\pi_0>0$, then the SBM reduces to the random corruption model with $p_{\mathrm{obs}}=1$. One can easily verify that the limit $\pi_0 = \sqrt{2\log n/n}$ we derive matches the recovery threshold given in\footnote{Note that the model studied in \cite{mossel2014consistency} is an SBM with $2n$ vertices with $n$ vertices
belonging to each cluster. Therefore, the threshold chacterization \cite[Proposition 2.9]{mossel2014consistency}
should read 
\[
  \frac{\sigma\sqrt{n/2}}{p-q}\exp\left(-\frac{n(p-q)^{2}}{4\sigma^{2}}\right) \rightarrow0
\]
when applied to our setting, with $\sigma:=\sqrt{p(1-q)+q(1-p)} \approx \sqrt{1/2}$.
} \cite[Theorem 2.5 and Proposition 2.9]{mossel2014consistency}.

\subsection{More general noise models}

The theoretical guarantees we develop for the random corruption model are 
special instances of a set of more general results. In this subsection, we cover a far more general class of noise models such that
\begin{equation}
  y_{i,j}\overset{\text{ind.}}{=}x_{i}-x_{j}+\eta_{i,j} ~\mathsf{mod} ~m, \qquad(i,j)\in\Omega, \label{eq:statistical-model}
\end{equation}
where the additive noise $\eta_{i,j} $ ($i>j$) are
i.i.d.~random variables supported on $\{0,1,\cdots,m-1\}$. In what follows, we define $P_0(\cdot)$ to be the distribution of $\eta_{i,j}$, i.e.
\begin{equation}
	P_0\left(y\right) \hspace{0.1em} = \hspace{0.1em} \mathbb{P}\left(\eta_{i,j}=y\right),\qquad 0\leq y < m.
  \label{eq:noise-pdf}
\end{equation}
For instance,  the random corruption model (\ref{eq:random-outlier}) is a special case of (\ref{eq:statistical-model}) with the noise distribution
\begin{equation}
	P_0(y) = \begin{cases} \pi_0 + \frac{1-\pi_0}{m},  \qquad & \text{if }y=0; \\  \frac{1-\pi_0}{m},  \quad & \text{if }y=1,\cdots, m-1. \end{cases}
\end{equation}

To simplify notation, we set $P_0(y)=P_0(y~\mathsf{mod}~m)$ for all $y\notin \{0,\cdots,m-1\}$ throughout the paper. 
Unless otherwise noted, we
take $\eta_{j,i}= - \eta_{i,j}$ for all $(i,j)\in\Omega$, and  restrict attention to the class of \emph{symmetric}
noise distributions obeying  
\begin{equation}
	P_0\left(y\right) \hspace{0.1em} = \hspace{0.1em} P_0 \left(m-y\right),\qquad y = 1,\cdots, m-1,
  \label{eq:assumption-symmetric}
\end{equation}
which largely simplifies the exposition.

\subsubsection{Key metrics}

The feasibility of accurate recovery necessarily depends on the noise
distribution or, more precisely, the distinguishability of the output
distributions $\left\{ y_{i,j}\right\} $ given distinct inputs. In
particular, there are $m$ distributions $\left\{ P_{l}\right\} _{0\leq l< m}$
that we'd like to emphasize, where $P_{l}\left(\cdot\right)$ represents
the distribution of $y_{i,j}$ conditional on $x_{i}-x_{j}=l$. Alternatively,
$P_{l}$ is also the $l$-shifted distribution of the noise $\eta_{i,j}$
given by
\begin{equation}
  P_{l}(y):=\mathbb{P}(y_{i,j}=y\mid x_{i}-x_{j}=l)=\mathbb{P}\big(\eta_{i,j}=y-l\big).
  \label{eq:defn-Pl}
\end{equation}
Here and below, we write $a-b$ and $a-b\text{ }\mathsf{mod}\text{ }m$
interchangeably whenever it is clear from the context, and adopt the
cyclic notation $c_{l}=c_{l+m}$ ($l\in\mathbb{Z}$) for any quantity
taking the form $c_{l}$. 

We would like to quantify the distinguishability of these distributions
via some  distance metric. One candidate is
the Kullback\textendash Leibler (KL) divergence 
defined by 
\begin{eqnarray}
  \mathsf{KL}\left(P_{i}\hspace{0.3em}\|\hspace{0.3em}P_{l}\right) 
  :=  \sum\nolimits _{y}P_{i}(y) \log\frac{P_{i} (y)}{ P_l(y) },\qquad0\leq i,l<m,
  \label{eq:KL-defn}
\end{eqnarray}
which plays an important role in our main theory.

\subsubsection{Performance guarantees}

We now proceed to the main findings. To simplify matters, we shall concern ourselves primarily with the kind of noise distributions obeying the following assumption.  

\begin{assumption}
\label{assumption-Pmin}
$m\min_{y}P_{0}\left(y\right)$
is bounded away from 0.
\end{assumption}

\begin{remark}
	When $m=O(1)$, one can replace $m\min_{y}P_{0}\left(y\right)$ with $\min_{y}P_{0}\left(y\right)$ in 
	Assumption \ref{assumption-Pmin}. However, if $m$ is allowed to scale with $n$---which is the case in Section \ref{sub:large-m-extension}---then the prefactor $m$ cannot be dropped. 
\end{remark}

In words, Assumption \ref{assumption-Pmin} ensures that the noise
density is not exceedingly lower than the average density $1/m$ at any
point. The reason why we introduce this assumption is two-fold.  To
begin with, this enables us to preclude the case where the entries of
$\bm{L}$---or equivalently, the log-likelihoods---are too wild.  For
instance, if $P_{0}\left(y\right)=0$ for some $y$, then
$\log P_0\left(y\right)=-\infty$, resulting in computational
instability.  The other reason is to simplify the analysis and
exposition slightly, making it easier for the readers.  We note,
however, that this assumption is not crucial and can be dropped by
means of a slight modification of the algorithm, which will be
detailed later.

Another assumption that we would like to introduce is more subtle:

\begin{assumption}
\label{assumption-KL-D}
$\mathsf{KL}_{\max}/\mathsf{KL}_{\min}$ is bounded, where
\begin{eqnarray}
  \mathsf{KL}_{\min} := \min_{1\leq l<m} \mathsf{KL}\left(P_{0} \hspace{0.3em}\|\hspace{0.3em} P_{l}\right)
  \quad\text{and}\quad  \mathsf{KL}_{\max} := \max_{1\leq l<m} \mathsf{KL}\left( P_{0}\hspace{0.3em}\|\hspace{0.3em}P_{l} \right). 
  \label{eq:Hel-KL-TV-max}
\end{eqnarray}
\end{assumption}

Roughly speaking, Assumption \ref{assumption-KL-D} states that the mutual distances of the $m$ possible output distributions $\{P_l\}_{1\leq l\leq m}$ lie within a reasonable dynamic range, so that one cannot find a pair of them that are considerably more separated than other pairs. 
Alternatively, it is understood that the variation of the log-likelihood
ratio, as we will show later, is often governed by the KL divergence
between the two corresponding distributions. From this point of view,
Assumption \ref{assumption-KL-D} tells us that there is no submatrix
of $\bm{L}$ that is significantly more volatile than the remaining
parts, which often leads to enhanced stability when computing the
power iteration.

With these assumptions in place, we are positioned to state our main
result. It is not hard to see that Theorem \ref{thm:ExactRecovery-outlier}
is an immediate consequence of the following theorem. 

\begin{theorem}
\label{thm:ExactRecovery-General}
Fix $m>0$, and
assume $p_{\mathrm{obs}}>c_{1}\log n/n$ and $\mu_{t}>c_{2}/\sigma_{2} (\bm{L})$
for some sufficiently large constants $c_{1},c_{2}>0$. Under Assumptions
\ref{assumption-Pmin}-\ref{assumption-KL-D}, there exist some absolute
constants $0<\rho,\nu<1$ such that with probability tending to one as $n$ scales,
the iterates of Algorithm \ref{alg:PP} with the input matrix (\ref{eq:Input-Matrix}) or (\ref{eq:debiased-MLE}) 
obey 
\begin{equation}
  \mcr(\bm{z}^{(t)},\bm{x}) \leq \nu\rho^{t},  \qquad \forall t \geq 0,
  \label{eq:error-contraction-general}
\end{equation}
provided that\footnote{Theorem \ref{thm:ExactRecovery-General} remains valid if we replace
4.01 by any other constant in (4,$\infty$). }
\begin{equation}
  \mathsf{KL}_{\min} \geq \frac{4.01 \log n}{ np_{\mathrm{obs}} }.
  \label{eq:IT-sufficient-KL}
\end{equation}
\end{theorem}
\begin{remark}
Alternatively, Theorem \ref{thm:ExactRecovery-General}
can be stated in terms of other divergence metrics like the squared
Hellinger distance $\mathsf{H}^{2}(\cdot,\cdot)$. Specifically, Theorem \ref{thm:ExactRecovery-General}
holds if the minimum squared Hellinger distance obeys 
\begin{equation}
  \mathsf{H}_{\min}^{2} := \min_{1\leq l<m} \mathsf{H}^{2}\left(P_{0},P_{l}\right) > \frac{1.01\log n}{np_{\mathrm{obs}}},
  \label{eq:IT-sufficient-H2}
\end{equation}
where $\mathsf{H}^{2}\left(P,Q\right):=\frac{1}{2}\sum_{y}(\sqrt{P(y)}-\sqrt{Q(y)})^{2}$.
We will see later in Lemma \ref{lem:KL-Var-Hel} that $\mathsf{KL}_{\min}\approx4\mathsf{H}_{\min}^{2}$,
which justifies the equivalence between (\ref{eq:IT-sufficient-KL})
and (\ref{eq:IT-sufficient-H2}). 
\end{remark}

The recovery condition (\ref{eq:IT-sufficient-KL}) is non-asymptotic, and takes the
form of a minimum KL divergence criterion. This is consistent with
the understanding that the hardness of exact recovery often arises
in differentiating minimally separated output distributions. Within
at most $O(\log n)$ projected power iterations, the PPM returns
an estimate with absolutely no error, as soon as the minimum KL divergence
exceeds some threshold. This threshold can be remarkably small when $p_{\mathrm{obs}}$
is large or, equivalently, when we have many pairwise measurements available.

Theorem \ref{thm:ExactRecovery-General} accommodates a broad class of noise models. Here we highlight
a few examples to illustrate its generality. To begin with, it is
self-evident that the random corruption model belongs to this class
with $\mathsf{KL}_{\max}/\mathsf{KL}_{\min}=1$. Beyond this
simple model, we list two important families which satisfy Assumption
\ref{assumption-KL-D} and which receive broad practical interest.
This list, however, is by no means exhaustive.

\begin{enumerate}

\item[(1)] A class of distributions that obey
\begin{equation}
  \mathsf{KL}_{\min}=\mathsf{KL}\left(P_{0}\hspace{0.3em}\|\hspace{0.3em}P_{1}\right)\quad\text{or}\quad\mathsf{KL}_{\min}=\mathsf{KL}\left(P_{0}\hspace{0.3em}\|\hspace{0.3em}P_{-1}\right).
  \label{eq:defn-worst-case}
\end{equation}
This says that the output distributions are the closest when the two
corresponding inputs are minimally separated. 

\item[(2)] A class of \emph{unimodal} distributions that satisfy
\begin{equation}
  P_{0}\left(0\right)\geq P_{0}\left(1\right)\geq\cdots\geq P_{0}\left(\left\lfloor m/2 \right\rfloor \right).
  \label{eq:defn-monotone}
\end{equation}
This says that the likelihood decays as the distance to the truth
increases. 

\end{enumerate}

\begin{lem}
\label{lem:two-class}
Fix $m>0$, and suppose Assumption \ref{assumption-Pmin} holds. Then the noise distribution
satisfying either (\ref{eq:defn-worst-case}) or (\ref{eq:defn-monotone})
obeys Assumption \ref{assumption-KL-D}. 
\end{lem}

\begin{proof} See Appendix \ref{sec:Proof-of-Lemma-two-class}.\end{proof}

\subsubsection{Why Algorithm \ref{alg:PP} works?}

We pause here to gain some insights about Algorithm \ref{alg:PP} and,
in particular, why the minimum KL divergence has emerged as a key
metric. Without loss of generality, we assume $x_{1}=\cdots=x_{n}=1$
to simplify the presentation.

Recall that Algorithm \ref{alg:PP} attempts to find the constrained
principal component of $\bm{L}$. To enable successful recovery, one
would naturally hope the structure of the data matrix $\bm{L}$ to
reveal much information about the truth. 
In the limit of large samples, it is helpful to start by looking at the mean of $\bm{L}_{i,j}$, which is given by
\begin{eqnarray}
  \mathbb{E}[\bm{L}_{i,j}]_{\alpha,\beta} 
  & = & p_{\mathrm{obs}}\mathbb{E}_{y\sim P_{0}}\big[\log P_{\alpha-\beta} (y) \big] \nonumber\\
   & = & p_{\mathrm{obs}}\mathbb{E}_{y\sim P_{0}}
  \left[ \log\frac{P_{\alpha-\beta} (y) }{P_{0} (y)} \right]  
   + p_{\mathrm{obs}} \mathbb{E}_{y\sim P_{0}} \big[ \log P_{0} (y) \big]  \nonumber \\
   & = & p_{\mathrm{obs}} \left[ -\mathsf{KL}_{\alpha-\beta}-\mathcal{H} (P_{0}) \right] 
\end{eqnarray}
%
 for any $i\neq j$ and $1\leq\alpha,\beta\leq m$; here and throughout,
$\mathcal{H}\left(P_{0}\right):=-\sum_{y}P_{0}\left(y\right)\log P_{0}\left(y\right)$
is the entropy functional, and
\begin{eqnarray}
  \mathsf{KL}_{l} & := & \mathsf{KL}\left(P_{0}\hspace{0.3em}\|\hspace{0.3em}P_{l}\right),\quad0\leq l<m.
  \label{eq:defn-KL-intuition}
\end{eqnarray}
We can thus write
\begin{equation}
  \mathbb{E}\left[\bm{L}\right] = p_{\mathrm{obs}}\left[
  \begin{array}{cccc}
  \bm{0} & \bm{K} & \cdots & \bm{K}\\
  \bm{K} & \bm{0} & \ddots & \vdots\\
  \vdots & \ddots & \bm{0} & \bm{K}\\
  \bm{K} & \cdots & \bm{K} & \bm{0}
  \end{array}\right]
  \label{eq:EL-general-block-intuition}
\end{equation}
with $\bm{K}\in\mathbb{R}^{m\times m}$ denoting a circulant matrix
\begin{equation}
  \bm{K} := \underset{:=\bm{K}^{0}}{\underbrace{\left[\begin{array}{cccc}
  -\mathsf{KL}_{0} & -\mathsf{KL}_{m-1} & \cdots & -\mathsf{KL}_{1}\\
  -\mathsf{KL}_{1} & -\mathsf{KL}_{0} & \ddots & -\mathsf{KL}_{2}\\
  \vdots & \ddots & \ddots & \vdots\\
  -\mathsf{KL}_{m-1} & \cdots & -\mathsf{KL}_{1} & -\mathsf{KL}_{0}
  \end{array}\right]}}-\mathcal{H}\left(P_{0}\right)\bm{1}\cdot\bm{1}^{\top}.
  \label{eq:defn-K-intuition}
\end{equation}
It is easy to see that the largest entries of $\bm{K}$ lie on the
main diagonal, due to the fact that
\[
  -\mathsf{KL}_{0}=0 \qquad \text{and} \qquad -\mathsf{KL}_{l} = -\mathsf{KL} (P_{0}\hspace{0.3em}\|\hspace{0.3em}P_{l}) < 0 
  \quad(1\leq l<m).
\]
Consequently, for any column of $\mathbb{E}[\bm{L}]$, knowledge of
the largest entries in each block reveals the relative positions across
all $\{x_{i}\}$. Take the 2nd column of $\mathbb{E}[\bm{L}]$ for
example: 
all but the first blocks of this column attain the maximum
values in their 2nd entries, telling us that $x_{2}=\cdots=x_{n}$.

Given the noisy nature of the acquired data, one would further need
to ensure that the true structure stands out from the noise. This
hinges upon understanding when $\bm{L}$ can serve as a reasonably
good proxy for $\mathbb{E}[\bm{L}]$ in the (projected) power iterations.
Since we are interested in identifying the largest entries, the signal
contained in each block---which is essentially the mean separation
between the largest and second largest entries---is of size 
\[
  p_{\mathrm{obs}}\min_{1\leq l<m}\mathbb{E}_{y\sim P_{0}}\left[\log P_{0}\left(y\right)-\log P_{l}\left(y\right)\right]
  ~=~  p_{\mathrm{obs}}\min_{1\leq l<m}\mathsf{KL}_{l} 
  ~=~ p_{\mathrm{obs}}\mathsf{KL}_{\min}.
\]
The total signal strength is thus given by $np_{\mathrm{obs}}\mathsf{KL}_{\min}$.
In addition, the variance in each measurement is bounded by
\begin{eqnarray*}
  \max_{1\leq l<m}{\mathsf{Var}}_{y\sim P_{0}}\left[\log P_{0}\left(y\right)-\log P_{l}\left(y\right)\right] 
  & \overset{(\text{a})}{\lesssim} & 
  \max_{1\leq l<m}\mathsf{KL}_{l} ~=~ \mathsf{KL}_{\max},
\end{eqnarray*}
where the inequality (a) will be demonstrated later in Lemma \ref{lem:KL-Var-Hel}.
From the semicircle law, the perturbation can be controlled by 
\begin{equation}
  \left\Vert \bm{L}-\mathbb{E}\left[\bm{L}\right] \right\Vert 
   = O\left( \sqrt{np_{\mathrm{obs}} \mathsf{KL}_{\max}} \right),
  \label{eq:semi-circle}
\end{equation}
where $\|\cdot\|$ is the spectral norm. This cannot exceed the size of the signal, namely,
\[
  np_{\mathrm{obs}}\mathsf{KL}_{\min}\gtrsim\sqrt{np_{\mathrm{obs}}\mathsf{KL}_{\max}}.
\]
This condition reduces to
\[
  \mathsf{KL}_{\min}\gtrsim\frac{1}{np_{\mathrm{obs}}}
\]
under Assumption \ref{assumption-KL-D}, which is consistent with
Theorem \ref{thm:ExactRecovery-General} up to some logarithmic factor.

\subsubsection{Optimality}

The preceding performance guarantee turns out to be information theoretically
optimal in the asymptotic regime. In fact, the KL divergence threshold given
in Theorem \ref{thm:ExactRecovery-General} is arbitrarily close to
the information limit, a level below which every procedure is bound
to fail in a minimax sense. We formalize this finding as a converse
result:

\begin{theorem}
\label{thm:lower-bound}
Fix $m>0$. Let $\left\{ P_{l,n}:\text{ }0\leq l<m\right\} _{n\geq1}$
be a sequence of probability measures supported on a finite set $\mathcal{Y}$,
where $\inf_{y\in\mathcal{Y},n,p}P_{l,n}\left(y\right)$ is bounded
away from 0. Suppose that there exists $1\leq l<m$ such that
\[
  \mathsf{KL}\left( P_{0,n} \hspace{0.3em}\|\hspace{0.3em} P_{l,n} \right)
  = \min_{1\leq j<m}\mathsf{KL}\left(P_{0,n}\hspace{0.3em}\|\hspace{0.3em}P_{j,n}\right) 
  := \mathsf{KL}_{\min,n}
\]
for all sufficiently large $n$,  and that $p_{\mathrm{obs}}>\frac{c_{0}\log n}{n}$ for some sufficiently
large constant $c_{0}>0$. If \footnote{Theorem \ref{thm:lower-bound} continues to hold if we replace 3.99
with any other constant between 0 and 4.} 
\begin{equation}
  \mathsf{KL}_{\min,n}<\frac{3.99\log n}{np_{\mathrm{obs}}},
  \label{eq:MinimaxBound}
\end{equation}
then the minimax probability of error
\begin{equation}
  \inf_{\hat{\bm{x}}}\max_{x_{i}\in[m],1\leq i\leq n}\mathbb{P}\left( \mcr \left(\hat{\bm{x}},\bm{x} \right) > 0  \mid \bm{x}\right) 
  \rightarrow 1,
  \qquad\text{as }n \rightarrow \infty,
\end{equation}
where the infimum is over all possible estimators and $\bm{x}$ is
the vector representation of $\left\{ x_{i}\right\}$ 
  as usual.\end{theorem}

\subsection{Extension: removing Assumption \ref{assumption-Pmin}}

We return to Assumption \ref{assumption-Pmin}. 
As mentioned before, an exceedingly small $P_{0}(y)$ might result
in unstable log-likelihoods, which suggests we regularize the data
before running the algorithm. To this end, one alternative is to introduce
a little more entropy to the samples so as to regularize the noise
density, namely, we add a small level of random noise to yield
\begin{equation}
  \tilde{y}_{i,j}\overset{\text{ind}.}{=}
  \begin{cases}
    y_{i,j}, & \text{with probability }1-\varsigma,\\
    \mathsf{Unif}\left(m\right),\quad & \text{else},
  \end{cases}
  \label{eq:tilde-y}
\end{equation}
for some appropriate small constant $\varsigma>0$. The distribution of
the new data $\tilde{y}_{i,j}$ given $x_{i}-x_{j}=l$ is thus given
by
\begin{equation}
  \tilde{P}_{l}\quad\leftarrow\quad(1-\varsigma)P_{l} + \varsigma \mathsf{Unif}\left(m\right),
  \label{eq:tilde-P}
\end{equation}
which effectively bumps $\min P_{0}\left(y\right)$ up to $\left(1- \varsigma \right)\min P_{0}\left(y\right)+ \varsigma /m$.
We then propose to run Algorithm \ref{alg:PP} using the new data
$\left\{ \tilde{y}_{i,j}\right\} $ and $\{ \tilde{P}_{l}\} $,
leading to the following performance guarantee. 

\begin{theorem}
\label{thm:ExactRecovery-General-2}
Take $\varsigma > 0$
to be some sufficiently small constant, and suppose that Algorithm
\ref{alg:PP} operates upon $\left\{ \tilde{y}_{i,j}\right\} $ and
$\{\tilde{P}_{l}\}$. Then Theorem \ref{thm:ExactRecovery-General}
holds without Assumption \ref{assumption-Pmin}. 
\end{theorem}
\begin{proof} See Appendix \ref{sec:proof-thm:ExactRecovery-General-2}. \end{proof}


\subsection{Extension: large-$m$ case}\label{sub:large-m-extension}

So far our study has focused on the case where the alphabet size $m$
does not scale with $n$. There are, however, no shortage of situations
where $m$ is so large that it cannot be treated as a fixed constant.
The encouraging news is that Algorithm \ref{alg:PP} appears surprisingly
competitive for the large-$m$ case as well. Once again, we begin with the
random corruption model, and our analysis developed for fixed $m$
immediately applies here. 

\begin{theorem}
\label{thm:ExactRecovery-outlier-large-m}
Suppose
that $m\gtrsim\log n$, $m=O(\mathrm{poly}(n))$, and $p_{\mathrm{obs}}\geq c_{1}\log^{2}n/n$
for some sufficiently large constant $c_{1}>0$. Then Theorem \ref{thm:ExactRecovery-outlier}
continues to hold with probability at least 0.99, as long as (\ref{eq:IT-sufficient-dense})
is replaced by
\begin{equation}
  \pi_{0}>\frac{c_{3}}{\sqrt{np_{\mathrm{obs}}}}
  \label{eq:IT-sufficient-dense-large-m}
\end{equation}
for some universal constant $c_{3}>0$. Here, 0.99 is arbitrary and
can be replaced by any constant between 0 and 1.
\end{theorem}

The main message of Theorem \ref{thm:ExactRecovery-outlier-large-m}
is that the error correction capability of the proposed method improves
as the number $n$ of unknowns grows. The quantitative bound (\ref{eq:IT-sufficient-dense-large-m})
implies successful recovery even
when an overwhelming fraction of the measurements are corrupted. 
Notably, when $m$ is exceedingly large, Theorem \ref{thm:ExactRecovery-outlier-large-m} might shed light on the continuous joint alignment problem. 
In particular, there are two cases worth emphasizing:

\begin{itemize}
\item When $m \gg n$ (e.g.~$m \gtrsim n^{10}$), the random corruption model converges to the following continuous spike model as $n$ scales:
\begin{align}
  x_i \in [0,1), ~ 1\leq i\leq n  \qquad \text{and} \qquad   
  y_{i,j} = \begin{cases}  
	x_i - x_j~\mathrm{mod}~1, \quad  &\text{with prob. } \pi_0, \\ 
	\mathsf{Unif}(0,1), 	  \quad  &\text{else},
  \end{cases}  \quad (i,j) \in \Omega. 
  \label{eq:continuous-RCM}
\end{align}  
This coincides with the setting studied in \cite{singer2011angular,wang2013exact} over the orthogonal group $\mathsf{SO}(2)$, under the name of {\em  synchronization} \cite{liu2016estimation,boumal2016nonconvex,bandeira2016low}. It has been shown that the leading eigenvector of a certain  data matrix becomes positively correlated with the  truth as long as $\pi_0 > 1/\sqrt{n}$ when $p_{\mathrm{obs}}=1$ \cite{singer2011angular}. In addition, a generalized power method---which is equivalent to projected gradient descent---provably converges to the solution of the nonconvex least-squares estimation, as long as the size of the noise is below some threshold \cite{liu2016estimation,boumal2016nonconvex}. 
 When it comes to exact recovery, Wang et al.~prove that semidefinite relaxation succeeds as long as $\pi_0 > 0.457$ \cite[Theorem 4.1]{wang2013exact}, a constant threshold irrespective of $n$. In contrast, the exact recovery performance of our approach---which operates over a lifted discrete space rather than $\mathsf{SO}(2)$---improves with $n$, allowing $\pi_0$ to be arbitrarily small when $n$ is sufficiently large. 
On the other hand, the model (\ref{eq:continuous-RCM})  is reminiscent of the more general robust PCA problem \cite{CanLiMaWri09,chandrasekaran2011rank}, which consists in recovering a
low-rank matrix when a fraction of observed entries are corrupted. We have learned from the literature \cite{ganesh2010dense,chen2013low} that 
perfect reconstruction is feasible and tractable even though a dominant
portion of the observed entries may suffer from random corruption, 
which is consistent with our finding in Theorem \ref{thm:ExactRecovery-outlier-large-m}.

\item In the preceding spike model, the probability density of each measurement experiences an impulse around the truth. In a variety of realistic scenarios, however, the noise density might be more smooth rather than being spiky. Such smoothness conditions can be modeled by enforcing $P_0(z) \asymp P_0(z+1)$ for all $z$, so as to rule out any sharp jump.  To satisfy this condition, we can at most take $m \asymp \sqrt{n p_{\mathrm{obs}}}$ in view of Theorem \ref{thm:ExactRecovery-outlier-large-m}. 
In some sense, this uncovers the ``resolution'' of our estimator under the ``smooth'' noise model: if we constrain the input domain to be the unit interval 
by letting 
$x_i = 1, \cdots, m$ represent the grid points $0, \frac{1}{m}, \cdots, \frac{m-1}{m}$, respectively, then the PPM can recover each variable up to a resolution of 
\begin{align}
  \frac{1}{m} \asymp  \frac{1}{ \sqrt{n p_{\mathrm{obs}}}} .
\end{align}
\end{itemize}
Notably, the discrete random corruption model has been investigated in prior literature \cite{huang2013consistent,chen2014near},
with the best theoretical support derived for convex programming. Specifically, it has been shown by \cite{chen2014near} that convex relaxation is guaranteed to work as soon as $\pi_{0}\gtrsim\frac{\log^{2}n}{\sqrt{np_{\mathrm{obs}}}}$. In comparison, this is more stringent than the recovery condition (\ref{eq:IT-sufficient-dense-large-m}) we develop for the PPM  by  some logarithmic factor. 
Furthermore, Theorem \ref{thm:ExactRecovery-outlier-large-m} is an immediate consequence
of a more general result: 

\begin{theorem}
\label{thm:ExactRecovery-General-large-m}
Assume
that $m\gtrsim\log n$, $m=O(\mathrm{poly}(n))$, and $p_{\mathrm{obs}}>c_{1}\log^{5}n/n$
for some sufficiently large constant $c_{1}>0$. Suppose that $\bm{L}$ is replaced by $\bm{L}^{\mathrm{debias}}$ when computing the initial guess $\bm{z}^{(0)}$ and $\mu \geq c_{3} \sqrt{m} / \sigma_{m}\left(\bm{L}^{\mathrm{debias}} \right)$ for some sufficiently large constant $c_{3}>0$. Then Theorem \ref{thm:ExactRecovery-General}
continues to hold with probability exceeding 0.99, provided that (\ref{eq:IT-sufficient-KL})
is replaced by
\begin{equation}
  \frac{\mathsf{KL}_{\min}^{2}}{\max_{1\leq l<m}\left\Vert \log\frac{P_{0}}{P_{l}}\right\Vert _{1}^{2}}\geq\frac{c_{3}}{np_{\mathrm{obs}}}
  \label{eq:SNR-condition-general-m}
\end{equation}
for some universal constant $c_{3}>0$. Here, $\left\Vert \log\frac{P_{0}}{P_{l}}\right\Vert _{1}:=\sum_{y}\left|\log\frac{P_{0}(y)}{P_{l}(y)}\right|$,
and 0.99 can be replaced by any constant in between 0 and 1. 
\end{theorem}

Some brief interpretations of (\ref{eq:SNR-condition-general-m}) are in order.
As discussed before, the quantity $\mathsf{KL}_{\min}$ represents
the strength of the signal. In contrast, the term $\max_{l}\left\Vert \log\frac{P_{0}}{P_{l}}\right\Vert _{1}^{2}$
controls the variability of each block of the data matrix $\bm{L}_{i,j}$;
to be more precise, 
\[
  \mathbb{E}\left[\left\Vert \bm{L}_{i,j}-\mathbb{E}\left[\bm{L}_{i,j}\right]\right\Vert ^{2}\right]
  \lesssim \max_{1\leq l<m}\left\Vert \log\frac{P_{0}}{P_{l}}\right\Vert _{1}^{2}
\]
as we will demonstrate in the proof. Thus, the left-hand side of
(\ref{eq:SNR-condition-general-m}) can be regarded as the signal-to-noise-ratio (SNR)
experienced in each block $\bm{L}_{i,j}$. The recovery criterion is
thus in terms of a lower threshold on the SNR, which can be vanishingly
small in the regime considered in Theorem
\ref{thm:ExactRecovery-General-large-m}.
We note that the general alignment problem has been studied in \cite{bandeira2014multireference} as well, although the focus therein is to show the stability of semidefinite relaxation in the presence of random vertex noise.

We caution, however, that the performance guarantees presented in
this subsection are in general not information-theoretically optimal.
For instance, it has been shown in \cite{chen2015information} that the
MLE succeeds as long as 
\[
  \pi_{0}\gtrsim\sqrt{\frac{\log n}{mnp_{\mathrm{obs}}}}
\]
in the regime where $m\ll np_{\mathrm{obs}}/\log n$; that is, the
performance of the MLE improves as $m$ increases. It is noteworthy that none
of the polynomial algorithms proposed in prior works achieves optimal
scaling in $m$. It remains to be seen whether this arises
due to some drawback of the algorithms, or due to the existence of some inherent information-computation
gap.

\section{Numerical experiments\label{sec:Numerics}}

This section examines the empirical performance of the projected power method on both synthetic instances and real image data.  While the statistical assumptions (i.e.~the i.i.d.~noise model) underlying our theory typically do not hold in practical applications (e.g.~shape alignment and graph matching),  our numerical experiments show that the PPM developed based on our statistical models enjoy favorable performances when applied to real datasets.

\subsection{Synthetic experiments\label{sec:synthetic}}

To begin with, we conduct a series of Monte Carlo trials for various
problem sizes under the random corruption model (\ref{eq:random-outlier}).
Specifically, we vary the number $n$ of unknowns, the input corruption rate $1-\pi_{0}$, and the alphabet size
$m$, with the observation
rate set to be $p_{\mathrm{obs}}=1$ throughout. For each $(n,\pi_{0},m)$
tuple, 20 Monte Carlo trials are conducted. In each trial, we
draw each $x_{i}$ uniformly at random over $[m]$, generate a set of measurements
$\left\{ y_{i,j}\right\} _{1\leq i,j\leq n}$ according to (\ref{eq:random-outlier}),
and record the misclassification rate 
$\mcr(\hat{\bm x}, \bm{x})$ 
of Algorithm \ref{alg:PP}. The mean empirical misclassification rate is then
calculated by averaging over 20 Monte Carlo trials. 

Fig.~\ref{fig:numerics-RCM} depicts the mean empirical misclassification rate when
$m=2,10,20$, and accounts for two choices of the scaling factors:
(1) $\mu_{t}\equiv 10/\sigma_{2}(\bm{L})$, and (2) $\mu_{t}\equiv\infty$.
In particular, the solid lines locate the asymptotic phase transitions 
for exact recovery predicted by our theory. In all cases, the empirical
phase transition curves come closer to the analytical prediction
as the problem size $n$ increases.  

Another noise model that we have studied numerically is a modified
Gaussian model. Specifically, we set $m=5,9,15$, and the random noise
$\eta_{i,j}$ is generated in such a way that 
\begin{equation}
  \mathbb{P}\left\{ \eta_{i,j}=z\right\} ~\propto~ \exp\left(-\frac{z^{2}}{2\sigma^{2}}\right),
  \qquad  -\frac{m-1}{2} \leq z \leq \frac{m-1}{2},
  \label{eq:modified-Gaussian}
\end{equation}
where $\sigma$ controls the flatness of the noise density. We vary the parameters $\left(\sigma,n,m\right)$, take $p_{\mathrm{obs}}=1$, and experiment on two choices of scaling factors $\mu_t \equiv 20 / \sigma_m(\bm{L})$ and $\mu_t \equiv \infty$. The mean misclassification rate
of the PPM is reported in Fig.~\ref{fig:numerics-Gaussian}, where the empirical phase transition matches the theory very well.

\begin{figure}
\begin{tabular}{ccc}
\includegraphics[width=0.33\textwidth]{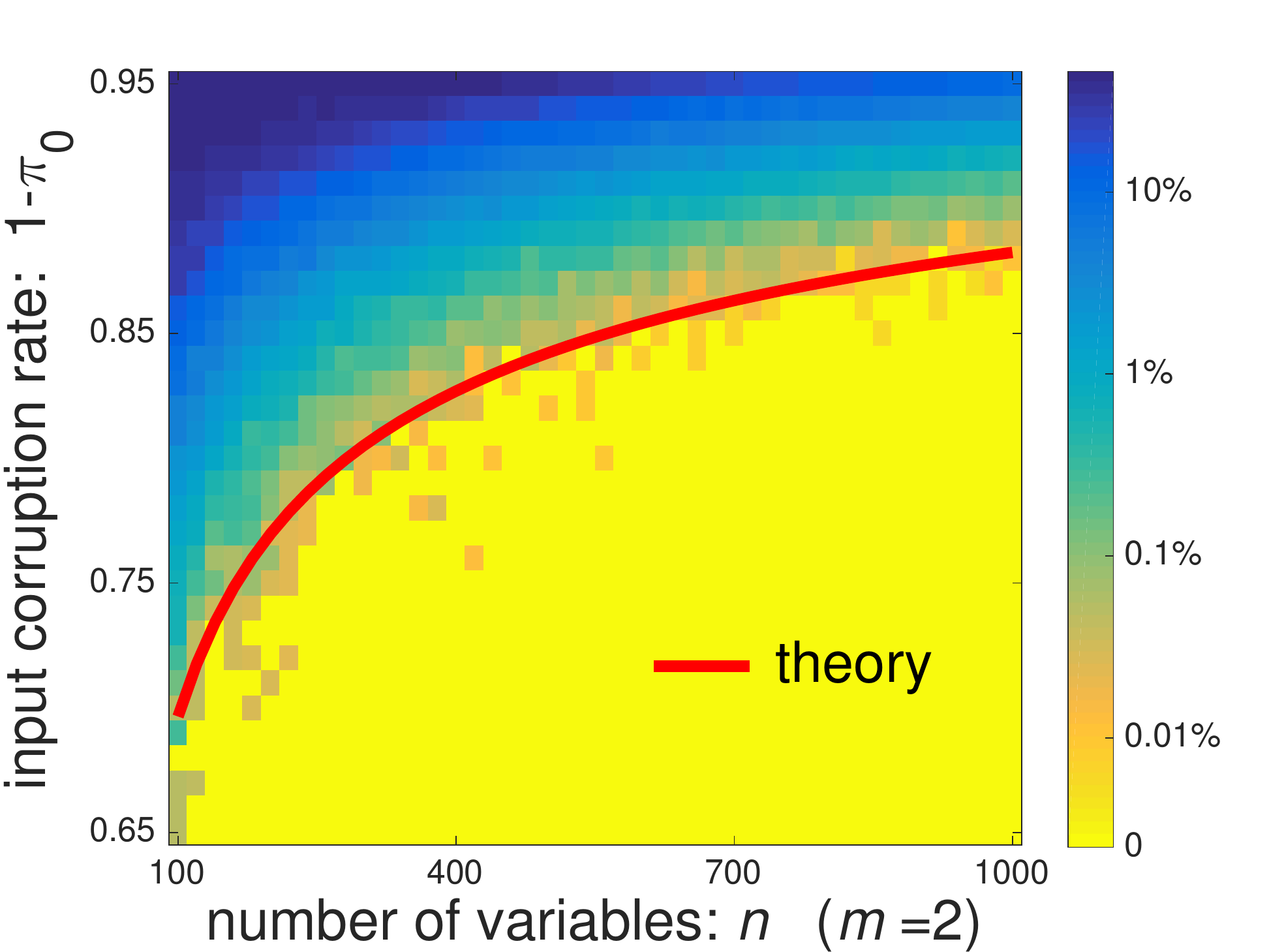} & \includegraphics[width=0.33\textwidth]{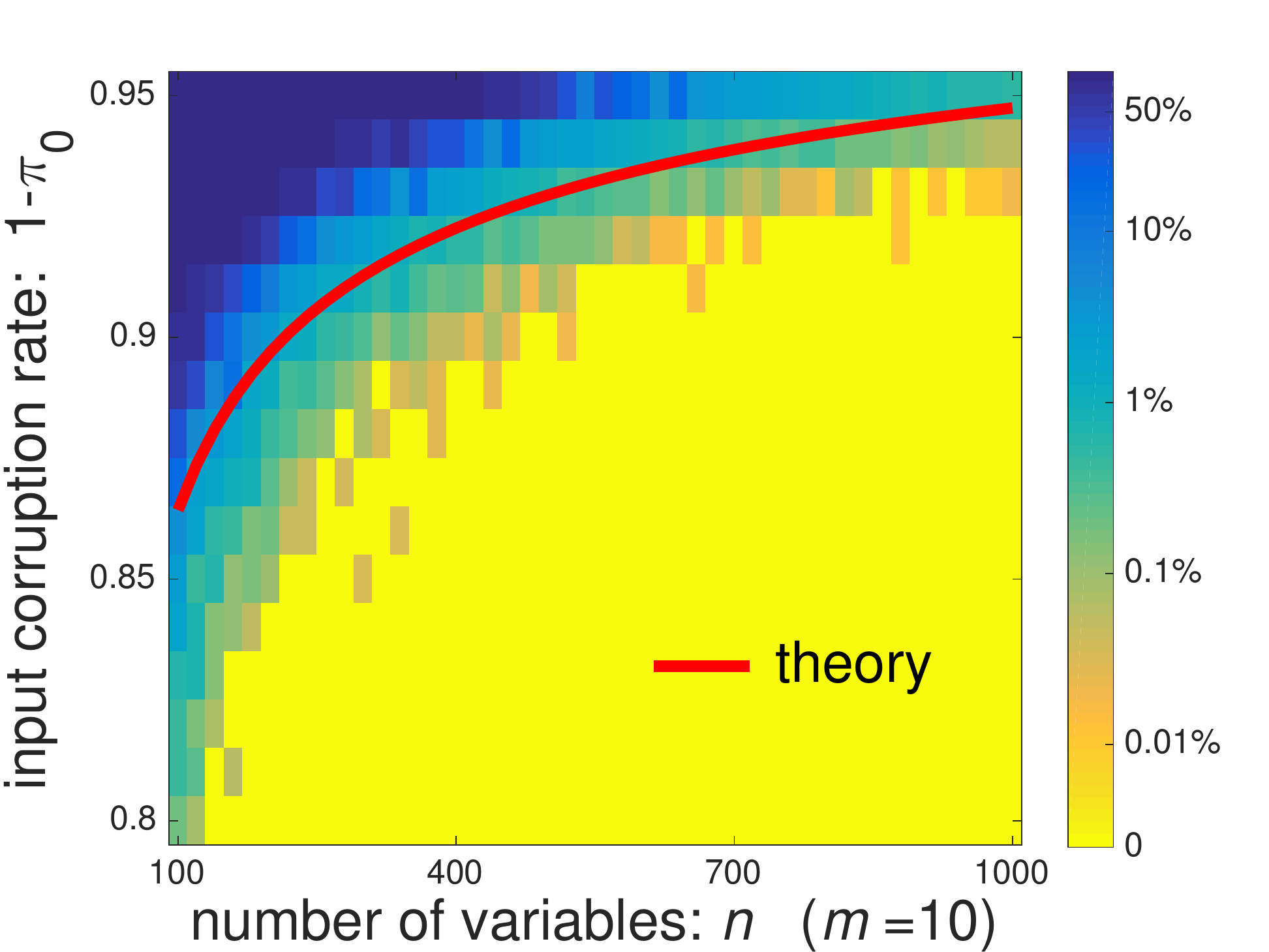} & \includegraphics[width=0.33\textwidth]{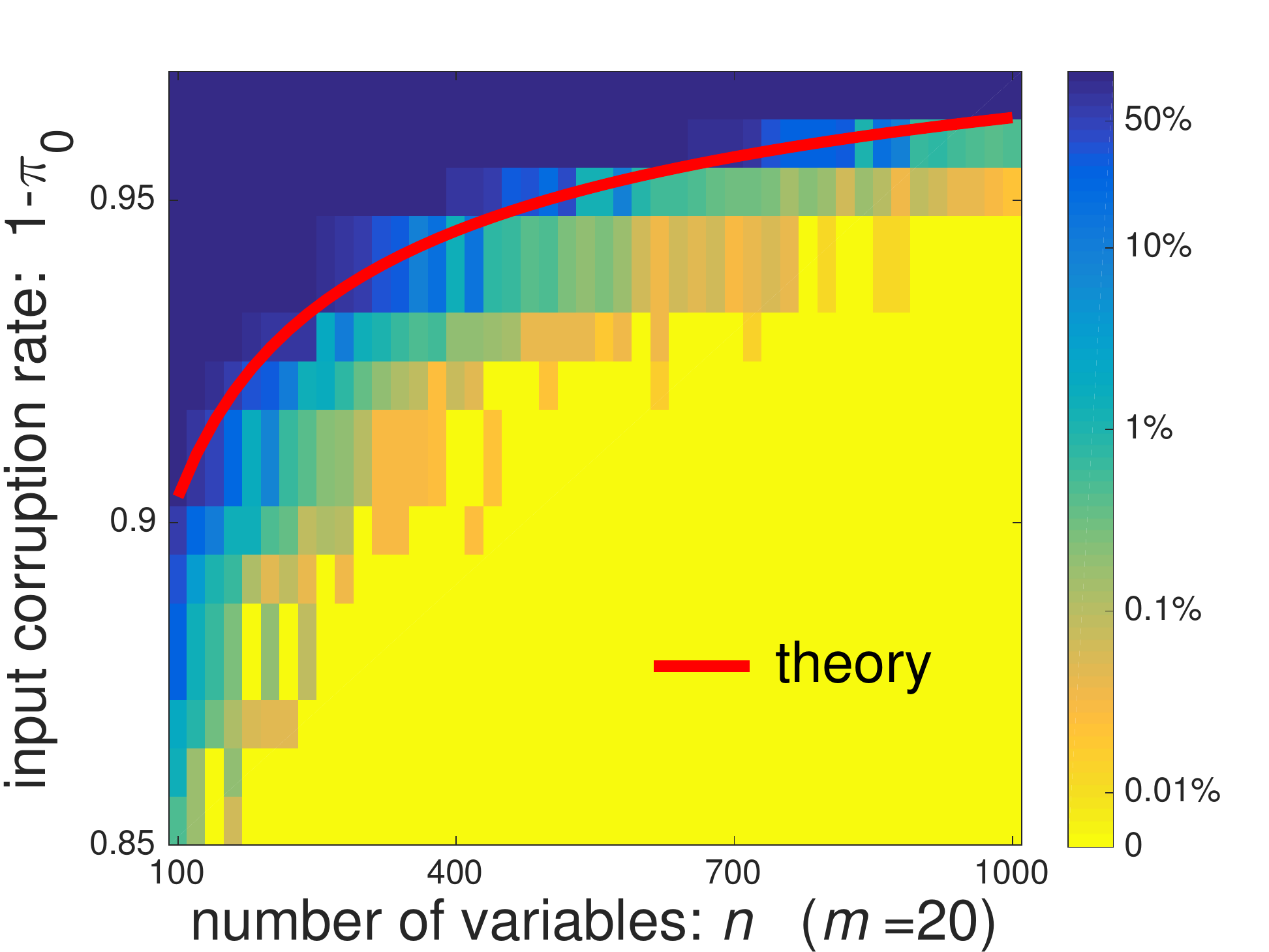}
\tabularnewline
$m=2,~\mu_{t}=10/\sigma_{2}(\bm{L})$ & $m=10,~\mu_{t}=10/\sigma_{2}(\bm{L})$ & $m=20,~\mu_{t}=10/\sigma_{2}(\bm{L})$
\tabularnewline
\includegraphics[width=0.33\textwidth]{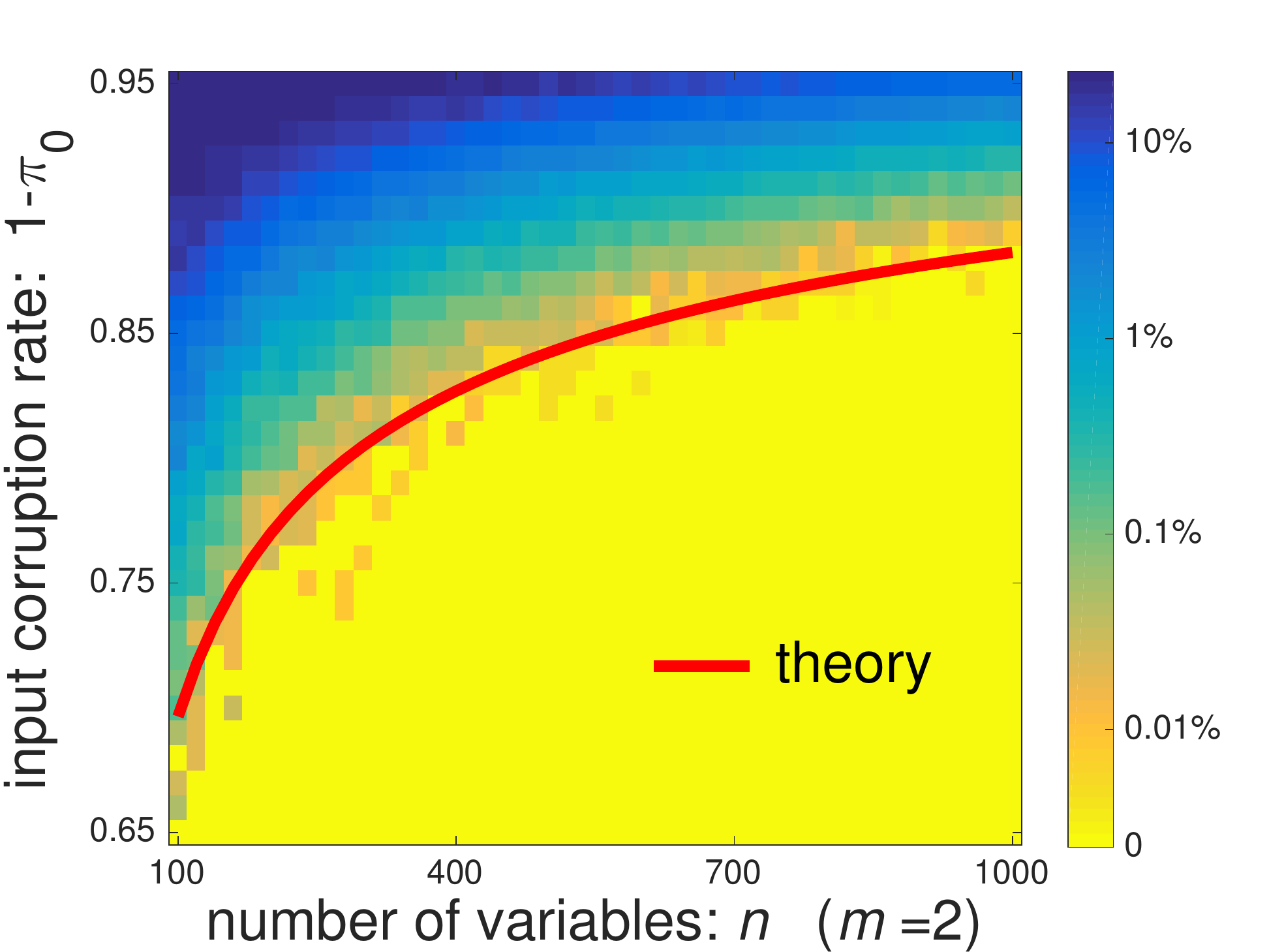} & \includegraphics[width=0.33\textwidth]{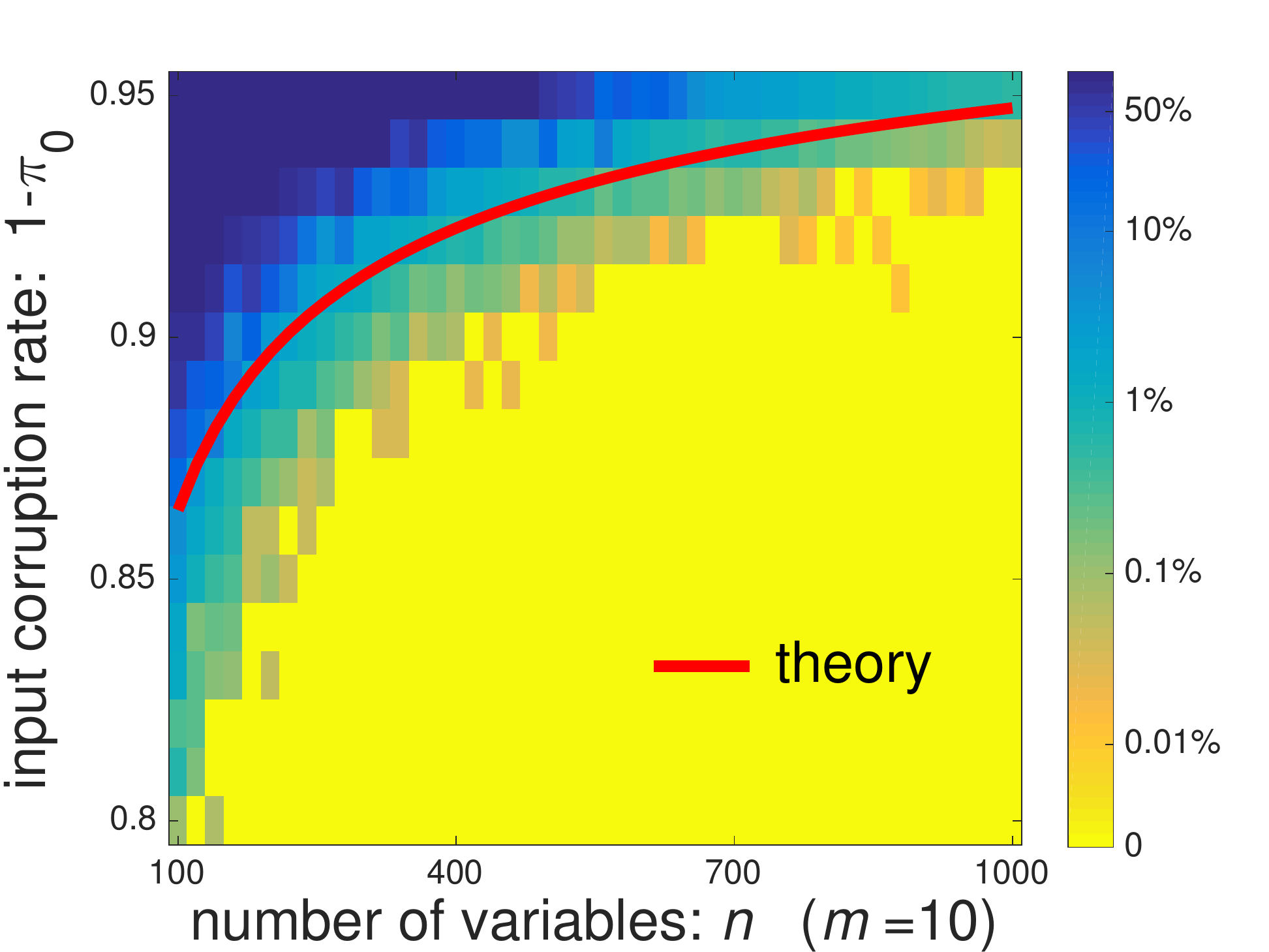} & \includegraphics[width=0.33\textwidth]{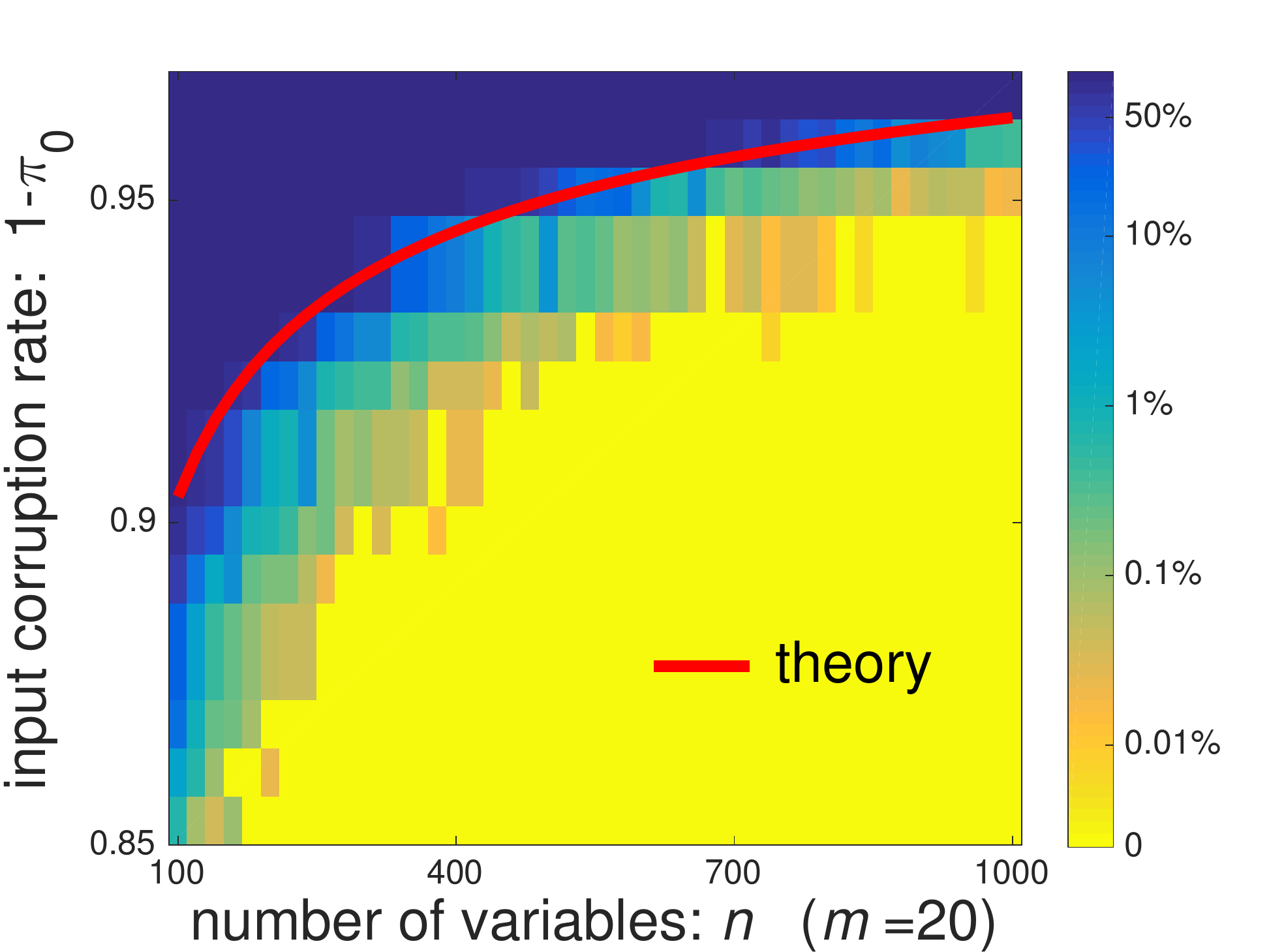}
\tabularnewline
$m=2,~\mu_{t}=\infty$ & $m=10,~\mu_{t}=\infty$ & $m=20,~\mu_{t}=\infty$
\tabularnewline
\end{tabular}
\caption{The empirical mean misclassification rate of Algorithm \ref{alg:PP} under the random corruption model.\label{fig:numerics-RCM}}
\end{figure}

\begin{figure}
\begin{tabular}{ccc}
\includegraphics[width=0.33\textwidth]{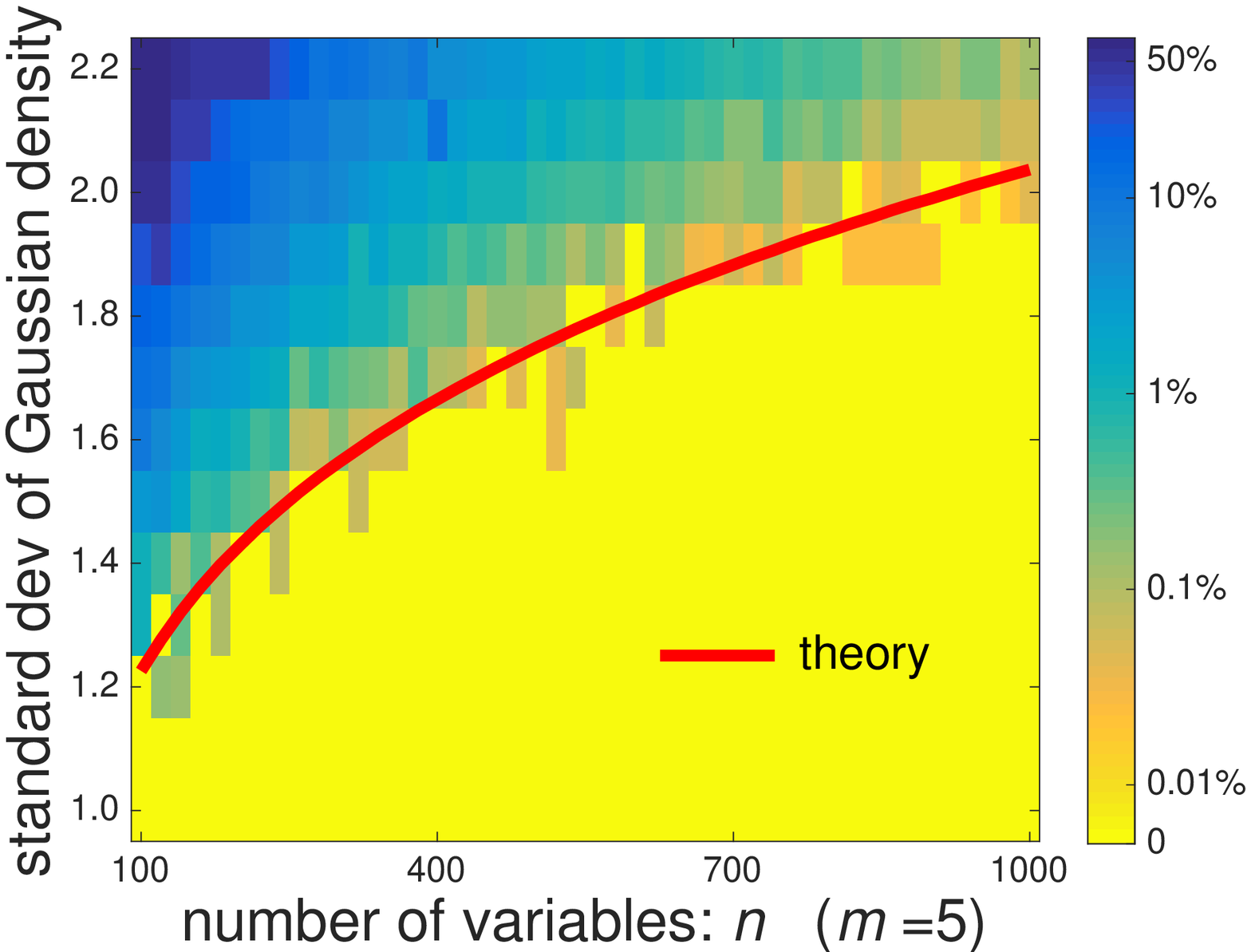} & \includegraphics[width=0.33\textwidth]{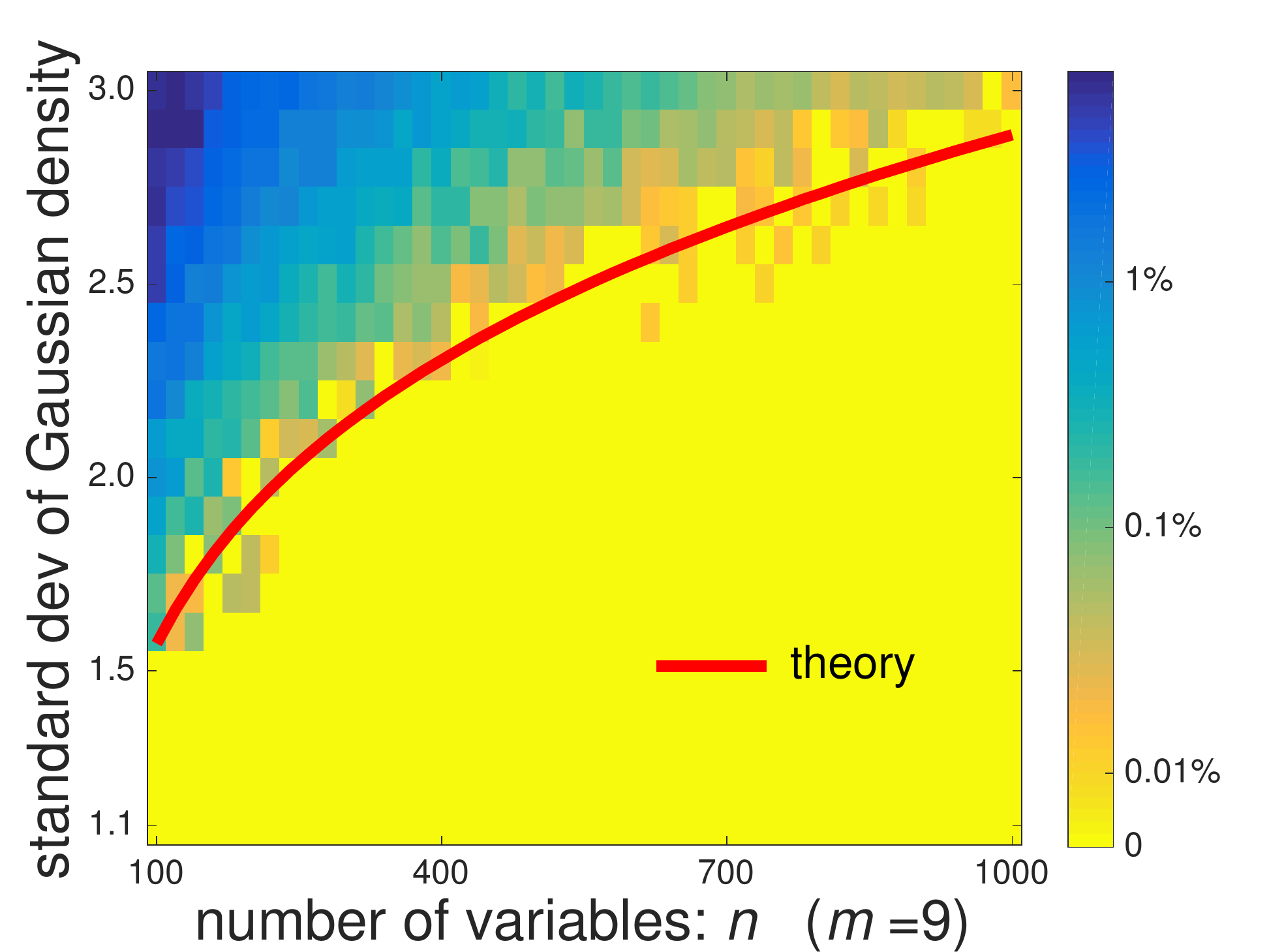} & \includegraphics[width=0.33\textwidth]{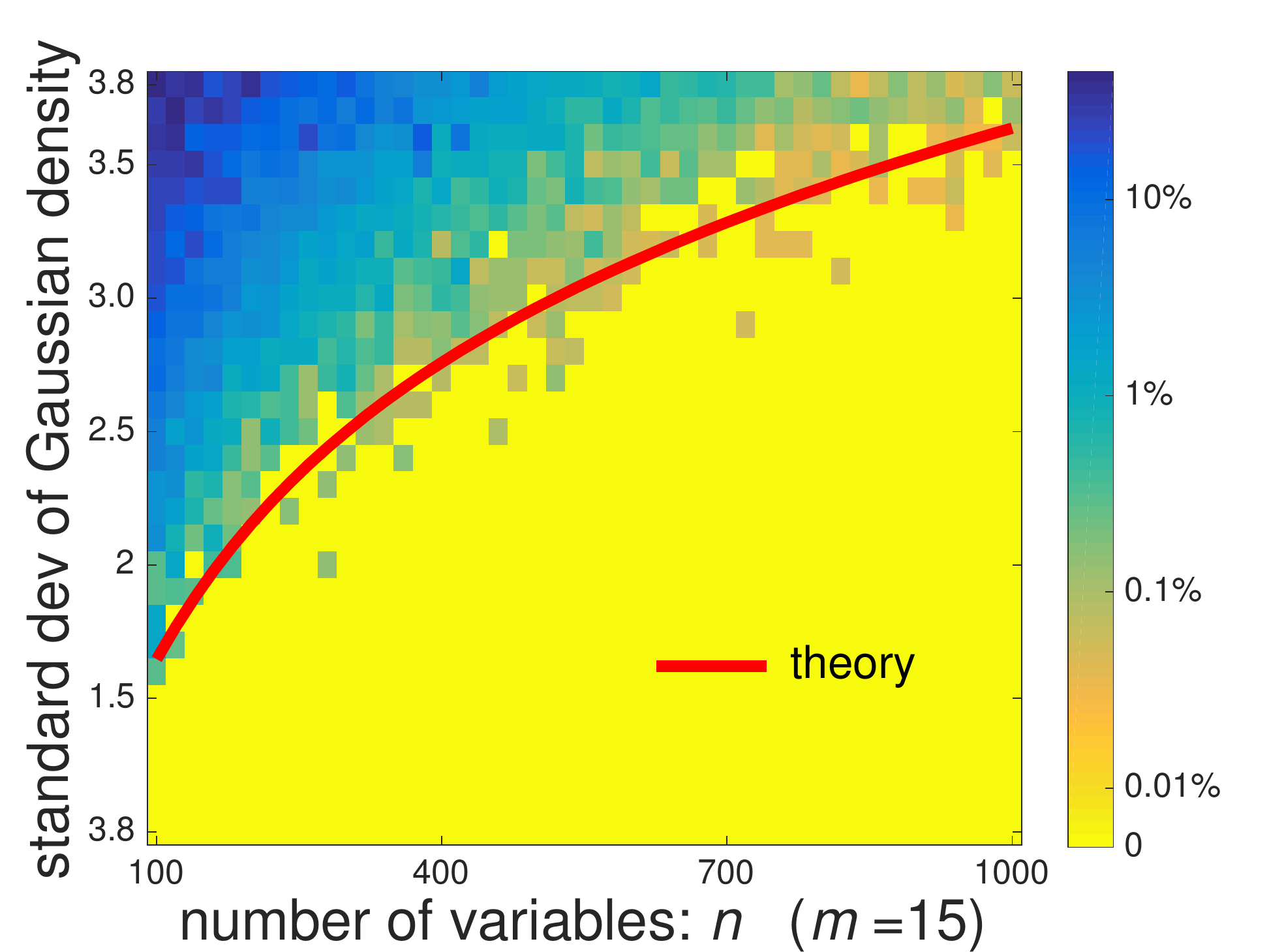}\tabularnewline
$m=5,~\mu_{t}=20/\sigma_{m}(\bm{L})$ & $m=9,~\mu_{t}=20/\sigma_{m}(\bm{L})$ & $m=15,~\mu_{t}=20/\sigma_{m}(\bm{L})$\tabularnewline
\includegraphics[width=0.33\textwidth]{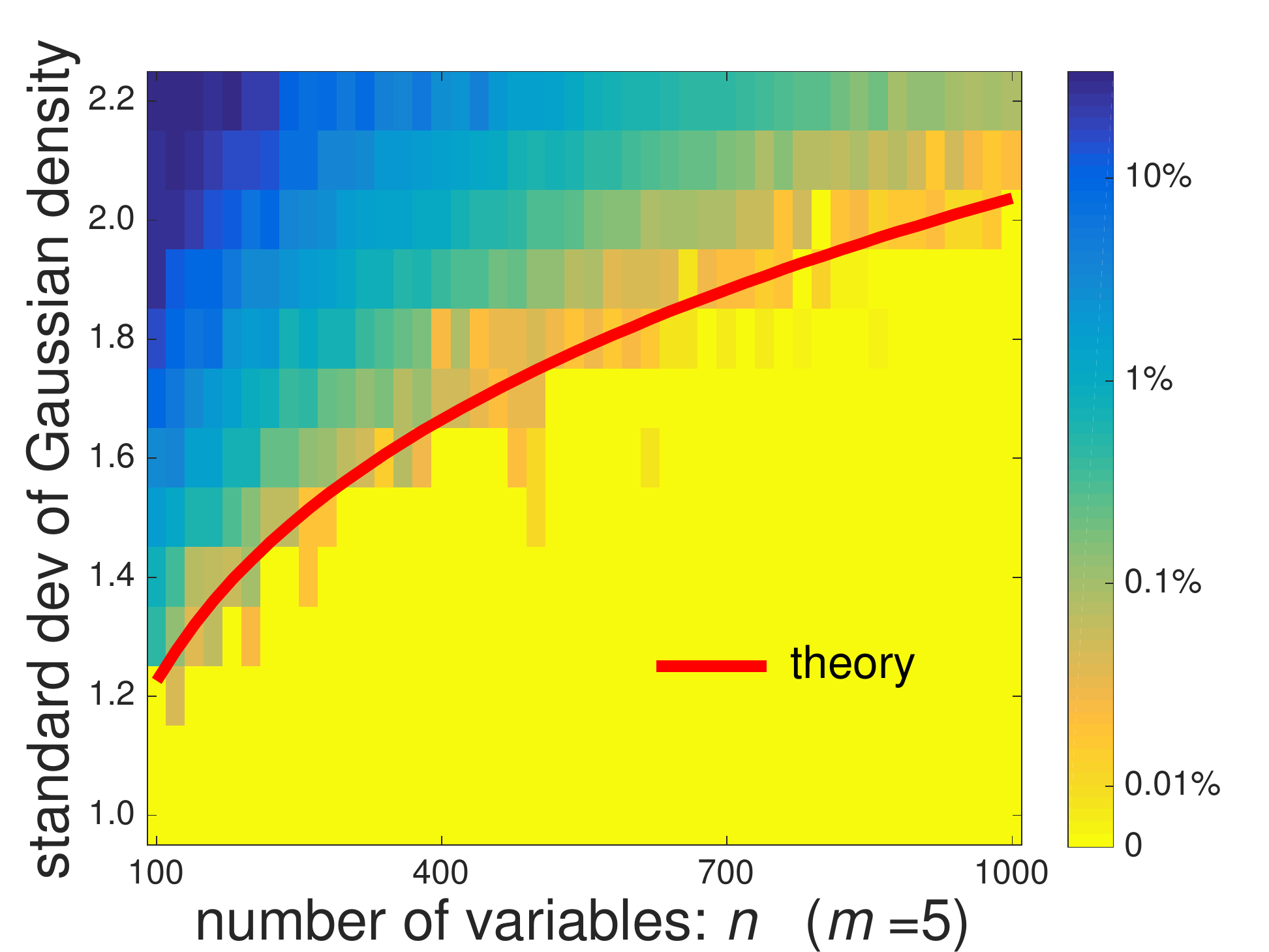} & \includegraphics[width=0.33\textwidth]{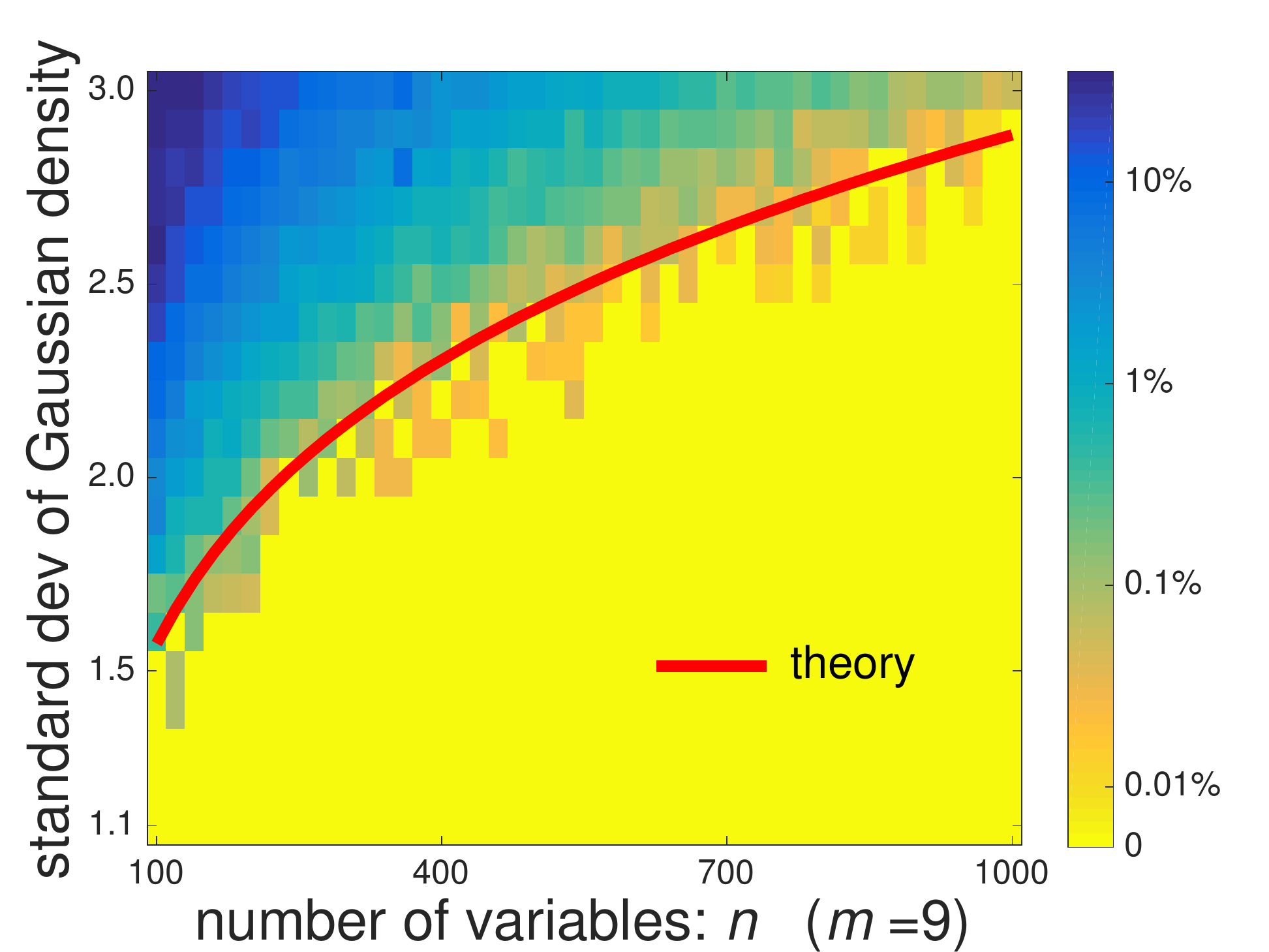} & \includegraphics[width=0.33\textwidth]{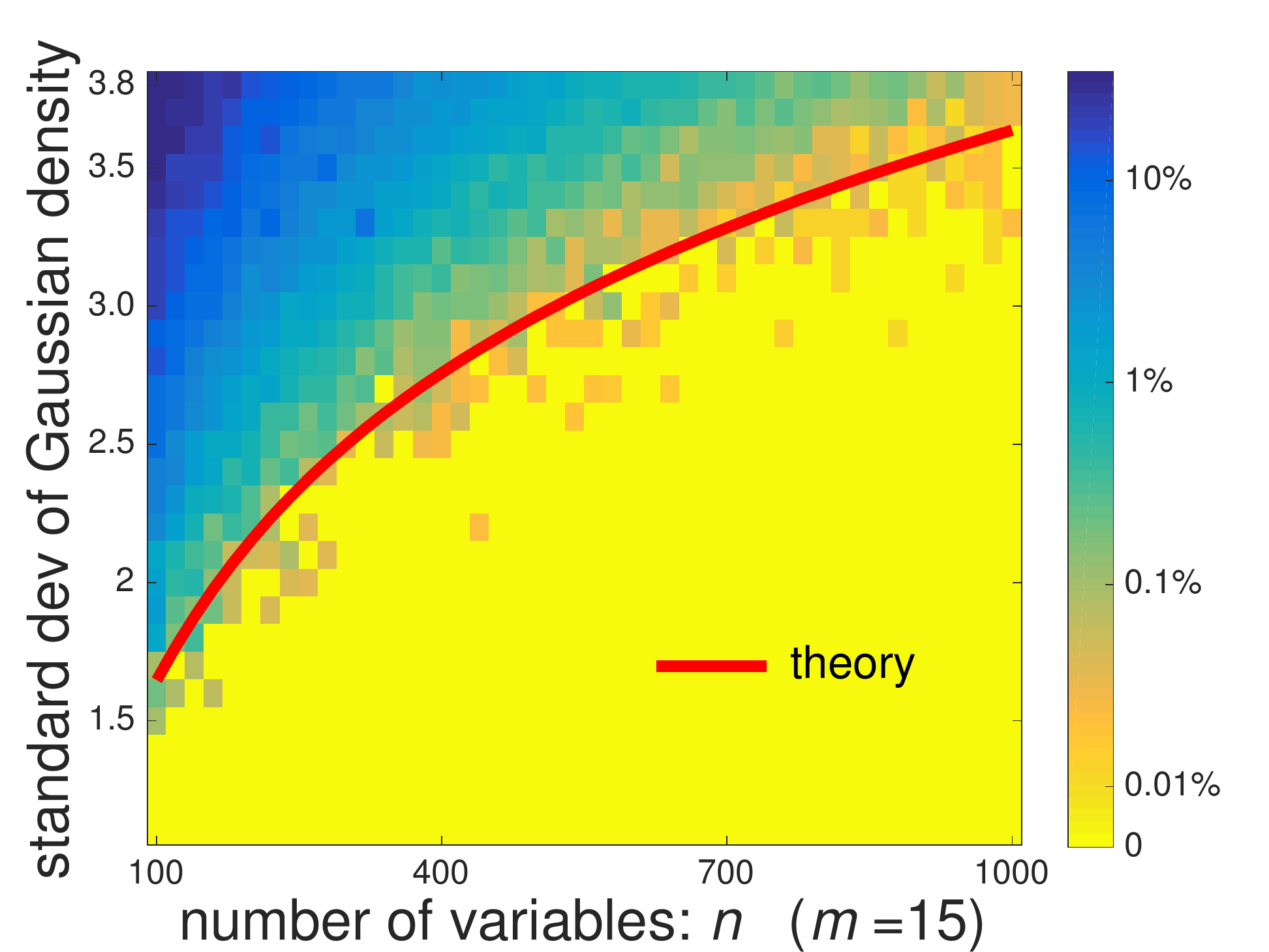}\tabularnewline
$m=5,~\mu_{t}=\infty$ & $m=9,~\mu_{t}=\infty$ & $m=15,~\mu_{t}=\infty$\tabularnewline
\end{tabular}
\caption{The empirical mean misclassification rate of Algorithm \ref{alg:PP} under the modified Gaussian model.\label{fig:numerics-Gaussian}
}
\end{figure}

\begin{figure}
\centering%
\begin{tabular}{cc}
\includegraphics[width=0.3\textwidth]{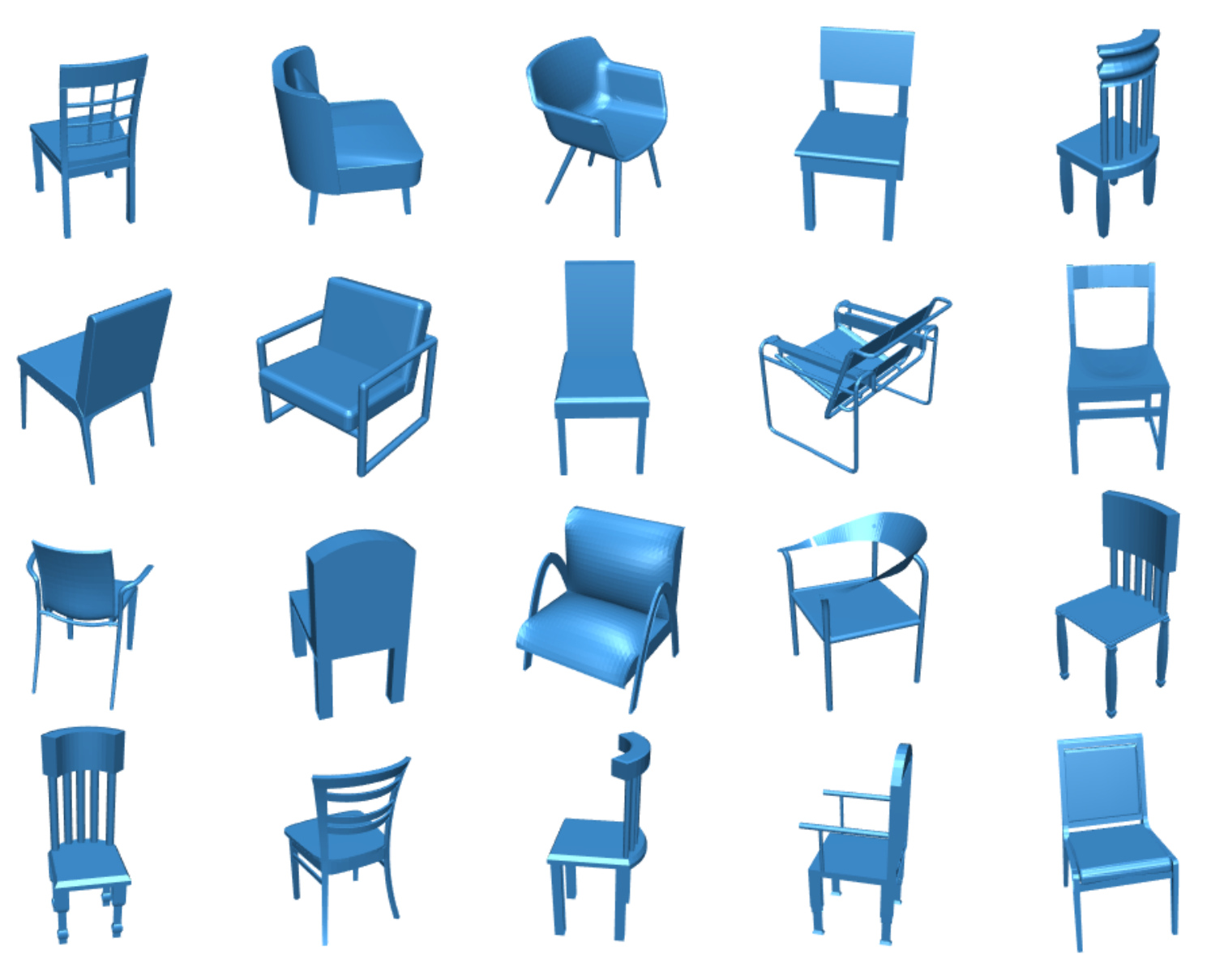} \qquad & \qquad  \includegraphics[width=0.3\textwidth]{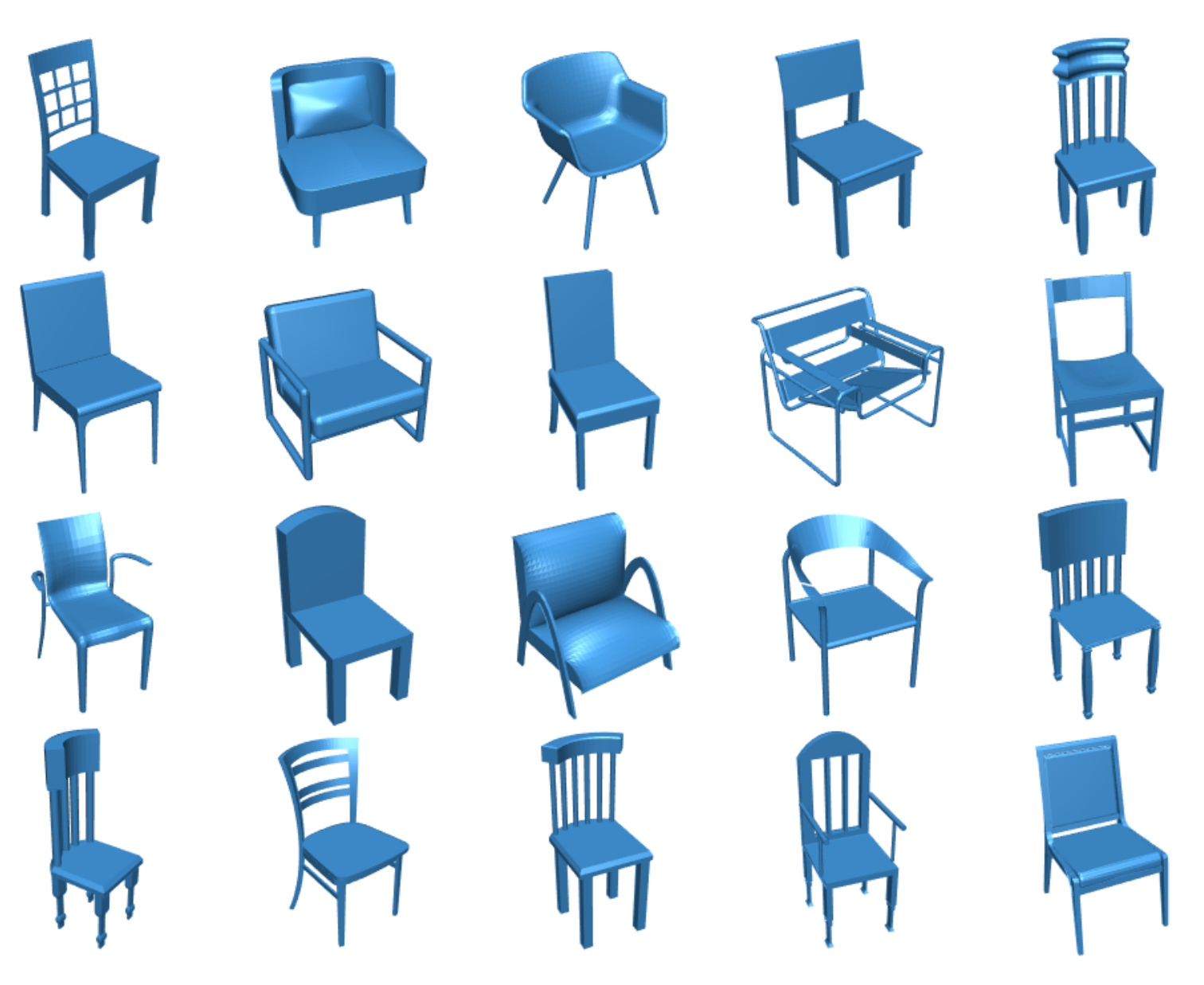} ~
\tabularnewline
\end{tabular}
\caption{The performance of the PPM on a Chair dataset of $n=50$ shapes: (left) the first 20 input shapes; (right) the first 20 shapes after alignment. \label{fig:shapes-chairs}}
\end{figure}

\begin{figure}
\centering%
\begin{tabular}{cccc}
\includegraphics[width=0.3\textwidth]{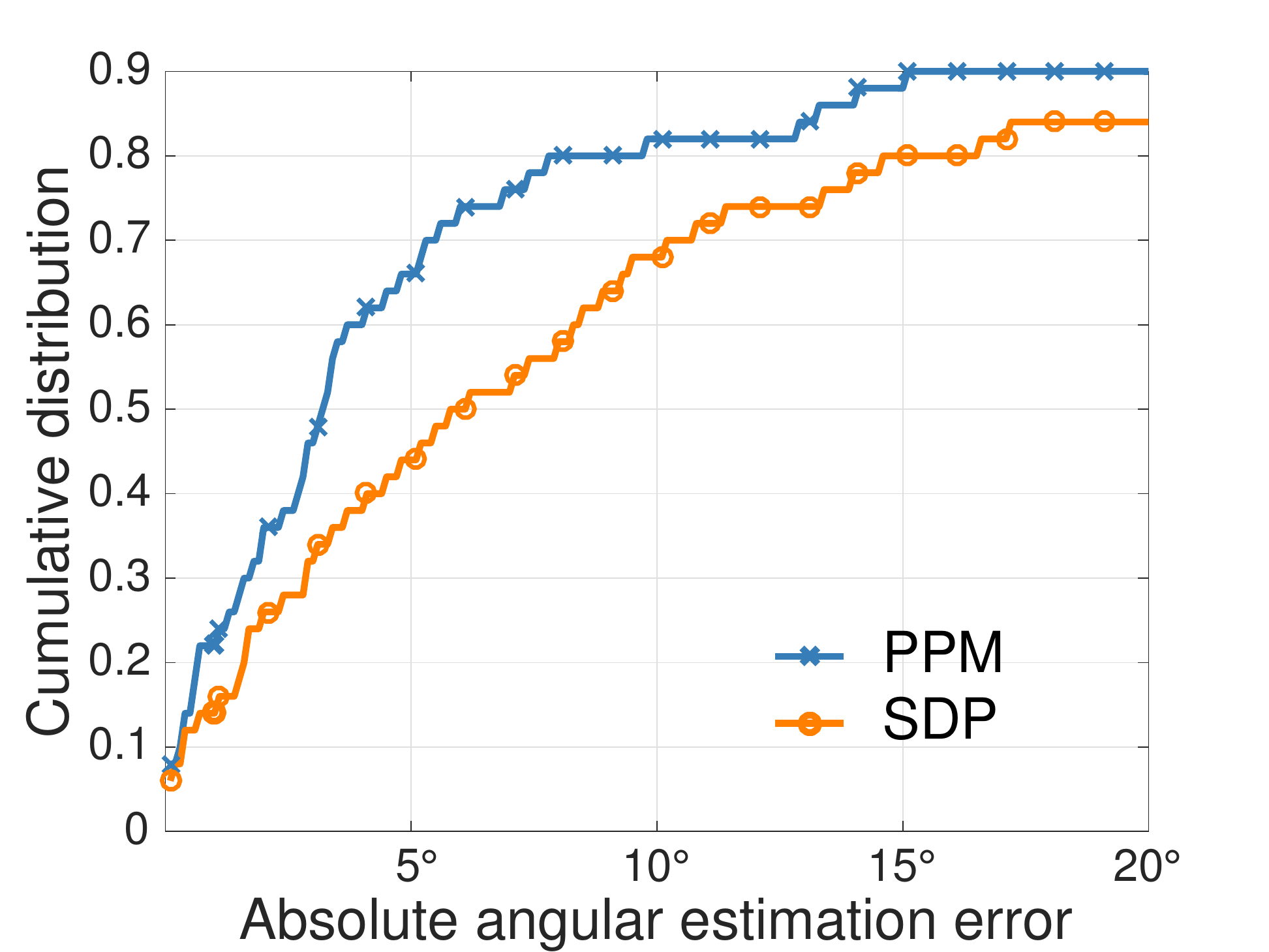} ~~ & \includegraphics[width=0.3\textwidth]{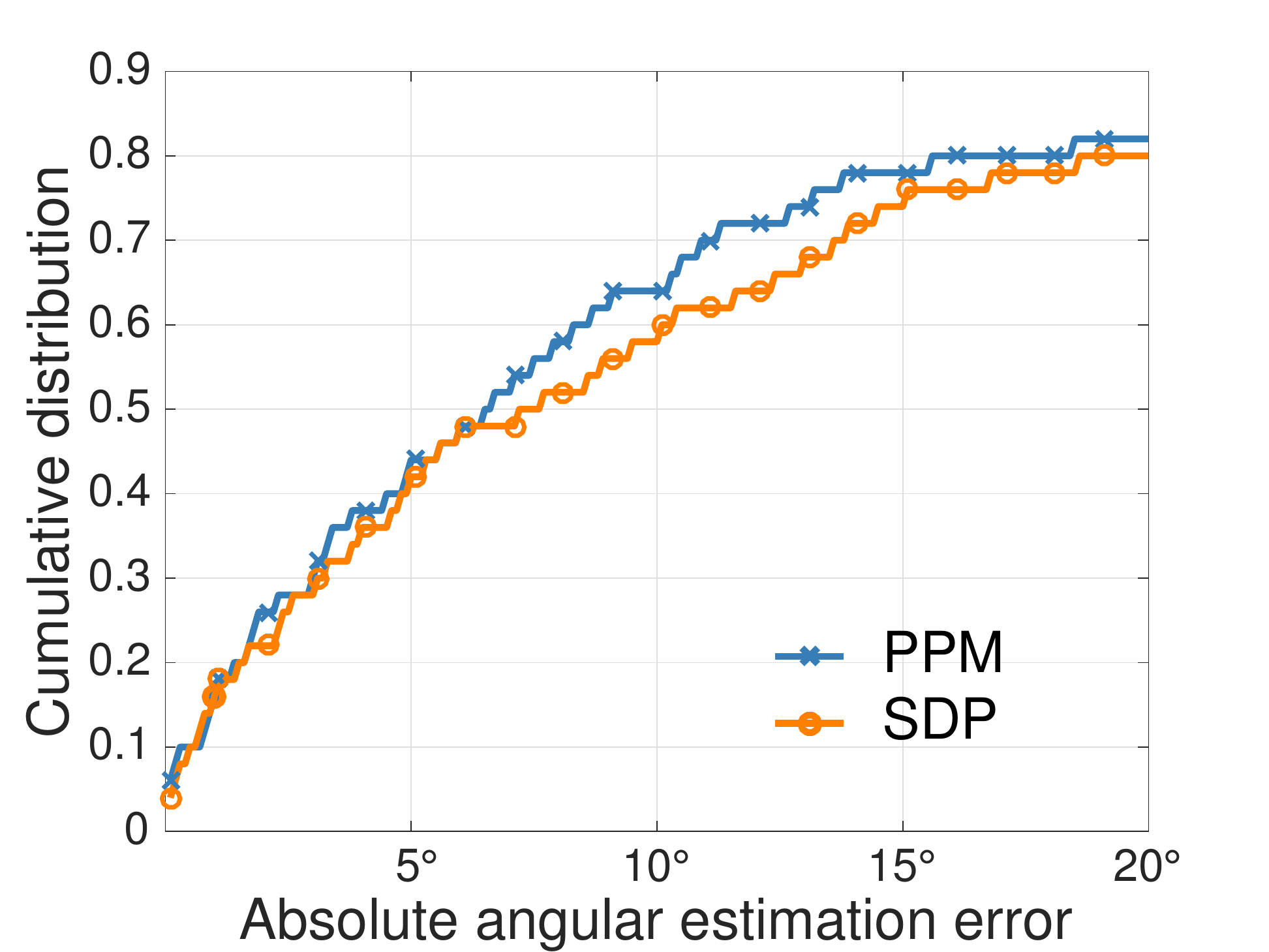}
\tabularnewline
\end{tabular}
\caption{The cumulative distributions of the absolute angular estimation errors on: (left) the Plane dataset, and (right) the Chair dataset. 
 \label{fig:PPM-SDP-chairs-planes}}
\end{figure}

\subsection{Joint shape alignment\label{sec:alignment-real}}

Next, we return to the motivating application---i.e.~joint image/shape alignment---of this work, and validate the applicability of the PPM on two datasets drawn from the ShapeNet repository \cite{chang2015shapenet}: (a) the Chair dataset (03001627), and (b) the Plane dataset (02691156).    Specifically,  $n=50$ shapes are taken from each dataset, 
and we randomly sample 8192 points from each shape as input features.  Each shape is rotated in the $x$-$z$ plane by a random continuous angle $\theta_i\in [0,360^{\circ})$.  Since the shapes in these datasets
have high quality and low noise, we perturb the shape data by adding
independent Gaussian noise $\mathcal{N}(0, 0.2^2)$ to each coordinate
of each point and use the perturbed data as inputs. This makes the task more challenging; for instance, the
resulting SNR on the Chair dataset is around 0.945 (since the mean square
values of each coordinate of the samples is 0.0378).

To apply the projected power method, we discretize the angular domain
  $[0,360^{\circ})$ by $m=32$ points, so that
  $x_i = j~(1 \leq j\leq 32)$ represents an angle
  $\theta_i = j360^{\circ}/32$. Following the procedure adopted
  in\footnote{\url{https://github.com/huangqx/map_synchronization/}}
  \cite{huang2013fine}, we compute the pairwise cost
  (i.e.~$-\ell(z_i,z_j)$) using some nearest-neighbor distance metric;
  to be precise, we set $-\ell(z_i,z_j)$ as the average
  nearest-neighbor squared distance between the samples of the $i$th
  and $j$th shapes, after they are rotated by
  $\frac{z_i}{32} 360^{\circ}$ and $\frac{z_j}{32} 360^{\circ}$,
  respectively. Such pairwise cost functions have been widely used in
  computer graphics and vision, and one can regard it as assuming that
  the average nearest-neighbor distance follows some Gaussian
  distribution.  Careful readers might remark that we have not
  specified $\{y_{i,j}\}$ in this experiment. Practically, oftentimes we
  only have access to some pairwise potential/cost functions rather
  than $\{y_{i,j}\}$. Fortunately, all we need to run the algorithm is
  $\ell(z_i,z_j)$, or some proxy of $\ell(z_i,z_j)$. 

Fig.~\ref{fig:shapes-chairs} shows the first 20 representative shapes before and
after joint alignment in the Chair dataset. As one can see, the shapes are
aligned in a reasonably good manner. More quantitatively, 
Fig.~\ref{fig:PPM-SDP-chairs-planes} displays the cumulative distributions
of the absolute angular estimation errors for both datasets. We have also reported in Fig.~\ref{fig:PPM-SDP-chairs-planes} the performance of semidefinite programming (SDP)---that is, the MatchLift algorithm presented in \cite{chen2014near}. 
 Note that the angular errors are measured as the distance to the un-discretized angles $\{\theta_i\}$, and are hence somewhat continuous. We see that for the PPM, 70\% (resp.~44\%) of the estimates on the Plane (resp.~Chair) dataset have an error of $5.5^{\circ}$ or lower, while the proportion is 48\% (resp.~44\%) for the SDP formulation. Recall that the resolution of the discretization is $360^{\circ}/32 \approx 11^{\circ}$, which would mean that all estimates with an error less than $5.5^{\circ}$ are, in some sense, perfect recoveries.

Computationally, it takes around 2.4 seconds to run the PPM, while SDP (implemented using the
alternating direction method of multipliers (ADMM)) runs in 895.6 seconds. All experiments are
carried out on a MacBook Pro equipped with a 2.9 GHz Intel Core i5 and
8GB of memory.

\subsection{Joint graph matching\label{sec:matching}}

The PPM is applicable to other combinatorial problems beyond joint alignment. We present here an example called joint graph matching \cite{kim2012exploring,huang2013consistent,chen2014near,PachauriKS13,gao2016geometry,Shen2016NIPS}. Consider a collection of $n$ images each containing $m$ feature points, and suppose that there exists one-to-one correspondence between the feature points in any pair of images. Many off-the-shelf algorithms are able to compute feature correspondence over the points in two images, and the joint matching problem concerns the recovery of a collection of globally consistent feature matches given these noisy pairwise matches. 
To put it mathematically, one can think of the ground truth as $n$ permutation matrices $\{\bm{X}_i\in \mathbb{R}^{m\times m} \}_{1\leq i\leq n}$ each representing the feature mapping between an image and a reference, and the true feature correspondence over the $i$th and $j$th images can be represented by $\bm{X}_i \bm{X}_j^{\top}$. The provided pairwise matches between the features of two images are encoded by $\bm{L}_{i,j}\in \mathbb{R}^{m\times m}$, which is a noisy version of $\bm{X}_i \bm{X}_j^{\top}$. The goal is then to recover $\{\bm{X}_i \}$---up to some global permutation---given a set of pairwise observations $\{ \bm{L}_{i,j} \}$. See \cite{huang2013consistent,chen2014near} for more detailed problem formulations as well as theoretical guarantees for convex relaxation. 

This problem differs from joint alignment in that the ground truth
$\bm{X}_i$ is an $m\times m$ permutation matrix. In light of this, we
make two modifications to the algorithm: (i) we maintain the iterates
$\bm{Z}^{t}=[\bm{Z}^{t}_i]_{1\leq i\leq n}$ as $nm \times m$ matrices
and replace $\mathcal{P}_{\Delta}$ by $\mathcal{P}_{\Pi}(\cdot)$ that
projects each $\bm{Z}^{t}_i\in \mathbb{R}^{m\times m}$ to the set of
permutation matrices (via the Jonker-Volgenant algorithm
\cite{jonker1987shortest}), which corresponds to hard rounding
(i.e.~$\mu_t=\infty$) in power iterations (ii) the initial guess
${\bm{Z}}^{(0)}\in \mathbb{R}^{nm \times m}$ is taken to be the
projection of a random column block of $\hat{\bm{L}}$ (which is the
rank-$m$ approximation of $\bm{L}$).

We first apply the PPM on two benchmark image datasets:  
(1) the CMU House dataset\footnote{\url{http://vasc.ri.cmu.edu/idb/html/motion/house/}} consisting of $n=111$ images of a house, and (2) the CMU Hotel dataset\footnote{\url{http://vasc.ri.cmu.edu//idb/html/motion/hotel/index.html}}
consisting of $n=101$ images of a hotel. Each image contains $m=30$  feature points that have been labeled consistently across all images. The initial pairwise matches, which are obtained through the Jonker-Volgenant algorithm, have mismatching rates of 13.36\% (resp.~12.94\%)  for the House (resp.~Hotel) dataset.  Our algorithm allows to lower the mismatching rate to 3.25\% (resp.~4.81\%) for House (resp.~Hotel). Some representative results from each dataset are depicted in Fig.~\ref{fig:matching-CMU-house-hotel}.

\begin{figure}
\centering%
\begin{tabular}{ccc}
\includegraphics[width=0.4\textwidth]{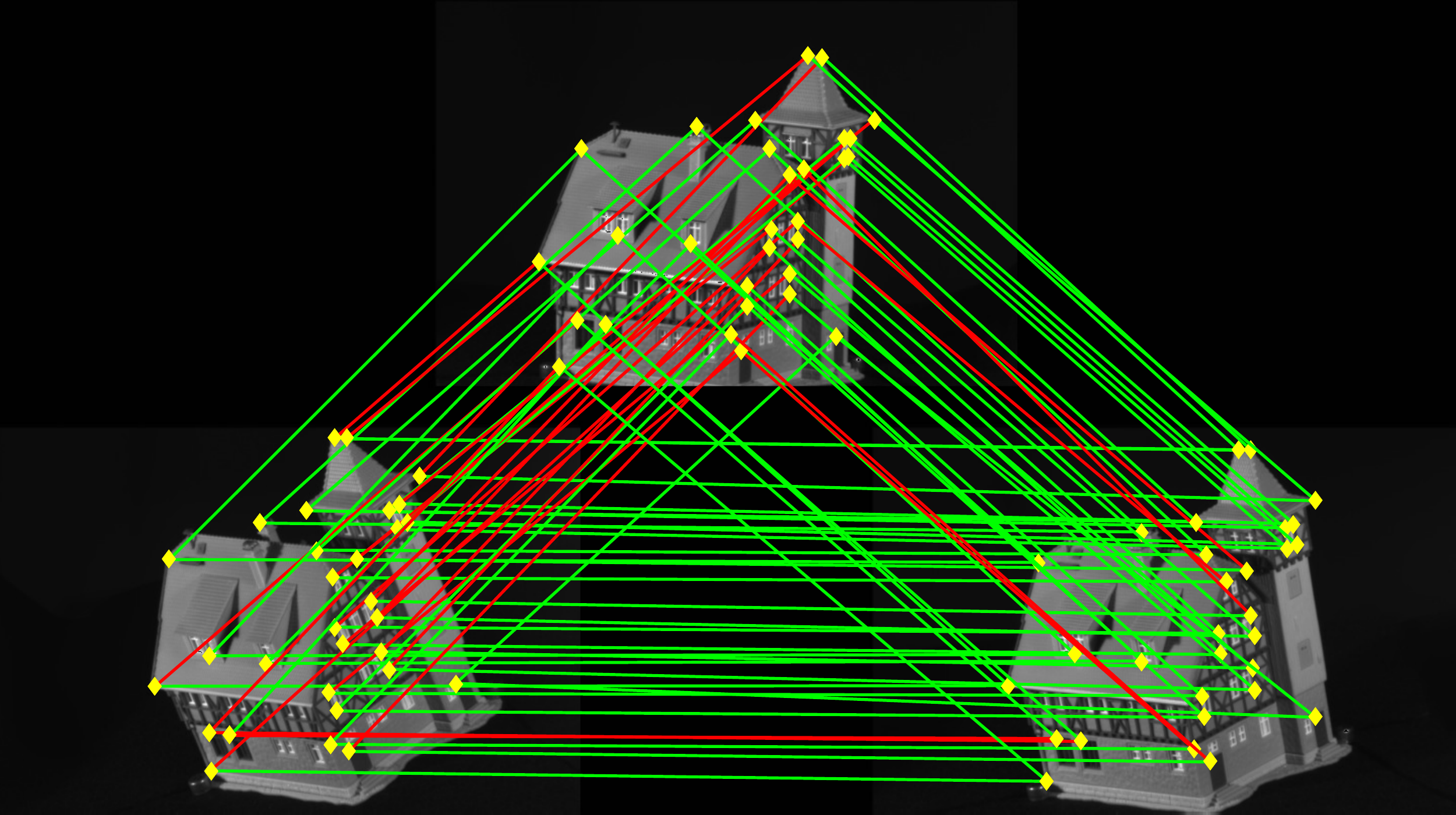} & $\qquad$ & \includegraphics[width=0.4\textwidth]{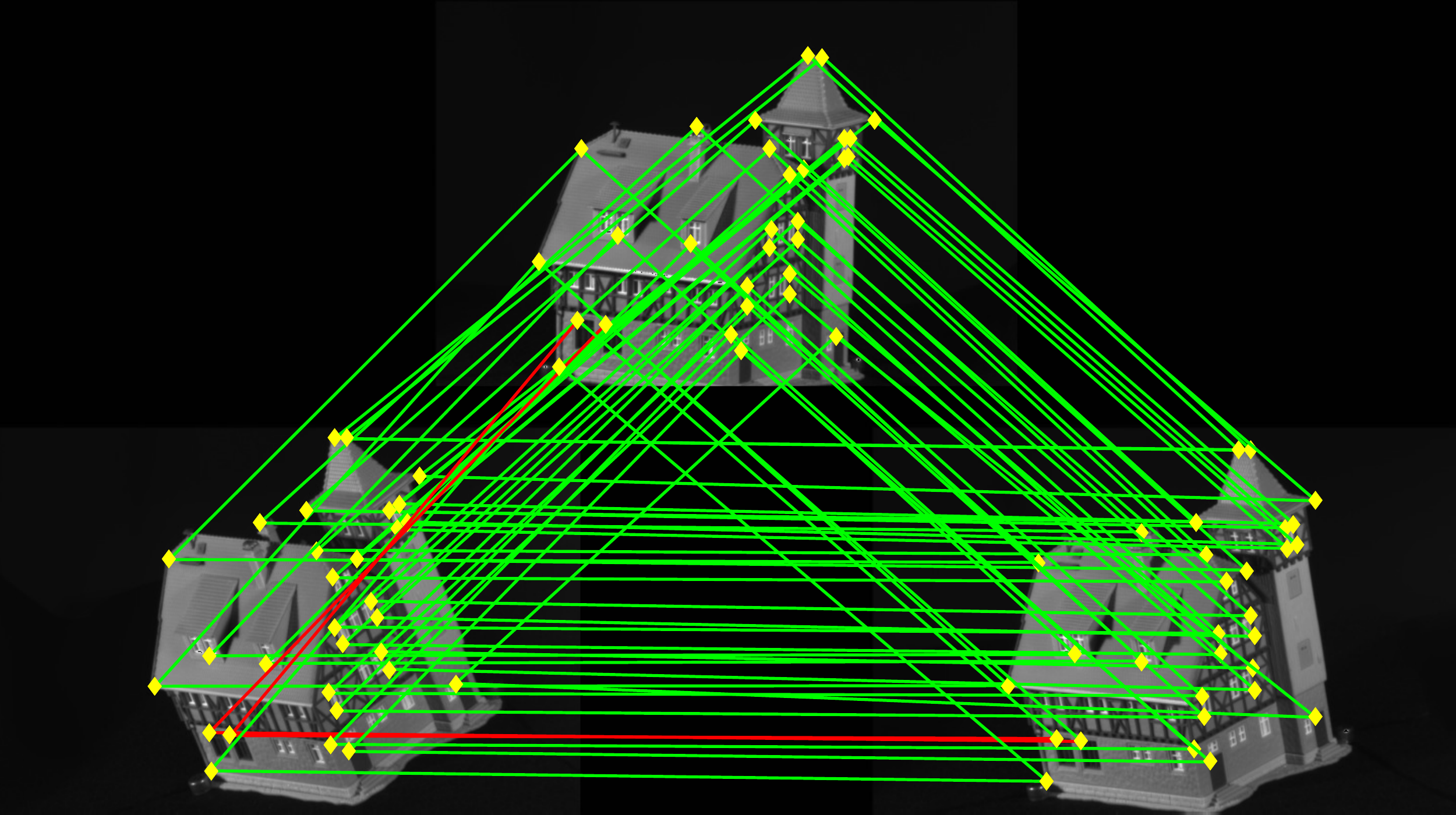}\tabularnewline
(a) initial pairwise matches (CMU House) &  & (b) optimized matches (CMU House)\tabularnewline
\includegraphics[width=0.4\textwidth]{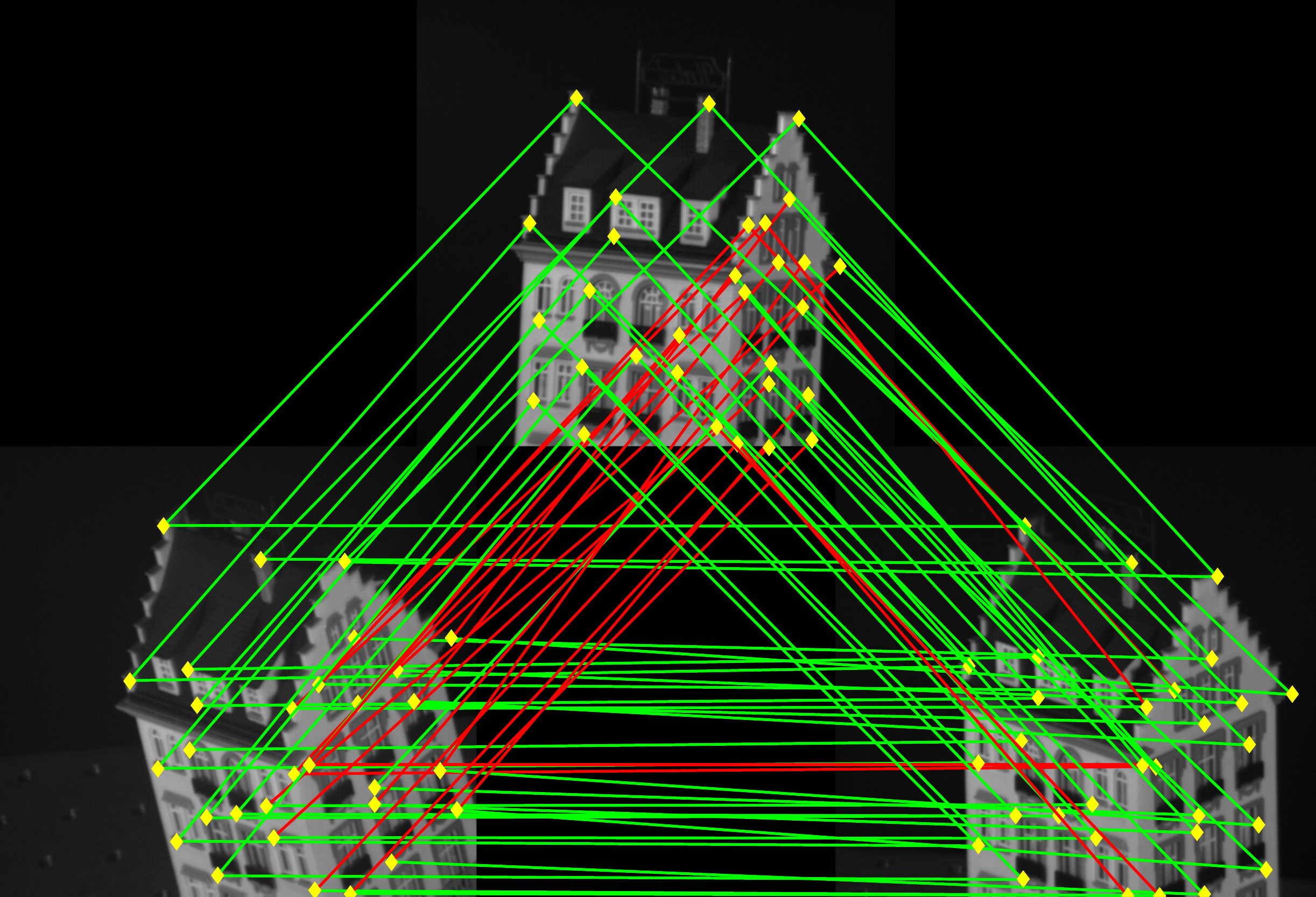} &  & \includegraphics[width=0.4\textwidth]{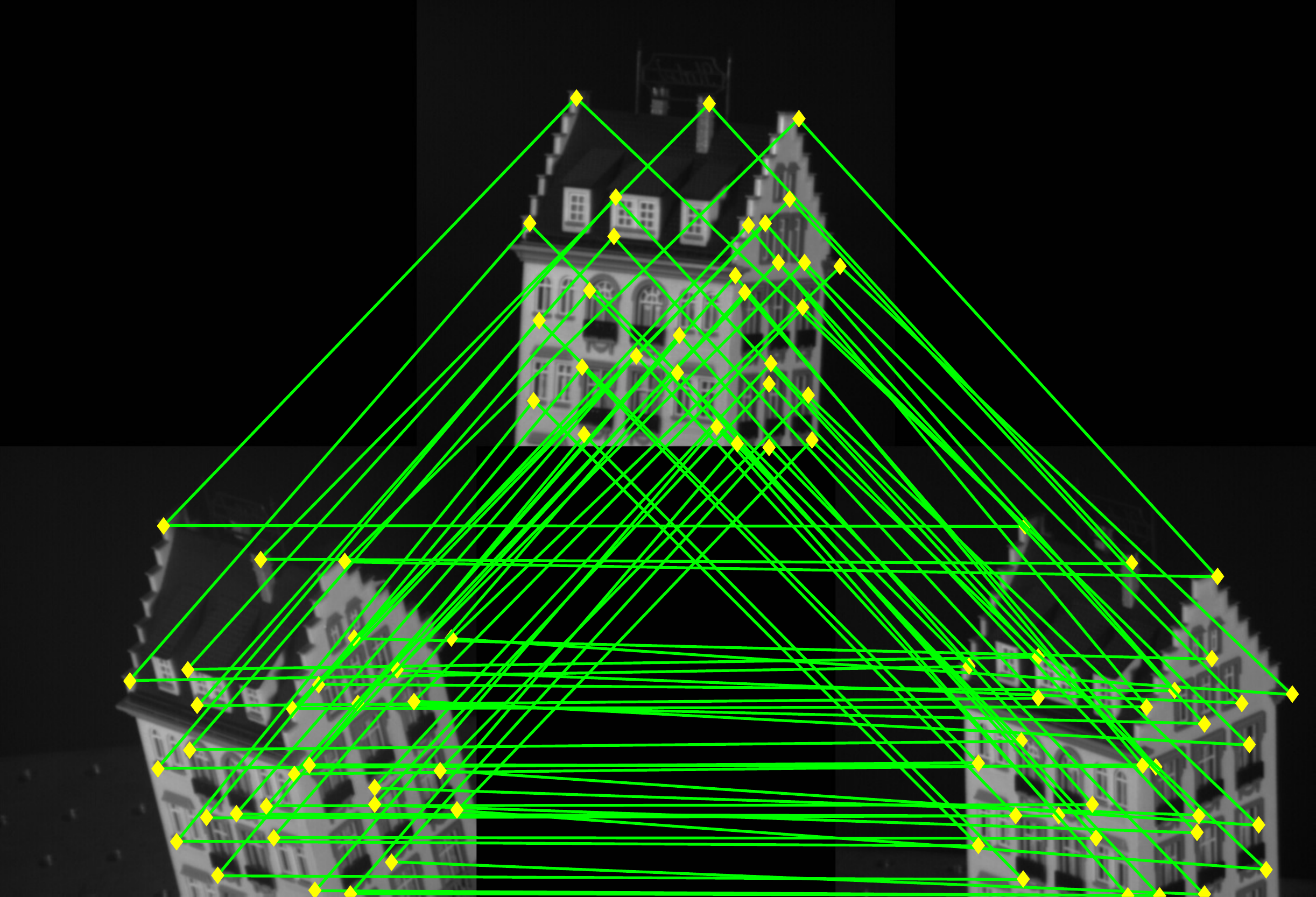}\tabularnewline
(c) initial pairwise matches (CMU Hotel) &  & (d) optimized matches (CMU Hotel)\tabularnewline
\end{tabular}
\caption{Comparisons between the input matches and the outputs of the PPM on the CMU House and Hotel
datasets, with 3 representative images shown for each dataset. The yellow dots refer to the manually labeled feature points, while the green (resp.~red) lines represent the set of matches consistent (resp.~inconsistent) with the ground truth. 
\label{fig:matching-CMU-house-hotel}
}
\end{figure}

Next, we turn to three shape datasets: (1) the Hand dataset containing
$n=20$ shapes, (2) the Fourleg dataset containing $n=20$ shapes, and
(3) the Human dataset containing $n=18$ shapes, all of which are drawn
from the collection SHREC07 \cite{giorgi2007shape}.  We set $m = 64$,
$m = 96$, and $m = 64$ feature points for Hand, Fourleg, and Human datasets,
respectively, and follow the shape sampling and pairwise matching
procedures described in
\cite{huang2013consistent}\footnote{\url{https://github.com/huangqx/CSP_Codes}}. To
evaluate the matching performance, we report the fraction of output
matches whose normalized geodesic errors (see
\cite{kim2012exploring,huang2013consistent}) are below some threshold
$\epsilon$, with $\epsilon$ ranging from 0 to 0.25. For the sake of
comparisons, we plot in Fig.~\ref{fig:matching-SHREC} the quality of
the initial matches, the matches returned by the projected power
method, as well as the matches returned by semidefinite relaxation
\cite{huang2013consistent,chen2014near}. The computation runtime is
reported in Table \ref{tab:Runtime-matching}.  The numerical results demonstrate that the projected
power method is significantly faster than SDP, while achieving a joint
matching performance as competitive as SDP.

\begin{figure}
\centering%
\begin{tabular}{ccc}
\includegraphics[width=0.3\textwidth]{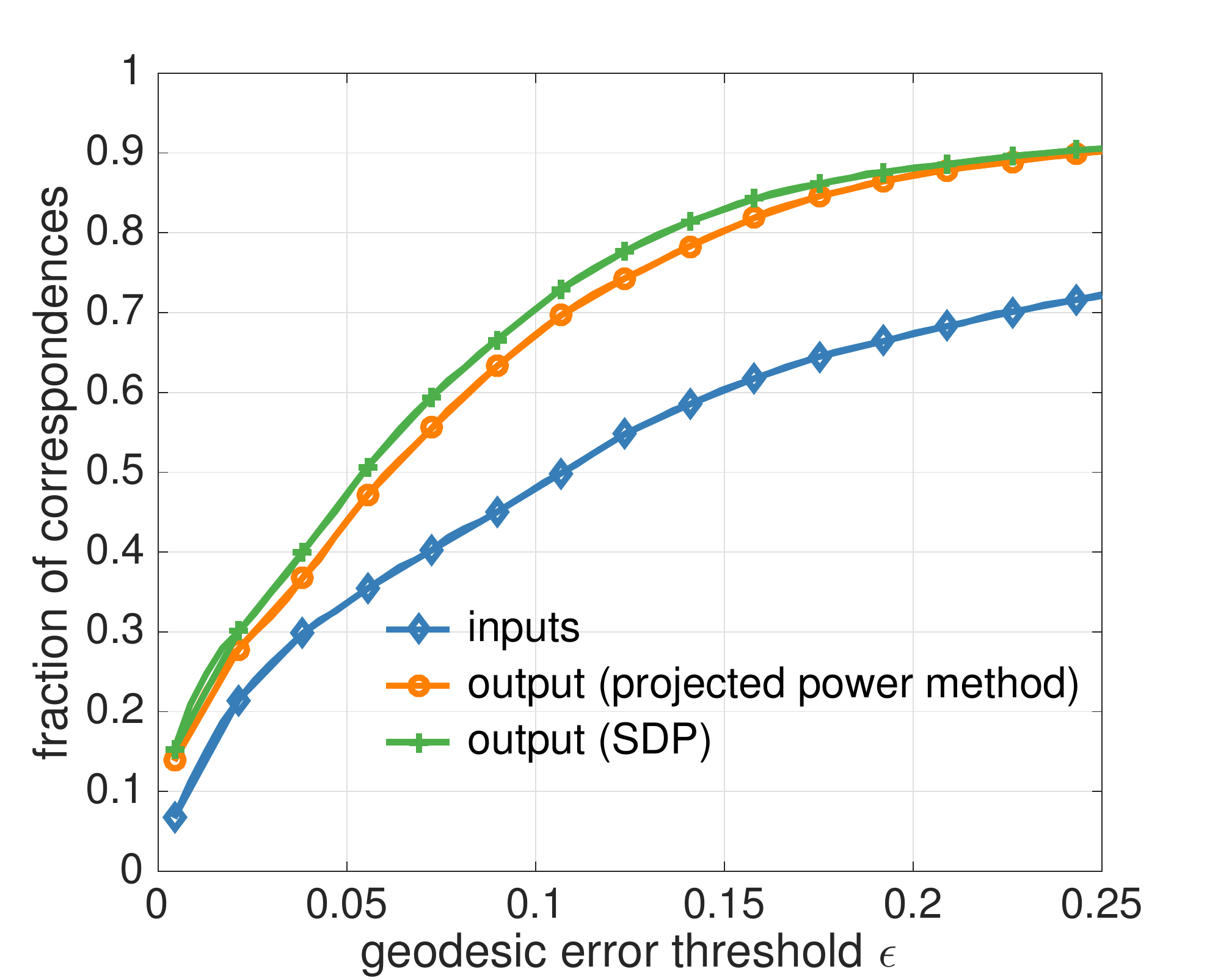} &  \includegraphics[width=0.3\textwidth]{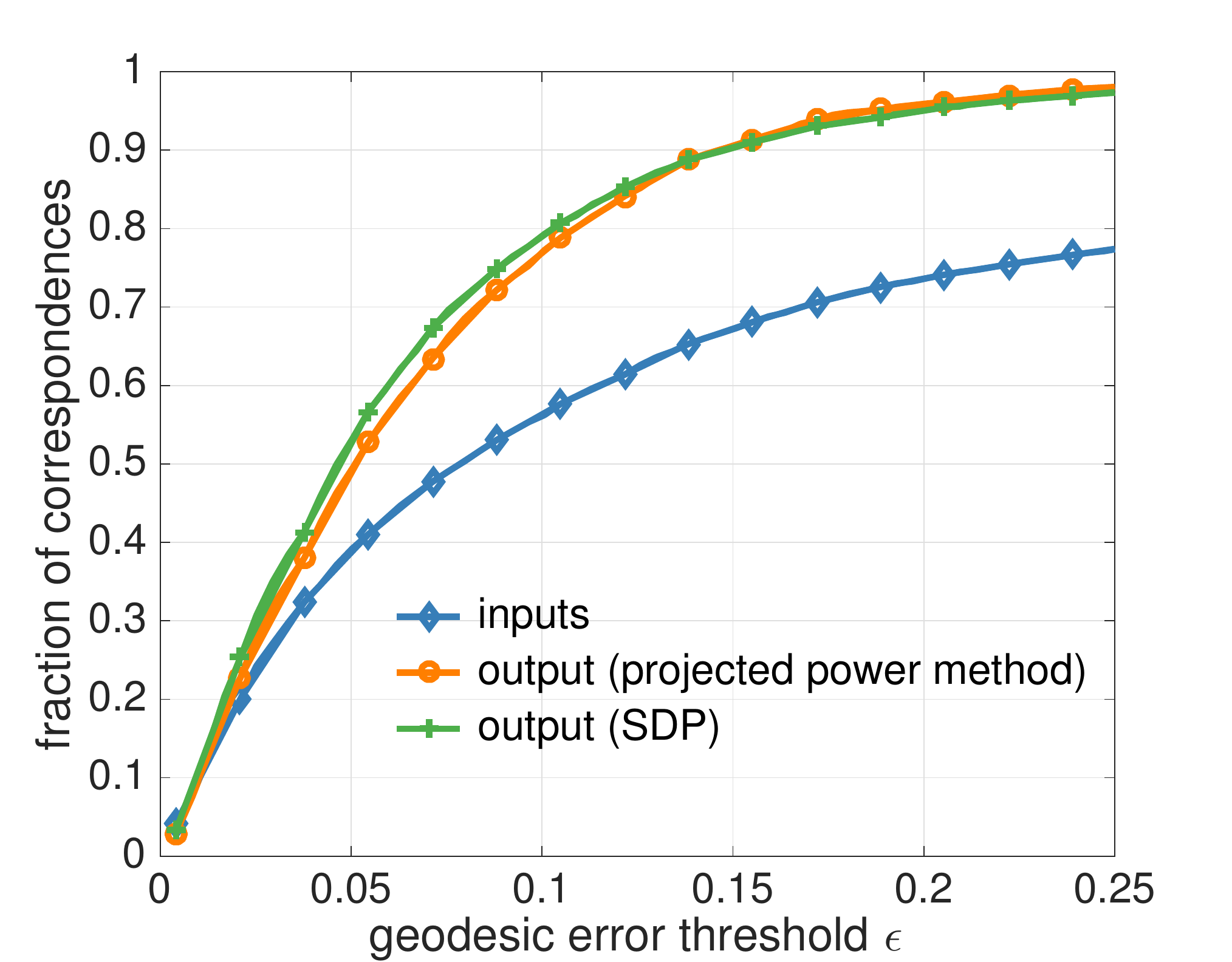} & \includegraphics[width=0.3\textwidth]{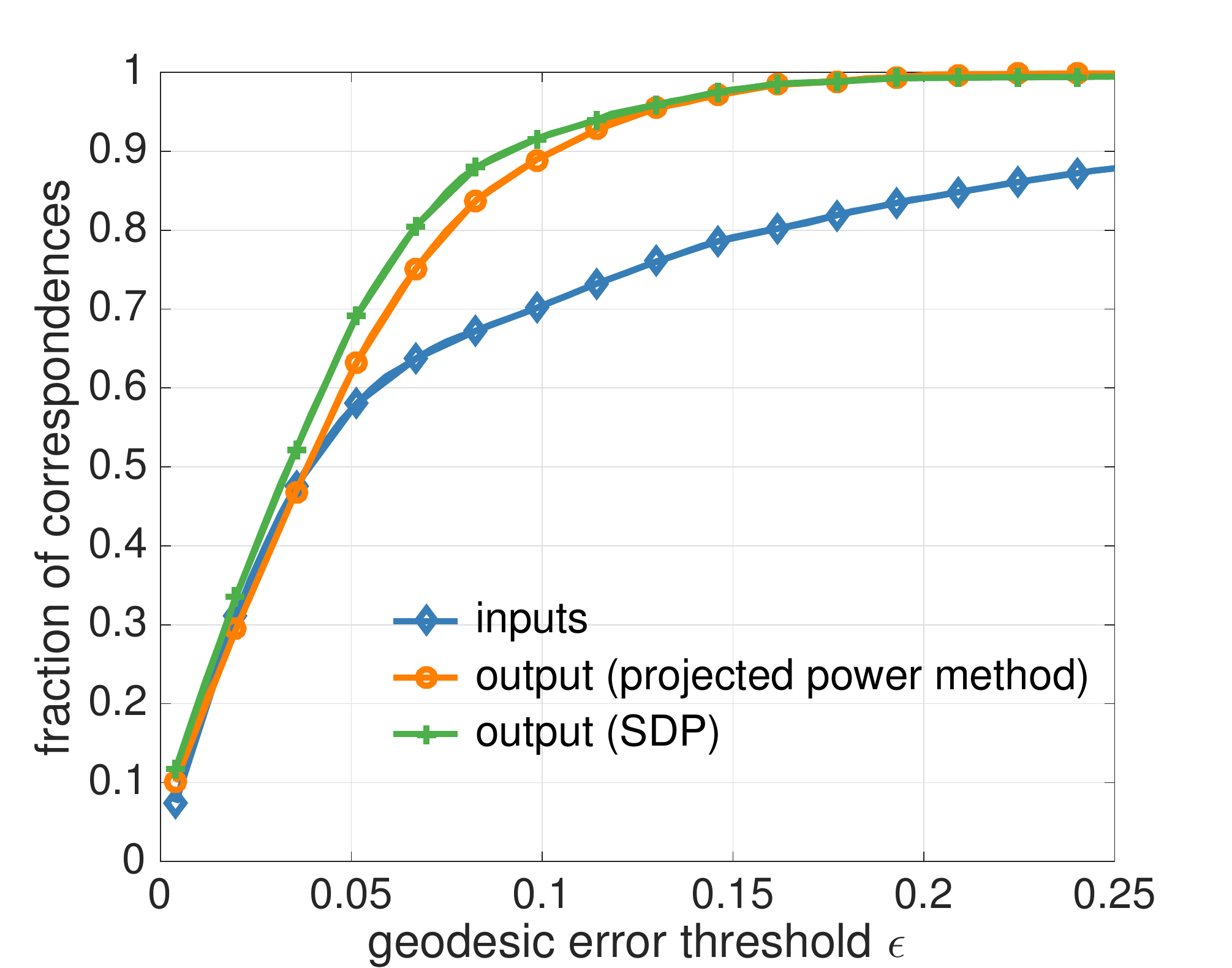} \tabularnewline
 (a) Hand & (b) Fourleg  & (c) Human \tabularnewline
\end{tabular}
\caption{The fraction of correspondences whose normalized geodesic errors are smaller than some threshold $\epsilon$. \label{fig:matching-SHREC}
 }
\end{figure}

\begin{table}
\centering

\begin{tabular}{c|c|c|c}
\hline 
 & Hand & Fourleg & Human\tabularnewline
\hline 
SDP & 455.5 sec & 1389.6 sec & 368.9 sec\tabularnewline
\hline 
PPM & 35.1 sec & 76.8 sec & 40.8 sec\tabularnewline
\hline 
\end{tabular}\caption{ 
Runtime of SDP (implemented using ADMM) and the PPM on 3 benchmark datasets, when carried out on a MacBook Pro equipped with a 2.9 GHz i5 and 8GB of memory.  
\label{tab:Runtime-matching}}

\end{table}

\section{Preliminaries and notation\label{sec:Preliminaries}}

Starting from this section, we turn attention to the analyses of the main results.  
Before proceeding, we gather a few preliminary facts and notations
that will be useful throughout.

\subsection{Projection onto the standard simplex}

Firstly, our algorithm involves projection onto the standard simplex $\Delta$. In light
of this, we single out several elementary facts concerning $\Delta$ and
$\mathcal{P}_{\Delta}$ as follows. Here and throughout, $\|\bm{a}\|$ is the $\ell_2$ norm of a vector $\bm{a}$.  

\begin{fact}
\label{fact:H}
Suppose that $\bm{a}=\left[a_{1},\cdots,a_{m}\right]^{\top}$
obeys $\bm{a}+\bm{e}_{l}\in\Delta$ for some $1\leq l\leq m$. Then
$\left\Vert \bm{a}\right\Vert \leq\sqrt{2}$. 
\end{fact}

\begin{proof}
The feasibility condition requires $0\leq\sum_{i:i\neq l}a_{i}=-a_{l}\leq1$
and $a_{i}\geq0$ for all $i\neq l$. Therefore, it is easy to check
that $\sum\nolimits _{i:i\neq l}a_{i}^{2}\leq\sum\nolimits _{i:i\neq l}a_{i}\leq1$,
and hence $\|\bm{a}\|^{2}=a_{l}^{2}+\sum\nolimits _{i:i\neq l}a_{i}^{2}\leq2$.
\end{proof}

\begin{fact}
\label{fact:Proj-const-shift}
For any vector $\bm{v}\in\mathbb{R}^{m}$
and any value $\delta$, one has 
\begin{equation}
  \mathcal{P}_{\Delta}\left(\bm{\bm{v}}\right) = \mathcal{P}_{\Delta}\left(\bm{v}+\delta\bm{1}\right).
  \label{eq:proj-constant-shift}
\end{equation}
\end{fact}

\begin{proof}
  For any $\bm{x} \in \Delta$,
\[
  \|\bm{v}+\delta\bm{1} -
  \bm{x}\|^2 =  \|\bm{v} -
  \bm{x}\|^2 + \delta^2 \|\bm{1}\|^2 + 2\delta (\bm{v}  - \bm{x})^{\top}
    \bm{1} =   \|\bm{v} -
  \bm{x}\|^2 + \delta^2 n + 2\delta (\bm{v}^{\top} \bm{1} - 1). 
\]
Hence,
$\mathcal{P}_{\Delta}\left(\bm{v}+\delta\bm{1}\right) = \text{arg min}_{\bm{x} \in \Delta} \|\bm{v} + \delta \bm{1} - 
\bm{x}\|^2 = \text{arg min}_{\bm{x} \in \Delta} \|\bm{v} -
\bm{x}\|^2 = \mathcal{P}_{\Delta}\left(\bm{\bm{v}}\right)$.

\end{proof}

\begin{fact}
\label{fact:Proj-separation}
For any non-zero vector
$\bm{v}=\left[v_{i}\right]_{1\leq i\leq m}$, let $v_{(1)}$ and $v_{(2)}$
be its largest and second largest entries, respectively. Suppose $v_{j}=v_{(1)}$.
If $\mu>1/(v_{(1)}-v_{(2)})$, then
\begin{equation}
\mathcal{P}_{\Delta}\left(\mu\bm{\bm{v}}\right)=\bm{e}_{j}.\label{eq:proj-sep}
\end{equation}
\end{fact}\begin{proof}
By convexity of $\Delta$, we have $\mathcal{P}_{\Delta}\left(\mu\bm{\bm{v}}\right)=\bm{e}_{j}$
if and only if, for any $\bm{x} =[x_i]_{1\leq i\leq m} \in \Delta$, 
\[
  (\bm{x} - \bm{e}_{j})^\top(\mu\bm{v} - \bm{e}_{j}) \le 0. 
\]
Since $\bm{x}^\top \bm{v} = x_j v_{(1)} + \sum_{i: i \neq j} x_i v_i
\le x_j v_{(1)} + v_{(2)} \sum_{i: i \neq j} x_i = x_j v_{(1)} + v_{(2)} (1-x_j)$, we see that 
\[
 (\bm{x} - \bm{e}_{j})^\top(\mu\bm{v} - \bm{e}_{j}) \le (1-x_j)(1 -
 \mu(v_{(1)} - v_{(2)})) \le 0.
\]
\end{proof}

In words, Fact \ref{fact:Proj-const-shift} claims that a global offset
does not alter the projection $\mathcal{P}_{\Delta}\left(\cdot\right)$,
while Fact \ref{fact:Proj-separation} reveals that a large scaling
factor $\mu$ results in sufficient separation between the largest
entry and the remaining ones. See Fig.~\ref{fig:Illustration-of-Facts}
for a graphical illustration. 

\begin{figure}
\centering
\begin{tabular}{cc}
\includegraphics[width=0.32\textwidth]{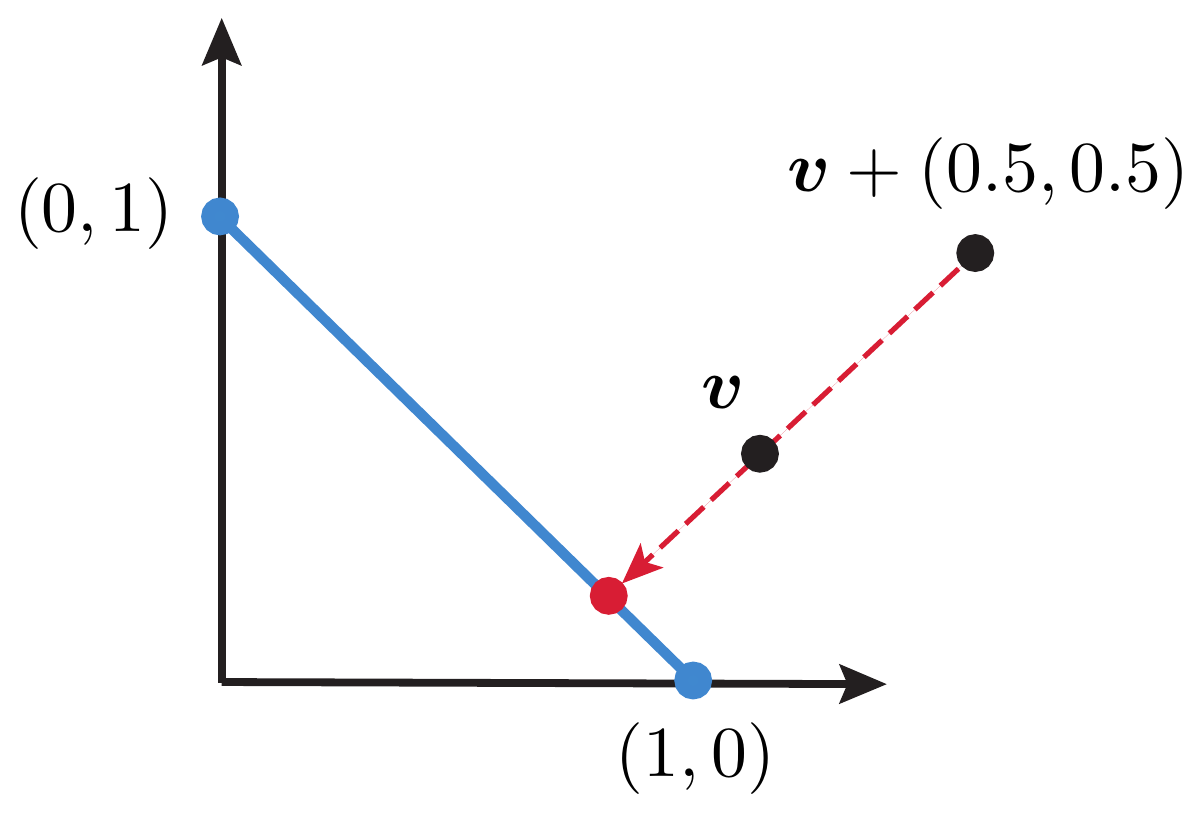} & \includegraphics[width=0.31\textwidth]{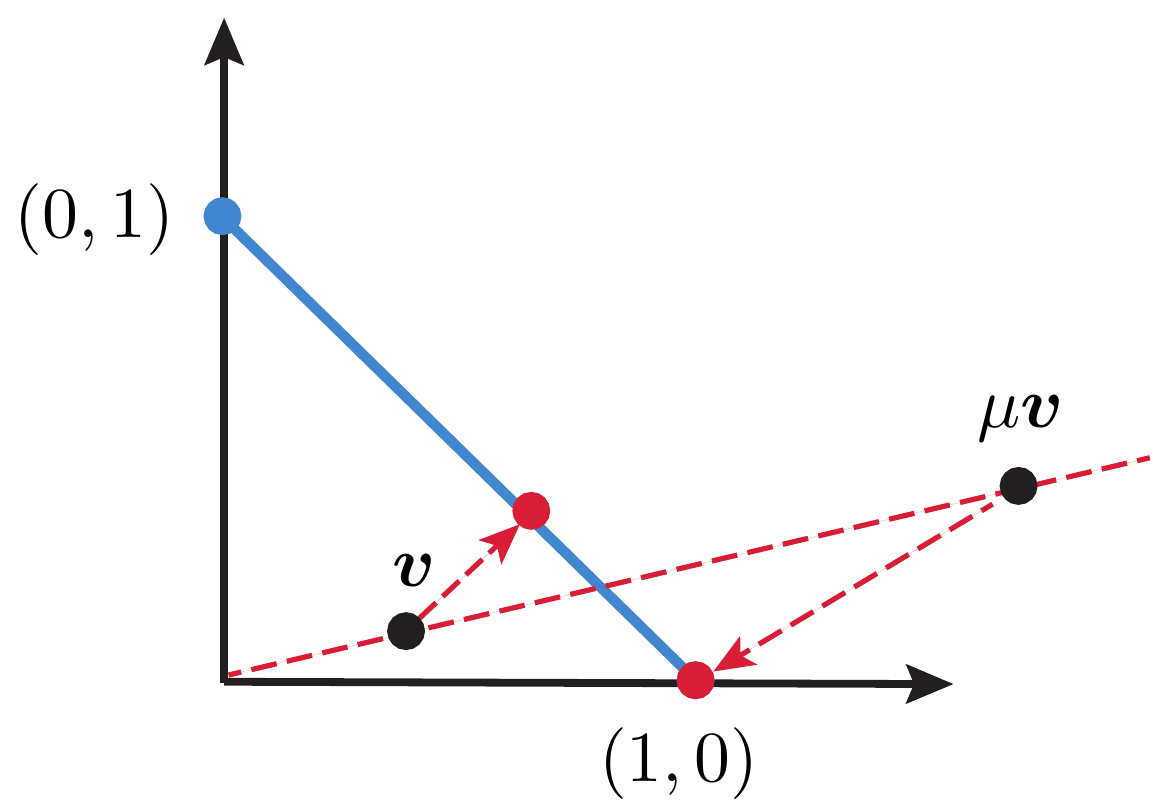} \tabularnewline
(a) & (b)\tabularnewline
\end{tabular}
\caption{Illustration of Facts \ref{fact:Proj-const-shift}-\ref{fact:Proj-separation}
on a two-dimensional standard simplex. 
(a) For any $\bm{v}$, $\mathcal{P}_{\Delta}\left(\bm{v}\right)=\mathcal{P}_{\Delta}\left(\bm{v}+0.5\cdot\bm{1}\right)$;
(b) for an vector $\bm{v}\in\mathbb{R}^{2}$ obeying $v_{1}>v_{2}$,
one has $\mathcal{P}_{\Delta}\left(\mu\bm{v}\right)=\bm{e}_{1}$ when
$\mu$ is sufficiently large. \label{fig:Illustration-of-Facts}}
\end{figure}

\subsection{Properties of the likelihood ratios}

Next, we study the log-likelihood
ratio statistics. The first result makes a connection between the
KL divergence and other properties of the log-likelihood ratio. Here
and below, for any two distributions $P$ and $Q$ supported on $\mathcal{Y}$,
the total variation distance between them is defined by 
$\mathsf{TV}\left(P,Q\right):=\frac{1}{2}\sum\nolimits _{y\in\mathcal{Y}}\left|P(y)-Q(y)\right|$.

\begin{lem}
\label{lemma:V-KL-UB}
(1) Consider two probability distributions
$P$ and $Q$ over a finite set $\mathcal{Y}$. Then 
\begin{eqnarray}
  \left| \log\frac{Q(y)}{P(y)}\right| & \leq & \frac{ 2\mathsf{TV}\left(P,Q\right)}{\min\left\{ P\left( y \right),Q\left( y \right)\right\} }
  \leq \frac{ \sqrt{2\mathsf{KL}\left(P\hspace{0.3em}\|\hspace{0.3em}Q\right)} }{ \min\left\{ P\left(y\right),Q\left(y\right)\right\} }.
\end{eqnarray}

(2) In addition, if both $\max_{y\in\mathcal{Y}}\frac{P(y)}{Q(y)} \leq \kappa_{0}$
and $\max_{y\in\mathcal{Y}}\frac{Q(y)}{P(y)} \leq \kappa_{0}$ hold,
then 
\begin{equation}
  \mathbb{E}_{y\sim P} \left[ \left(\log\frac{P(y)}{Q(y)}\right)^2 \right] 
  \leq 2 \kappa_{0}^{2} \mathsf{KL}\left(Q \hspace{0.3em}\|\hspace{0.3em} P \right).
  \label{eq:V-KL-UB}
\end{equation}

\end{lem}
\begin{proof} See Appendix \ref{sec:Proof-of-Lemma-V-KL-UB}.\end{proof}

In particular, when $\mathsf{KL}\left(P\hspace{0.3em}\|\hspace{0.3em}Q\right)$
is small, one almost attains equality in (\ref{eq:V-KL-UB}), as stated
below.

\begin{lem}
\label{lem:KL-Var-Hel}
Consider two probability distributions
$P$ and $Q$ on a finite set $\mathcal{Y}$ such that $P(y)$ and
$Q(y)$ are both bounded away from zero. If $\mathsf{KL}\left(P\hspace{0.3em}\|\hspace{0.3em}Q\right) \leq \varepsilon$ for some  $0<\varepsilon<1$,
then one has
\begin{equation}
  \mathsf{KL}\left(P\hspace{0.3em}\|\hspace{0.3em}Q\right) = \frac{1+\zeta_{1}\left(\varepsilon\right)}{2}{\bf Var}_{y\sim P}\left[\log\frac{P\left(y\right)}{Q\left(y\right)}\right]
  \label{eq:lemma-V-KL}
\end{equation}
\begin{equation}
  \text{and} \qquad \mathsf{H}^{2}(P,Q) = \frac{1+\zeta_{2} (\varepsilon)}{4} \mathsf{KL}( P\hspace{0.3em}\|\hspace{0.3em}Q ),
  \label{eq:lemma-H2-KL}
\end{equation}
where $\zeta_{1}$ and $\zeta_{2}$ are functions satisfying $|\zeta_{1}(\varepsilon)|,|\zeta_{2}(\varepsilon)|\leq c_{0}\sqrt{\varepsilon}$
for some universal constant $c_{0}>0$. 

\end{lem}

\begin{proof}See Appendix \ref{sec:Proof-of-Lemma-Chernoff-LB}.\end{proof}

\subsection{Block random matrices}

Additionally, the data matrix $\bm{L}$ is assumed to have independent
blocks. It is thus crucial to control the fluctuation of such random
block matrices, for which the following lemma proves useful. 

\begin{lem}
\label{lem:random-matrix-norm}
Let $\bm{M}:=[\bm{M}_{i,j}]_{1\leq i,j\leq n}$
be any random symmetric block matrix, where $\left\{ \bm{M}_{i,j}\in\mathbb{R}^{m\times m}\mid i\geq j\right\} $
are independently generated. Suppose that $m=O\left(\mathrm{poly}(n)\right)$,
$\mathbb{E}\left[\bm{M}_{i,j}\right]=\bm{0}$, $\max_{i,j}\|\bm{M}_{i,j}\|\leq K$,
and  $\mathbb{P}\left\{ \bm{M}_{i,j}=\bm{0}\right\} =p_{\mathrm{obs}}$
for some $p_{\mathrm{obs}}\gtrsim\log n/n$. Then with probability
exceeding $1-O(n^{-10})$, 
\begin{eqnarray}
  \|\bm{M}\| & \lesssim & K\sqrt{np_{\mathrm{obs}}}.
\end{eqnarray}
\end{lem}

\begin{proof}  
See Appendix \ref{sec:Proof-of-Lemma-Random-Matrix-Norm}.
\end{proof}

Lemma \ref{lem:random-matrix-norm}
immediately leads to an upper estimate on the fluctuations of $\bm{L}$ and $\bm{L}^{\mathrm{debias}}$. 

\begin{lem}
\label{lem:deviation-Lhat}
Suppose $m=O(\mathrm{poly}\left(n\right))$, and define $\left\Vert \log\frac{P_{0}}{P_{l}}\right\Vert _{1}:=\sum_{y}\left|\log\frac{P_{0}(y)}{P_{l}(y)}\right|$.
If $p_{\mathrm{obs}}\gtrsim \log n / n$, then with probability
exceeding $1-O\left(n^{-10}\right)$, the matrices $\bm{L}$ and $\bm{L}^{\mathrm{debias}}$ given respectively in (\ref{eq:Input-Matrix}) and (\ref{eq:debiased-MLE}) satisfy
\begin{eqnarray}
  \left\Vert \bm{L}-\mathbb{E}[\bm{L}]\right\Vert   
   ~=~ \left\Vert \bm{L}^{\mathrm{debias}} - \mathbb{E}[ \bm{L}^{\mathrm{debias}} ]\right\Vert
  ~\lesssim~ \left(\frac{1}{m} \sum\nolimits _{l=1}^{m-1}\left\Vert \log\frac{P_{0}}{P_{l}}\right\Vert _{1}\right) \sqrt{np_{\mathrm{obs}}}.
\end{eqnarray}
\end{lem}
\begin{proof}
  See Appendix \ref{sec:Proof-of-Lemma-deviation-Lhat}.
\end{proof}

\subsection{Other notation}

For any vector $\bm{h}=\left[\bm{h}_{i}\right]_{1\leq i\leq n}\in\mathbb{R}^{mn}$
with $\bm{h}_{i}\in\mathbb{R}^{m}$, we denote by $h_{i,j}$
the $j$th component of $\bm{h}_{i}$. For any $m\times n$ matrix
$\bm{A}=\left[a_{i,j}\right]_{1\leq i\leq m, 1\leq j\leq n}$ and any matrix $\bm{B}$, the Kronecker
product $\bm{A}\otimes\bm{B}$ is defined as
\[
  \bm{A}\otimes\bm{B} := \left[  
  \begin{array}{ccc}
    a_{1,1}\bm{B} & \cdots & a_{1,n}\bm{B}\\
    \vdots & \ddots & \vdots\\
    a_{m,1}\bm{B} & \cdots & a_{m,n}\bm{B}
  \end{array} \right].
\]
%


\section{Iterative stage\label{sec:Iterative-stage-general}}

We establish the performance guarantees of our two-stage algorithm
in a reverse order. Specifically, we demonstrate in this section that
the iterative refinement stage achieves exact recovery,
provided that the initial guess is reasonably close to the truth. 
The analysis for the initialization
is deferred to Section \ref{sec:Spectral-Initialization}.

\subsection{Error contraction\label{sec:Error-Contraction}}

This section mainly consists of establishing
the following claim, which concerns error contraction of iterative refinement in the presence
of an appropriate initial guess.

\begin{theorem}
\label{thm:Iterative-stage}
Under the conditions of
Theorem \ref{thm:ExactRecovery-General} or Theorem \ref{thm:ExactRecovery-General-large-m},
there exist some absolute constants $0<\rho,c_{1}<1$ such that with
probability exceeding $1-O\left(n^{-5}\right)$, 
\begin{equation}
  \left\Vert \mathcal{P}_{\Delta^{n}}\left(\mu\bm{L}\bm{z}\right)-\bm{x}\right\Vert _{*,0} 
  \leq \rho\min\left\{ \left\Vert \bm{z}-\bm{x}\right\Vert _{*,0},\text{ }\left\Vert \bm{z}-\bm{x}\right\Vert ^{2}\right\} 
  \label{eq:error-contraction-general-proof}
\end{equation}
holds simultaneously for all $\bm{z}\in\Delta^{n}$ obeying\footnote{The numerical constant 0.49 is arbitrary and can be replaced by any
other constant in between 0 and 0.5.} 
\begin{equation}
  \min\left\{ \frac{\left\Vert \bm{z}-\bm{x}\right\Vert }{\|\bm{x}\|},~\frac{ \Vert \bm{z}-\bm{x}\Vert _{*,0}}{\|\bm{x}\|_{*,0}}\right\} 
  \leq 0.49 \frac{\mathsf{KL}_{\min}}{\mathsf{KL}_{\max}},
  \label{eq:basin-attraction-general-proof}
\end{equation}
provided that 
\begin{equation}
	\mu > \frac{c_{5}}{ np_{\mathrm{obs}}\mathsf{KL}_{\min} } \label{eq:mu-KL}
\end{equation}
 for some sufficiently large constant $c_5>0$.

\end{theorem}

At each iteration, the PPM produces a more accurate estimate as
long as the iterates $\{\bm{z}^{(t)}\}$ stay within a reasonable
neighborhood surrounding $\bm{x}$. Here and below, we term this
neighborhood a \emph{basin of attraction}. In fact, if the initial
guess $\bm{z}^{(0)}$ successfully lands within this basin, then the
subsequent iterates will never jump out of it. To see this, observe
that for any $\bm{z}^{(t)}$ obeying
\[
  \frac{\|\bm{z}^{(t)}-\bm{x}\|}{\|\bm{x}\|}
  \leq 0.49 \frac{\mathsf{KL}_{\min}}{\mathsf{KL}_{\max}}
  \quad\text{or} \quad  \frac{\|\bm{z}^{(t)}-\bm{x}\|_{*,0}}{\|\bm{x}\|_{*,0}}
  \leq 0.49\frac{\mathsf{KL}_{\min}}{\mathsf{KL}_{\max}},
\]
the inequality (\ref{eq:error-contraction-general-proof}) implies
error contraction
\[
  \| \bm{z}^{(t+1)}-\bm{x}\|_{*,0}\leq\rho\|\bm{z}^{(t)}-\bm{x}\|_{*,0}.
\]
Moreover, since $\left\Vert \bm{x}\right\Vert _{*,0}=\left\Vert \bm{x}\right\Vert ^{2}=n$,
one has
\begin{align*}
  \frac{\left\Vert \bm{z}^{(t+1)}-\bm{x}\right\Vert _{*,0}}{\left\Vert \bm{x}\right\Vert _{*,0}} 
	&< \min\left\{ \frac{\left\Vert \bm{z}^{(t)}-\bm{x}\right\Vert _{*,0}}{\left\Vert \bm{x}\right\Vert _{*,0}},\frac{\left\Vert \bm{z}^{(t)}-\bm{x}\right\Vert ^{2}}{\left\Vert \bm{x}\right\Vert ^{2}}\right\} \nonumber\\
	&\leq  \min\left\{ \frac{\left\Vert \bm{z}^{(t)}-\bm{x}\right\Vert _{*,0}}{\left\Vert \bm{x}\right\Vert _{*,0}},\frac{\left\Vert \bm{z}^{(t)}-\bm{x}\right\Vert }{\left\Vert \bm{x}\right\Vert }\right\} 
  \leq  0.49 \frac{\mathsf{KL}_{\min}}{\mathsf{KL}_{\max}},
\end{align*}
precluding the possibility that $\bm{z}^{(t+1)}$ leaves the basin.
As a result, invoking the preceding theorem iteratively we arrive
at
\[
  \|\bm{z}^{(t)}-\bm{x}\|_{*,0}\leq\rho^{t}\|\bm{z}^{(0)}-\bm{x}\|_{*,0},
\]
indicating that the estimation error reduces to zero within at most
logarithmic iterations. 

\begin{remark}
In fact, the contraction rate $\rho$ can be as small
as $O\left(\frac{1}{np_{\mathrm{obs}}\mathsf{KL}_{\min}}\right)=O\left(\frac{1}{\log n}\right)$
in the scenario considered in Theorem \ref{thm:ExactRecovery-General}
or $O\Big(\max_{l}\left\Vert \log\frac{P_{0}}{P_{l}}\right\Vert _{1}^{2}/(np_{\mathrm{obs}}\mathsf{KL}_{\min}^{2})\Big)$
in the case studied in Theorem \ref{thm:ExactRecovery-General-large-m}.
\end{remark}

Furthermore, we emphasize that Theorem \ref{thm:Iterative-stage}
is a uniform result, namely, it holds simultaneously for all $\bm{z}$
within the basin, regardless of whether $\bm{z}$ is independent of
the data $\{y_{i,j}\}$ or not. Consequently, the theory and the analyses remain valid for other initialization schemes that can produce a suitable first guess. 

The rest of the section is thus devoted to establishing Theorem \ref{thm:Iterative-stage}.
The proofs for the two scenarios---the fixed $m$ case and the large
$m$ case---follow almost identical arguments, and hence we shall
merge the analyses.

\subsection{Analysis \label{sub:Proof-general}}

We outline the key steps for the proof of Theorem \ref{thm:Iterative-stage}.
Before continuing, it is helpful to introduce additional assumptions
and notation that will be used throughout. From now on, we will assume
$\bm{x}_{i}=\bm{e}_{1}$ for all $1\leq i\leq n$ without loss of
generality. We shall denote $\bm{h}=\left\{ \bm{h}_{i}\right\} {}_{1\leq i\leq n}\in\mathbb{R}^{nm}$
and $\bm{w}=\left\{ \bm{w}_{i}\right\} {}_{1\leq i\leq n}\in\mathbb{R}^{nm}$
as 
\begin{equation}
  \bm{h}:=\bm{z}-\bm{x}\qquad\text{and}\qquad\bm{w}:=\bm{L}\bm{z},
  \label{eq:defn-h-w}
\end{equation}
and set
\begin{eqnarray}
  k & := & \left\Vert \bm{h}\right\Vert _{*,0};  \label{eq:current-sparsity}\\
	\epsilon & := & \frac{\|\bm{h}\| }{ \|\bm{x}\| } = \frac{ \|\bm{h}\| }{ \sqrt{n} };  \label{eq:current-error}\\
  k^{*} & := & \min\left\{ \|\bm{h}\|_{*,0},\text{ }\frac{\|\bm{h}\|^{2}}{\|\bm{x}\|^{2}}\|\bm{x}\|_{*,0}\right\} =\min\left\{ k,\epsilon^{2}n\right\} .
  \label{eq:defn-k-star}
\end{eqnarray}

One of the key metrics that will play an important role in our proof
is the following separation measure
\begin{eqnarray}
  \mathscr{S}\left(\bm{a}\right) & := & \min _{2\leq l\leq m}\left(a_{1}-a_{l}\right)
  \label{eq:defn-sep}
\end{eqnarray}
defined for any vector $\bm{a}=[a_{l}]_{1\leq l\leq m}\in\mathbb{R}^{m}$.
This metric is important because, by Fact \ref{fact:Proj-separation},
the projection of a block $\mu\bm{w}_{i}$ onto the standard simplex
$\Delta$ returns the correct solution---that is, $\mathcal{P}_{\Delta}(\mu\bm{w}_{i})=\bm{e}_{1}$---as
long as $\mathscr{S}(\bm{w}_{i})>0$ and $\mu$ is sufficiently large.
As such, our aim is to show that the vector $\bm{w}$ given in (\ref{eq:defn-h-w})
obeys 
\begin{equation}
  \mathscr{S}\left(\bm{w}_{i}\right)>0.01np_{\mathrm{obs}}\mathsf{KL}_{\min}\qquad\forall i\in\mathcal{I}\subseteq\{1,\cdots,n\},
  \label{eq:sep-condition-simple}
\end{equation}
for some index set $\mathcal{I}$ of size $$|\mathcal{I}|\geq n-\rho k^{*},$$
where $0<\rho<1$ is bounded away from 1 (which will be specified
later). This taken collectively with Fact \ref{fact:Proj-separation}
implies $\mathcal{P}_{\Delta}\left(\mu\bm{w}_{i}\right)=\bm{x}_{i}=\bm{e}_{1}$
for every $i\in\mathcal{I}$ and, as a result,
\begin{align*}
	\|\mathcal{P}_{\Delta^{n}}\left(\mu\bm{w}\right)-\bm{x}\|_{*,0} &\leq \sum_{i\notin\mathcal{I}} \| \mathcal{P}_{\Delta}\left(\mu\bm{w}_{i}\right)-\bm{x}_{i}\|_{0} = n-|\mathcal{I}| \nonumber\\
	& \leq \rho k^{*} = \rho \min \left\{ \|\bm{h}\|_{*,0},~\frac{\|\bm{h}\|^2}{\|\bm{x}\|^2}\|\bm{x}\|_{*,0} \right\} ,
\end{align*}
provided that the scaling factor obeys $\mu>100/(np_{\mathrm{obs}}\mathsf{KL}_{\min})$. 

We will organize the proof of the claim (\ref{eq:sep-condition-simple})
based on the size / block sparsity of $\bm{h}$, leaving us with two
separate regimes to deal with:
\begin{itemize}
\item The \emph{large-error regime} in which
\begin{equation}
  \xi < \min\left\{ \frac{k}{n},\epsilon\right\} \leq 0.49 \frac{\mathsf{KL}_{\min}}{\mathsf{KL}_{\max}};
  \label{eq:large-error-regime}
\end{equation}

\item The \emph{small-error regime} in which
\begin{equation}
  \min\left\{ \frac{k}{n},\epsilon\right\} \leq\xi.
  \label{eq:small-error-regime}
\end{equation}

\end{itemize}
Here, one can take $\xi>0$ to be any (small) positive constant independent of $n$. In what follows, the input matrix $\bm{L}$ takes either the original form (\ref{eq:Input-Matrix}) or the debiased form (\ref{eq:debiased-MLE}). The version (\ref{eq:L-random-corruption}) tailored to the random corruption model will be discussed in Section \ref{sub:Consequences-RCM}. 

\vspace{0.5em}

(1)\textbf{ Large-error regime.} Suppose that $\bm{z}$ falls within
the regime (\ref{eq:large-error-regime}). In order to control $\mathscr{S}(\bm{w}_{i})$,
we decompose $\bm{w}=\bm{L}\bm{z}$ into a few terms that are easier
to work with. Specifically, setting 
\begin{align}
	\overline{\bm{h}}:=\frac{1}{n}\sum_{i=1}^{n}\bm{h}_{i} \qquad \text{and} \qquad \tilde{\bm{L}}:=\bm{L}-\mathbb{E}\left[\bm{L}\right], 
\end{align}
we can expand
\begin{eqnarray}
  \bm{w}  =  \bm{L}\bm{z} = \left(\mathbb{E}\left[\bm{L}\right]+\tilde{\bm{L}}\right)\left(\bm{x}+\bm{h}\right)
  = \underset{:=\bm{t}}{\underbrace{\mathbb{E}\left[\bm{L}\right]\bm{x}}} 
  + \underset{:=\bm{r}}{\underbrace{\mathbb{E}\left[\bm{L}\right]\bm{h}+\tilde{\bm{L}}\bm{x}+\tilde{\bm{L}}\bm{h}}}.
  \label{eq:Lz-general}
\end{eqnarray}
This allows us to lower bound the separation for the $i$th component
by
\begin{align}
	\mathscr{S}\left(\bm{w}_{i}\right) & = \mathscr{S}\left(\bm{t}_{i}+\bm{r}_{i}\right) 
  \geq \mathscr{S}\left(\bm{t}_{i}\right)+\mathscr{S}\left(\bm{r}_{i}\right)  \nonumber\\
  &= \mathscr{S}\left(\bm{t}_{i}\right)+\min_{2\leq l\leq m}\left(r_{i,1}-r_{i,l}\right)
  \geq \mathscr{S}\left(\bm{t}_{i}\right)-2\|\bm{r}_{i}\|_{\infty}.
  \label{eq:Sep_w-KL}
\end{align}
With this in mind, attention naturally turns to controlling $\mathscr{S}\left(\bm{t}_{i}\right)$
and $\|\bm{r}_{i}\|_{\infty}$. 

The first quantity $\mathscr{S}\left(\bm{t}_{i}\right)$ admits a
closed-form expression. From (\ref{eq:defn-K-intuition}) and (\ref{eq:Lz-general})
one sees that
\[
  \bm{t}_{i}=p_{\mathrm{obs}}\left(\sum\nolimits _{j:j\neq i}\bm{K}\bm{x}_{j}\right) 
  = p_{\mathrm{obs}}\left(n-1\right)\bm{K}\bm{e}_{1}.
\]
It is self-evident that $\mathscr{S}\left(\bm{K}\bm{e}_{1}\right)=\mathsf{KL}_{\min}$,
giving the formula
\begin{equation}
  \mathscr{S}\left(\bm{t}_{i}\right) = p_{\mathrm{obs}}\left(n-1\right)\mathscr{S}\left(\bm{K}\bm{e}_{1}\right)
  = p_{\mathrm{obs}}\left(n-1\right)\mathsf{KL}_{\min}.
  \label{eq:Sep-ti}
\end{equation}

We are now faced with the problem of estimating $\|\bm{r}_{i}\|_{\infty}$. To this end, we
make the following observation, which holds uniformly over all $\bm{z}$
residing within this regime: 

\begin{lem}
\label{lemma:r-norm-KL}
Consider the regime (\ref{eq:large-error-regime}).
Suppose  $m=O(\mathrm{poly}(n))$, $p_{\mathrm{obs}}\gtrsim\log n/n$, and 
\begin{equation}
  \frac{\mathsf{KL}_{\max}^{2}}{\left\{ \frac{1}{m}\sum_{l=1}^{m-1}\left\Vert \log\frac{P_{0}}{P_{l}}\right\Vert _{1}\right\} ^{2}} > \frac{c_{2}}{np_{\mathrm{obs}}}
  \label{eq:KLmax-general-m}
\end{equation}
for some sufficiently large constant $c_{2}>0$. With probability
exceeding $1-O(n^{-10})$, the index set 
\begin{equation}
  \mathcal{I} := \left\{ 1\leq i\leq n\mbox{\text{ }}\Big|\mbox{\text{ }}\|\bm{r}_{i}\|_{\infty}\leq np_{\mathrm{obs}}\mathsf{KL}_{\max}\min\left\{ \frac{k}{n},\epsilon\right\} + \alpha np_{\mathrm{obs}}\mathsf{KL}_{\max}\right\} 
  \label{eq:ri-bound-KL}
\end{equation}
has cardinality exceeding $n-\rho k^{*}$ for some $0<\rho<1$ bounded
away from 1, where $k^{*}=\min\left\{ k,\epsilon^{2}n\right\} $ and
$\alpha>0$ is some arbitrarily small constant.

In particular, if $m$ is fixed and if Assumption \ref{assumption-Pmin}
holds, then (\ref{eq:KLmax-general-m}) can be replaced by
\begin{equation}
  \mathsf{KL}_{\max}>\frac{c_{4}}{np_{\mathrm{obs}}}
  \label{eq:KLmax-fixed-m}
\end{equation}
for some sufficiently large constant $c_{4}>0$. \end{lem}\begin{proof}See
Appendix \ref{sec:Proof-of-Lemma-r-norm-KL}.\end{proof}

Combining Lemma \ref{lemma:r-norm-KL} with the preceding bounds (\ref{eq:Sep_w-KL})
and (\ref{eq:Sep-ti}), we obtain 
\begin{eqnarray}
  \mathscr{S}\left(\bm{w}_{i}\right) 
  & \geq & 
  \left(n-1\right)p_{\mathrm{obs}}\mathsf{KL}_{\min}
  - 2np_{\mathrm{obs}}\mathsf{KL}_{\max}\min\left\{ \frac{k}{n},\epsilon\right\} 
  - 2\alpha np_{\mathrm{obs}}\mathsf{KL}_{\max}\\
  & > & 0.01np_{\mathrm{obs}}\mathsf{KL}_{\min}
\end{eqnarray}
for all $i\in\mathcal{I}$ as given in (\ref{eq:ri-bound-KL}), provided
that (i) $\min\left\{ \frac{k}{n},\epsilon\right\} \leq0.49\frac{\mathsf{KL}_{\min}}{\mathsf{KL}_{\max}}$,
(ii) $\mathsf{KL}_{\max}/\mathsf{KL}_{\min}$ is bounded, (iii) $\alpha$
is sufficiently small, and (iv) $n$ is sufficiently large. This concludes
the treatment for the large-error regime. 

\vspace{0.5em}

(2) \textbf{Small-error regime}. 
We now turn to the second regime
obeying (\ref{eq:small-error-regime}). Similarly, we find it convenient
to decompose $\bm{w}$ as
\begin{equation}
  \bm{w} = \bm{L}\bm{x}+\bm{L}\bm{h} = \underset{:=\bm{s}}{\underbrace{\bm{L}\bm{x}}}+\underset{:=\bm{q}}{\underbrace{\mathbb{E}\left[\bm{L}\right]\bm{h}+\tilde{\bm{L}}\bm{h}}} .
  \label{eq:decompose-w-s}
\end{equation}
We then lower bound the separation measure by controlling $\bm{s}_{i}$
and $\bm{q}_{i}$ separately, i.e. 
\begin{equation}
  \mathscr{S}\left(\bm{w}_{i}\right)\geq\mathscr{S}\left(\bm{s}_{i}\right) + \mathscr{S}\left(\bm{q}_{i}\right)\geq\mathscr{S}\left(\bm{s}_{i}\right)-2\left\Vert \bm{q}_{i}\right\Vert _{\infty}.
  \label{eq:sep-w-general}
\end{equation}

We start by obtaining uniform control over the separation of all components
of $\bm{s}$:

\begin{lem}\label{lem:sep-s}Suppose that Assumption \ref{assumption-Pmin}
holds and that $p_{\mathrm{obs}}>c_{0}\log n/n$ for some sufficiently
large constant $c_{0}>0$ . 

(1) Fix $m>0$, and let $\zeta>0$ be any sufficiently small constant.
Under Condition (\ref{eq:IT-sufficient-KL}), one has
\begin{equation}
  \mathscr{S}\left(\bm{s}_{i}\right)>\zeta np_{\mathrm{obs}}\mathsf{KL}_{\min},
  \qquad1\leq i\leq n
  \label{eq:sep-si-1}
\end{equation}
with probability exceeding $1-C_{6}\exp\left\{ -c_{6}\zeta\log(nm)\right\} -c_{7}n^{-10}$,
where $C_{6},c_{6},c_{7}>0$ are some absolute constants. 

(2) There exist some constants $c_{4},c_{5},c_{6}>0$ such that
\begin{equation}
  \mathscr{S}\left(\bm{s}_{i}\right)>c_{4}np_{\mathrm{obs}}\mathsf{KL}_{\min},
  \qquad1\leq i\leq n
  \label{eq:sep-si-1-1}
\end{equation}
with probability $1-O\left(m^{-10}n^{-10}\right)$, provided
that
\begin{equation}
  \frac{\mathsf{KL}_{\min}^{2}}{\max_{0\leq l<m}{\mathsf{Var}}_{y\sim P_{0}}
  \left[ \log\frac{P_{0}\left(y\right)}{P_{l}\left(y\right)} \right]} \geq \frac{c_{5}\log\left(mn\right)}{np_{\mathrm{obs}}}
  \quad\text{and}\quad
  \mathsf{KL}_{\min}\geq\frac{c_{6}\left\{ \max_{l,y}\left|\log\frac{P_{0}\left(y\right)}{P_{l}\left(y\right)}\right|\right\} \log\left(mn\right)}{np_{\mathrm{obs}}}.
  \label{eq:sep-si-1-2}
\end{equation}
\end{lem}

\begin{proof}
See Appendix \ref{sec:Proof-of-Lemma-Chernoff-UB}.
\end{proof}

The next step comes down to controlling $\|\bm{q}_{i}\|_{\infty}$.
This can be accomplished using similar argument as for Lemma \ref{lemma:r-norm-KL},
as summarized below. 

\begin{lem}
\label{lemma:r-norm-KL-q}
Consider the regime (\ref{eq:small-error-regime}).
Then Lemma \ref{lemma:r-norm-KL} continues to hold if $\bm{r}$ is
replaced by $\bm{q}$. 
\end{lem} 
\begin{remark}
  Notably, Lemma \ref{lemma:r-norm-KL-q} does not rely on the definition of the small-error regime. 
\end{remark}

Putting the inequality (\ref{eq:sep-w-general})
and Lemma \ref{lemma:r-norm-KL-q} together yields
\begin{align}
  \mathscr{S}\left(\bm{w}_{i}\right)
  & \geq \mathscr{S}\left(\bm{s}_{i}\right)- 2np_{\mathrm{obs}}\mathsf{KL}_{\max}\min\left\{ \frac{k}{n},\epsilon\right\} -2\alpha np_{\mathrm{obs}}\mathsf{KL}_{\max}  \\
	& \geq \mathscr{S}\left(\bm{s}_{i}\right) - 2(\xi+\alpha) np_{\mathrm{obs}}\mathsf{KL}_{\max}	\label{eq:sep-wi-LB-0}  \\ 
	& \geq \mathscr{S}\left(\bm{s}_{i}\right) - \left\{ 2(\xi+\alpha) \frac{\mathsf{KL}_{\max}}{\mathsf{KL}_{\min}} \right\} np_{\mathrm{obs}}\mathsf{KL}_{\min} 
  \label{eq:sep-wi-LB-1}
\end{align}
for all $i\in\mathcal{I}$ with high probability, where (\ref{eq:sep-wi-LB-0}) follows from the definition of the small-error regime. 
Recall that ${\mathsf{KL}_{\max}}/{\mathsf{KL}_{\min}}$ is bounded according to
Assumption \ref{assumption-KL-D}. 
Picking ${\xi}$ and $\alpha$ to be sufficiently
small constants and applying Lemma \ref{lem:sep-s}, we arrive at (\ref{eq:sep-condition-simple}). 

To summarize, we have established the claim (\ref{eq:sep-condition-simple})---and
hence the error contraction---as long as (a) $m$ is fixed and Condition
(\ref{eq:IT-sufficient-KL}) is satisfied, or (b) the conditions (\ref{eq:KLmax-general-m})
and (\ref{eq:sep-si-1-2}) hold. Interestingly, one can 
simplify Case (b) when $p_{\mathrm{obs}}\gtrsim\log^{5}n/n$, leading
to a matching condition to Theorem \ref{thm:ExactRecovery-General-large-m}.

\begin{lem}
\label{lem:simple-condition-general} 
Suppose $m \gtrsim \log n$, $m = \mathrm{poly}(n)$, and $p_{\mathrm{obs}}\geq c_{6}\log^{5}n/n$
for some sufficiently large constant $c_{6}>0$. The inequalities
(\ref{eq:KLmax-general-m}) and (\ref{eq:sep-si-1-2}) hold under
Condition (\ref{eq:SNR-condition-general-m}) in addition to Assumptions
\ref{assumption-Pmin}-\ref{assumption-KL-D}.
\end{lem}

\begin{proof} See Appendix \ref{sec:Proof-of-Lemma-simple-condition}. \end{proof}

\subsection{Choice of the scaling factor $\mu$}

So far we have proved the result under the scaling factor condition (\ref{eq:mu-KL}) given in Theorem \ref{thm:Iterative-stage}. To conclude the analysis for Theorem \ref{thm:ExactRecovery-General} and Theorem \ref{thm:ExactRecovery-General-large-m}, it remains to convert it to conditions in terms of the singular value $\sigma_i(\cdot)$. 

To begin with, it follows from (\ref{eq:EL-general-block-intuition}) that 
\[
  \bm{L} = \mathbb{E}\left[\bm{L}\right] + (\bm{L} - \mathbb{E}[\bm{L}]) 
  = p_{\mathrm{obs}} \bm{1} \bm{1}^{\top} \otimes \bm{K} - p_{\mathrm{obs}} \bm{I}_n \otimes \bm{K} + (\bm{L} - \mathbb{E}[\bm{L}]),
\]
leading to an upper estimate
  \begin{align}
     \sigma_i( \bm{L} ) 
     & ~\leq~  \sigma_i( p_{\mathrm{obs}} \bm{1} \bm{1}^{\top} \otimes \bm{K} ) 
  	+ \| p_{\mathrm{obs}} \bm{I}_n \otimes \bm{K} \| + \| \bm{L} - \mathbb{E}[\bm{L}] \|  \nonumber \\
     & ~=~ np_{\mathrm{obs}} \sigma_i( \bm{K} ) + p_{\mathrm{obs}} \| \bm{K} \| + \| \bm{L} - \mathbb{E}[\bm{L}] \|.
     \label{eq:LB-sigma_m1}
  \end{align}
Since $\bm{K}$ is circulant, its eigenvalues are given by
\[
  \lambda_{l} = \sum\nolimits _{i=0}^{m-1}\left(-\mathsf{KL}_{i}-\mathcal{H}(P_{0})\right) \exp\left(j \frac{2\pi il}{m}\right),
  \quad 0\leq l <m
\]
with $j = \sqrt{-1}$. In fact, except for $\lambda_{0}$, one can simplify
\[
\lambda_{l}= -\sum\nolimits _{i=0}^{m-1} \mathsf{KL}_{i}  \exp\left(j \frac{2\pi il}{m}\right), \quad 1\leq l<m,
\]
which are eigenvalues of $\bm{K}^{0}$ (see \eqref{eq:defn-K-intuition}) as well. 
This leads to the upper bounds
\begin{align}
  \sigma_{2}\left(\bm{K}\right) &\leq \sqrt{\sum\nolimits _{j=1}^{m-1}\lambda_{j}^{2}}\leq\|\bm{K}^{0}\|_{\mathrm{F}}
   \leq m\mathsf{KL}_{\max}  \lesssim  m\mathsf{KL}_{\min}, 
   \label{eq:LB_sigma2_K}\\
  \sigma_{m}\left(\bm{K}\right) &\leq  \sqrt{\frac{1}{m-1}\sum\nolimits_{j=1}^{m-1}\lambda_{j}^{2}}\leq\sqrt{\frac{1}{m-1}\|\bm{K}^{0}\|_{\mathrm{F}}^{2}}
  \leq \sqrt{\frac{m^{2} \mathsf{KL}_{\max}^{2} }{ m-1 }}  \nonumber\\
   & \leq  \sqrt{2m}\mathsf{KL}_{\max}
  \lesssim  \sqrt{m}\mathsf{KL}_{\min},
  \label{eq:LB_sigma_K}
\end{align}
where both (\ref{eq:LB_sigma2_K}) and (\ref{eq:LB_sigma_K}) follow from Assumption \ref{assumption-KL-D}. 
In addition, it is immediate to see that (\ref{eq:LB-sigma_m1}) and (\ref{eq:LB_sigma_K}) remain valid if we replace $\bm{L}$ with $\bm{L}^{\mathrm{debias}}$ and take $\bm{K} = \frac{1}{p_{\mathrm{obs}}} \mathbb{E}[ \bm{L}_{i,j}^{\mathrm{debias}} ]$ ($i\neq j$) instead.

To bound the remaining terms on the right-hand side of (\ref{eq:LB-sigma_m1}), we divide into two separate cases:
\begin{itemize}
  \item[(i)] When $m$ is fixed,  
it follows from (\ref{eq:defn-K-intuition}) that 
\[
	\left\Vert \bm{K}\right\Vert ~
	\leq~ \|\bm{K}_{0}\| + m \mathcal{H}\left(P_{0}\right)
 	~\leq~ m\mathsf{KL}_{\max} + m \log m
  ~=O(1).
\]
When combined with Lemma \ref{lem:deviation-Lhat}, this yields
\begin{align}
  p_{\mathrm{obs}}\|\bm{K}\| + \|\tilde{\bm{L}}\|
	&\lesssim p_{\mathrm{obs}} + \left(\frac{1}{m}\sum\nolimits _{l=1}^{m-1}\left\Vert \log\frac{P_{0}}{P_{l}}\right\Vert _{1}\right)\sqrt{np_{\mathrm{obs}}} \nonumber\\
	&\lesssim p_{\mathrm{obs}} + \sqrt{np_{\mathrm{obs}}\mathsf{KL}_{\max}},
  \label{eq:gamma-fixed-m}
\end{align}
where the last inequality follows from (\ref{eq:log1-KL}). Putting this together with (\ref{eq:LB-sigma_m1}) and (\ref{eq:LB_sigma2_K}) and using Assumption \ref{assumption-KL-D}, we get
\begin{align}
  \sigma_{2} (\bm{L})
  \lesssim np_{\mathrm{obs}} m \mathsf{KL}_{\min}+p_{\mathrm{obs}}+\sqrt{np_{\mathrm{obs}}\mathsf{KL}_{\max}}
  \asymp np_{\mathrm{obs}}\mathsf{KL}_{\min}.
\end{align}
Thus, one would satisfy (\ref{eq:mu-KL}) by taking $\mu \geq c_{12}/ \sigma_{2}\left(\bm{L}\right)$ for some sufficiently large $c_{12}>0$.
 
 \item[(ii)] When $m = O(\mathrm{poly}(n))$, 
we consider $\bm{L}^{\mathrm{debias}}$ and set $\bm{K} = \frac{1}{p_{\mathrm{obs}}}\mathbb{E} [ \bm{L}_{i,j}^{\mathrm{debias}}]$. 
According to (\ref{eq:debias-block-UB}), one has 
$\left\Vert \bm{L}_{i,j}^{\mathrm{debias}}\right\Vert \leq\frac{1}{m}\sum_{l=1}^{m-1}\left\Vert \log\frac{P_{0}}{P_{l}}\right\Vert _{1}$, 
thus indicating that
\begin{align}
  p_{\mathrm{obs}} \|\bm{K}\| =  \left\Vert  \mathbb{E} \left[ \bm{L}_{i,j}^{\mathrm{debias}} \right] \right\Vert 
 ~\leq~  \frac{1}{m} \sum\nolimits_{l=1}^{m-1}\left\Vert \log\frac{P_{0}}{P_{l}}\right\Vert _{1}.
\end{align}
This together with Lemma \ref{lem:deviation-Lhat} gives
\begin{align}
  & p_{\mathrm{obs}}\|\bm{K}\| + \| \bm{L}^{\mathrm{debias}} - \mathbb{E}[\bm{L}^{\mathrm{debias}}] \|
  ~\lesssim \left(\frac{1}{m}\sum\nolimits _{l=1}^{m-1}\left\Vert \log\frac{P_{0}}{P_{l}}\right\Vert _{1}\right) \left(1+\sqrt{np_{\mathrm{obs}}}\right) \\
  & \quad \asymp \left(\frac{1}{m}\sum\nolimits _{l=1}^{m-1}\left\Vert \log\frac{P_{0}}{P_{l}}\right\Vert _{1}\right)\sqrt{np_{\mathrm{obs}}}
   ~\lesssim np_{\mathrm{obs}}\mathsf{KL}_{\min}
  \label{eq:gamma-general-m}
\end{align}
under the condition (\ref{eq:SNR-condition-general-m}).
Combine all of this to derive
\[
   \sigma_{m}\left( \bm{L}^{\mathrm{debias}} \right)
   ~\lesssim  \sqrt{m}np_{\mathrm{obs}}\mathsf{KL}_{\min}
   + np_{\mathrm{obs}} \mathsf{KL}_{\min}
   ~\asymp  \sqrt{m}np_{\mathrm{obs}}\mathsf{KL}_{\min},
\]
thus justifying (\ref{eq:mu-KL}) as long as $\mu \geq c_{12} \sqrt{m} / \sigma_{m}\left(\bm{L}^{\mathrm{debias}} \right)$ for some sufficiently large constant $c_{12}>0$.
\end{itemize}

\subsection{Consequences for random corruption models\label{sub:Consequences-RCM}}

Having obtained the qualitative behavior of the iterative stage for
general models, we can now specialize it
to the random corruption model (\ref{eq:random-outlier}). Before continuing, it is straightforward to compute two metrics:
\begin{align*}
	\mathsf{KL}_{\min} = \mathsf{KL}_{\max} &= \left(\pi_{0}+\frac{1-\pi_{0}}{m}\right)\log\frac{\pi_{0}+\frac{1-\pi_{0}}{m}}{\frac{1-\pi_{0}}{m}}+\frac{1-\pi_{0}}{m}\log\frac{\frac{1-\pi_{0}}{m}}{\pi_{0}+\frac{1-\pi_{0}}{m}} \nonumber\\
  &=\pi_{0}\log\frac{1+\left(m-1\right)\pi_{0}}{1-\pi_{0}};
\end{align*}
\[
  \text{and} \qquad\quad
  \left\Vert \log\frac{P_{0}}{P_{l}}\right\Vert _{1} = 2 \left|\log\frac{\pi_{0}+\frac{1-\pi_{0}}{m}}{\frac{1-\pi_{0}}{m}}\right|
  = 2 \log\frac{1+\left(m-1\right)\pi_{0}}{1-\pi_{0}}.
\]

\begin{enumerate}
\item When $m$ is fixed and $\pi_{0}$ is small, it is not hard to see
that 
\[
  \mathsf{KL}_{\min} = \mathsf{KL}_{\max}\approx m\pi_{0}^{2},
\]
which taken collectively with (\ref{eq:IT-sufficient-KL}) leads
to (\ref{eq:IT-sufficient-dense}). 
\item When $m\gtrsim\log n$ and $m=O\left(\mathrm{poly}\left(n\right)\right)$,
the condition (\ref{eq:SNR-condition-general-m}) reduces to $\pi_{0}^{2}\gtrsim1/(np_{\mathrm{obs}})$,
which coincides with (\ref{eq:IT-sufficient-dense-large-m}). In fact,
one can also easily verify (\ref{eq:sep-si-1-2}) under Condition
(\ref{eq:IT-sufficient-dense-large-m}), assuming that $p_{\mathrm{obs}}\gtrsim\log^{2}n/n$.
This improves slightly upon the condition $p_{\mathrm{obs}}\gtrsim\log^{5}n/n$
required in the general theorem. 
\end{enumerate}
Next, we demonstrate that the algorithm with the input
matrix (\ref{eq:L-random-corruption}) undergoes the same trajectory
as the version using (\ref{eq:Input-Matrix}). To avoid confusion,
we shall let $\bm{L}^{\mathrm{rcm}}$  denote the matrix (\ref{eq:L-random-corruption}),
and set $$\bm{w}^{\mathrm{rcm}}:=\bm{L}^{\mathrm{rcm}}\bm{z}.$$ As
discussed before, there are some constants $a>0$ and $b$ such that
$\bm{L}_{i,j}^{\mathrm{rcm}}=a\bm{L}_{i,j}+b\bm{1}\bm{1}^{\top}$
for all $(i,j)\in\Omega,$ indicating that
\[
  \bm{w}_{i}^{\mathrm{rcm}}=a\bm{w}_{i}+\tilde{b}_{i}\bm{1}, \quad 1 \leq i\leq n
\]
for some numerical values $\{\tilde{b}_{i}\}$. In view of Fact \ref{fact:Proj-const-shift},
the projection $\mathcal{P}_{\Delta}$ remains unchanged up to global
shift. This justifies the equivalence between the two input matrices
when running Algorithm \ref{alg:PP}. 

Finally, one would have to adjust the scaling factor accordingly. It is straightforward to show that the scaling factor condition (\ref{eq:mu-KL}) can be translated into $\mu > c_5 / (np_{\mathrm{obs}}\pi_0)$ when the input
matrix (\ref{eq:L-random-corruption}) is employed. Observe that (\ref{eq:LB-sigma_m1}) continues to hold as long as we set $\bm{K}=\pi_{0}\bm{I}_{m}+\left(1-\pi_{0}\right)\bm{1}\bm{1}^{\top}$. We can also verify that
\begin{equation}
	\sigma_2( \bm{K} ) = \pi_0,  \quad \| \bm{K} \| \leq 1, \quad \text{and} \quad \| \bm{L} - \mathbb{E}[ \bm{L} ]  \| \lesssim \sqrt{np_{\mathrm{obs}}},
  \label{eq:K-L-RCM}
\end{equation}
where the last inequality follows from Lemma \ref{lem:random-matrix-norm}. These taken collectively with (\ref{eq:LB-sigma_m1}) lead to
\begin{equation}
  \sigma_{2}(\bm{L})\lesssim np_{\mathrm{obs}}\pi_{0}+p_{\mathrm{obs}}\pi_{0}+\sqrt{np_{\mathrm{obs}}}\lesssim np_{\mathrm{obs}}\pi_{0}
\end{equation}
 under the condition (\ref{eq:IT-sufficient-dense-large-m}). This
justifies the choice  $\mu\gtrsim1/\sigma_{2}(\bm{L})$ as advertised.

\section{Spectral initialization\label{sec:Spectral-Initialization}}

We come back to assess the performance of spectral initialization by
establishing the theorem below. Similar to the definition (\ref{eq:mcr}),
we introduce the counterpart of $\ell_{2}$ distance modulo the global offset as
\[
  \mathrm{dist}(\hat{\bm{x}},\bm{x}) := \min_{0\leq l<m} \left\Vert \hat{\bm{x}} - \mathsf{shift}_{l}\left(\bm{x}\right) \right\Vert .
\]

\begin{theorem} 
\label{thm:SpectralMethod} 
Fix $\delta>0$,
and suppose that $p_{\mathrm{obs}}\gtrsim \log n/n$. Under Assumptions \ref{assumption-Pmin}-\ref{assumption-KL-D}, there are some universal constants $c_{1},c_{2},c_{3}>0$
such that with probability at least $1-\xi$, the initial estimate
$\bm{z}^{(0)}$ of Algorithms \ref{alg:PP} obeys 
\begin{equation}
  \mathrm{dist}(\bm{z}^{(0)},\bm{x}) \leq \delta\|\bm{x}\| 
  \qquad \text{and} \qquad 
  \mcr(\bm{z}^{(0)}, \bm{x} ) \leq \delta^2 / 2
  \label{eq:initialization-bound}
\end{equation}
in the following scenarios:
\begin{enumerate}
\item[(i)] the random corruption model with $m=O(\mathrm{poly}(n))$, provided
that $\bm{L}$ is given by (\ref{eq:L-random-corruption}) and that
\begin{equation}
  \pi_{0} \geq \frac{c_{1}}{\delta\sqrt{\xi}}\frac{1}{\sqrt{p_{\mathrm{obs}}n}};
  \label{eq:pi0-spectral}
\end{equation}

\item[(iii)] the general model with a fixed $m$, provided that $\bm{L}$ is given
by (\ref{eq:Input-Matrix}) and that
\begin{equation}
  \mathsf{KL}_{\min} \geq \frac{c_{2}}{\delta^{2}\xi}\frac{1}{p_{\mathrm{obs}}n};
  \label{eq:condition-KL-spectral-1}
\end{equation}
\item[(iii)] the general model with $m=O(\mathrm{poly}(n))$, provided that $\bm{L}$
is replaced by $\bm{L}^{\mathrm{debias}}$ (given in (\ref{eq:debiased-MLE}))
and that
\begin{equation}
  \frac{\mathsf{KL}_{\min}^{2}}{\max_{1\leq l<m}\left\Vert \log\frac{P_{0}}{P_{l}}\right\Vert _{1}^{2}}
  \geq  \frac{c_{3}}{\delta^{2}\xi} \frac{1}{p_{\mathrm{obs}}n}.
  \label{eq:condition-KL-spectral}
\end{equation}
\end{enumerate}
\end{theorem}

The main reason for the success of spectral initialization  is that the
low-rank approximation of $\bm{L}$ (resp.~$\bm{L}^{\mathrm{debias}}$)
produce a decent estimate of $\mathbb{E}\left[\bm{L}\right]$ (resp.~$\mathbb{E}[\bm{L}^{\mathrm{debias}}]$)
and, as discussed before, $\mathbb{E}\left[\bm{L}\right]$ (resp.~$\mathbb{E}[\bm{L}^{\mathrm{debias}}]$)
reveals the structure of the truth. In what follows, we will first prove the result for  general $\bm{L}$, and then specialize
it to the three choices considered in the theorem. As usual, we suppose
without loss of generality that $\bm{x}_{i}=\bm{e}_{1}$, $1\leq i\leq n$. 

To begin with, we
set 
$  \bm{K} = \frac{1}{p_{\mathrm{obs}}} \mathbb{E}[\bm{L}_{i,j}]$ ($i \neq j$) as before
and write  
\begin{eqnarray}
  \bm{L}  =   p_{\mathrm{obs}}\bm{1}\bm{1}^{\top}\otimes\bm{K} - p_{\mathrm{obs}}\bm{I}_{n}\otimes\bm{K} + (\bm{L} - \mathbb{E}[\bm{L}]),
  \label{eq:L-expansion}
\end{eqnarray}
where the first term on the right-hand side of (\ref{eq:L-expansion})
has rank at most $m$. If we let $\hat{\bm{L}}$ be the best rank-$m$
approximation of $\bm{L}$, then matrix perturbation theory gives
\[
  \| \hat{\bm{L}}- \bm{L}\| \leq \| p_{\mathrm{obs}}\bm{1}\bm{1}^{\top}\otimes\bm{K} - \bm{L}\|.
\]
Hence, the triangle inequality yields
\begin{align}
	\| \hat{\bm{L}}-p_{\mathrm{obs}}\bm{1}\bm{1}^{\top}\otimes\bm{K}\|   &\leq~ \| \hat{\bm{L}}- \bm{L}\| + \| \bm{L}-p_{\mathrm{obs}}\bm{1}\bm{1}^{\top}\otimes\bm{K}\| \nonumber\\
	&\leq~ 2  \| \bm{L}-p_{\mathrm{obs}}\bm{1}\bm{1}^{\top}\otimes\bm{K}\|  \nonumber\\
	&\overset{\text{(i)}}{\leq}~ 2\|p_{\mathrm{obs}}\bm{I}_{n}\otimes\bm{K}\| + 2\| \bm{L} - \mathbb{E}[\bm{L}] \|  \nonumber\\
	&=~ 2p_{\mathrm{obs}}\|\bm{K}\| + 2\| \bm{L} - \mathbb{E}[\bm{L}] \| 
  ~:=~\gamma_{\bm{L}},
  \label{defn:gamma_L}
\end{align}
where (i) follows from (\ref{eq:L-expansion}). 
This together with the facts $\mathrm{rank}(\hat{\bm{L}})\leq m$
and $\mathrm{rank}(\bm{K})\leq m$ gives 
\[
  \|\hat{\bm{L}}-p_{\mathrm{obs}}\bm{1}\bm{1}^{\top}\otimes\bm{K}\|_{\mathrm{F}}^{2}\leq2m\|\hat{\bm{L}}-p_{\mathrm{obs}}\bm{1}\bm{1}^{\top}\otimes\bm{K}\|^{2}
  \leq 2m \gamma_{\bm{L}}^{2}.
\]

Further, let $\bm{K}_{:,1}$ (resp.~$\bm{1}\otimes\bm{K}_{:,1}$) be the first
column of $\bm{K}$ (resp.~$\bm{1}\otimes\bm{K}$). When $\hat{\bm{z}}$
is taken to be a random column of $\hat{\bm{L}}$, it is straightforward
to verify that 
\begin{align*}
  \mathbb{E}\left[\mathrm{dist}^{2}\left(\hat{\bm{z}},~p_{\mathrm{obs}}\bm{1}\otimes\bm{K}_{:,1}\right)\right] 
  & \leq \frac{1}{nm} \|\hat{\bm{L}}-p_{\mathrm{obs}}\bm{1}\bm{1}^{\top}\otimes\bm{K}\|_{\mathrm{F}}^{2} 
  \leq \frac{2\gamma_{\bm{L}}^{2}}{n},
\end{align*}
where the expectation is w.r.t.~the randomness in picking the column (see Section \ref{sec:Algorithm}). 
Apply Markov's inequality to deduce that, with probability at least
$1-\xi$, 
\begin{equation}
  \mathrm{dist}\left(\hat{\bm{z}},~p_{\mathrm{obs}}\bm{1}\otimes\bm{K}_{:,1}\right)
  \leq \frac{\sqrt{2}\gamma_{\bm{L}}}{\sqrt{\xi n}}.
  \label{eq:dist-gamma}
\end{equation}

For simplicity of presentation, we shall assume 
\[
  \left\Vert \hat{\bm{z}}-p_{\mathrm{obs}}\bm{1}\otimes\bm{K}_{:,1}\right\Vert 
  = \mathrm{dist}\left(\hat{\bm{z}},~p_{\mathrm{obs}}\bm{1}\otimes\bm{K}_{:,1}\right)
\]
from now on. We shall pay particular attention to the index set 
\[
  \mathcal{J} := \left\{ i\in[n]~\Big|  
  \left\Vert \hat{\bm{z}}_{i}-p_{\mathrm{obs}}\bm{K}_{:,1} \right\Vert 
  \leq \frac{\sqrt{2}}{ \delta\sqrt{n} }  \left\Vert \hat{\bm{z}}-p_{\mathrm{obs}}\bm{1}\otimes\bm{K}_{:,1}\right\Vert 
  = \frac{\sqrt{2}}{\delta\sqrt{n}} \mathrm{dist} \left( \hat{\bm{z}}, ~p_{\mathrm{obs}}\bm{1} \otimes \bm{K}_{:,1}\right) \right\} ,
\]
which consists of all blocks whose estimation errors are not much larger than the average estimation error. 
It is easily seen that the set $\mathcal{J}$ satisfies 
\begin{equation}
	|\mathcal{J}|\geq (1-\delta^{2}/2)\,n \label{eq:J-size}
\end{equation}
and hence contains most blocks. This comes from the fact that
\begin{eqnarray*}
  \left\Vert \hat{\bm{z}}-p_{\mathrm{obs}}\bm{1}\otimes\bm{K}_{:,1}\right\Vert ^{2} 
   & \geq & \sum\nolimits _{i\notin\mathcal{\mathcal{J}}}\left\Vert \hat{\bm{z}}_{i}-p_{\mathrm{obs}}\bm{K}_{:,1}\right\Vert ^{2} \\
   & > & (n-|\mathcal{\mathcal{J}}|)\cdot\frac{2}{\delta^{2}n} \left\Vert \hat{\bm{z}}-p_{\mathrm{obs}}\bm{1}\otimes\bm{K}_{:,1} \right\Vert ^2.
\end{eqnarray*}
which can only happen if (\ref{eq:J-size}) holds. 
The $\ell_{\infty}$ error in each block can also be bounded by 
\begin{equation}
  \|\hat{\bm{z}}_{i}-p_{\mathrm{obs}}\bm{K}_{:,1}\|_{\infty} \leq \|\hat{\bm{z}}_{i}-p_{\mathrm{obs}}\bm{K}_{:,1}\|
  \leq \frac{\sqrt{2}}{\delta\sqrt{n}} \mathrm{dist}\left(\hat{\bm{z}},~p_{\mathrm{obs}}\bm{1} \otimes \bm{K}_{:,1}\right),
  \quad i\in\mathcal{\mathcal{J}}.
  \label{eq:zi-UB}
\end{equation}

If the above $\ell_{\infty}$ error is sufficiently small for each  $i\in \mathcal{J}$, then the projection operation recovers the truth for all blocks falling in $\mathcal{J}$. Specifically, 
 adopting the separation measure $\mathscr{S}(\cdot)$ as defined in
(\ref{eq:defn-sep}), we obtain
\begin{align*}
  \mathscr{S}(\hat{\bm{z}}_{i})
	&\geq  \mathscr{S}(p_{\mathrm{obs}}\bm{K}_{:,1})+\mathscr{S}(\hat{\bm{z}}_{i}-p_{\mathrm{obs}}\bm{K}_{:,1}) \\
	&\geq  p_{\mathrm{obs}}\mathscr{S}(\bm{K}_{:,1}) - 2\|\hat{\bm{z}}_{i}-p_{\mathrm{obs}}\bm{K}_{:,1}\|_{\infty}.
\end{align*}
If 
\begin{equation}
	\label{eq:z-infty-norm}
	\|\hat{\bm{z}}_{i}-p_{\mathrm{obs}}\bm{K}_{:,1}\|_{\infty}\leq c_{5}p_{\mathrm{obs}}\mathscr{S}\left(\bm{K}_{:,1}\right), \quad i\in\mathcal{J}
\end{equation}
for some constant $0<c_{5} < 1/2$, then it would follow from Fact \ref{fact:Proj-separation}
that
\[
  \bm{z}_{i}^{(0)}  = \mathcal{P}_{\Delta}\left(\mu_{0}\hat{\bm{z}}_{i}\right) = \bm{e}_{1}, \qquad i\in\mathcal{J}
\]
as long as 
$  \mu_{0} > \frac{1}{\left(1-2c_{5}\right) p_{\mathrm{obs}}\mathscr{S}\left(\bm{K}_{:,1}\right)}.$
This taken collectively with Fact \ref{fact:H} further reveals that
\begin{equation}
   \begin{cases}
	\mathrm{dist}(\bm{z}^{(0)}, \bm{x} ) ~\leq~ \sqrt{ 2} \cdot \sqrt{ n-|\mathcal{J}|  } ~\leq~ \delta \sqrt{n} ~=~ \delta \| \bm{x} \|  \\
	\mcr(\bm{z}^{(0)}, \bm{x} ) ~\leq~ ( n-|\mathcal{J}| ) / n  ~\leq~   \delta^2 / 2
   \end{cases}
\end{equation}
as claimed. 
As a result, everything boils down to proving \eqref{eq:z-infty-norm}. 
In view of (\ref{eq:dist-gamma}) and (\ref{eq:zi-UB}), this condition
(\ref{eq:z-infty-norm}) would hold if 
\begin{equation}
  \gamma_{\bm{L}} \leq c_{6} \delta\sqrt{\xi}np_{\mathrm{obs}}\mathscr{S}\left(\bm{K}_{:,1}\right)
  \label{eq:gamma_L-ub}
\end{equation}
for some sufficiently small constant $c_{6}>0$. 

To finish up, we establish (\ref{eq:gamma_L-ub}) for the three scenarios
considered Theorem \ref{thm:SpectralMethod}.
\begin{enumerate}
\item[(i)] \emph{The random corruption model with $\bm{L}$ given by (\ref{eq:L-random-corruption}).}
Simple calculation gives $\mathscr{S}\left(\bm{K}_{:,1}\right)=\pi_{0}$, and
it follows from (\ref{eq:K-L-RCM}) that
$ \gamma_{\bm{L}} \lesssim \sqrt{np_{\mathrm{obs}}} $.
Thus, the condition (\ref{eq:gamma_L-ub}) would hold under
the assumption (\ref{eq:pi0-spectral}).
\item[(ii)] \emph{The general model with fixed $m$ and with $\bm{L}$ given
by }(\ref{eq:Input-Matrix})\emph{.} Observe that $\mathscr{S}\left(\bm{K}_{:,1}\right) = \mathsf{KL}_{\min}$, 
and recall from (\ref{eq:gamma-fixed-m}) that
$
 \gamma_{\bm{L}} \lesssim~ p_{\mathrm{obs}}+\sqrt{np_{\mathrm{obs}}\mathsf{KL}_{\max}}.
$
Thus, we necessarily have (\ref{eq:gamma_L-ub}) under the condition (\ref{eq:condition-KL-spectral-1})  and Assumption \ref{assumption-KL-D}.

\item[(iii)] \emph{The general model with $m=O(\mathrm{poly}(n))$ and with $\bm{L}$
replaced by $\bm{L}^{\mathrm{debias}}$}, in which $\mathscr{S}\left(\bm{K}_{:,1}\right)=\mathsf{KL}_{\min}$.
It has been shown in (\ref{eq:gamma-general-m}) that
\[
  \gamma_{\bm{L}^{\mathrm{debias}}}
  \lesssim \left(\frac{1}{m}\sum\nolimits _{l=1}^{m-1}\left\Vert \log\frac{P_{0}}{P_{l}}\right\Vert _{1}\right)\sqrt{np_{\mathrm{obs}}}.
\]
As a consequence, we establish (\ref{eq:gamma_L-ub}) under the assumption
(\ref{eq:condition-KL-spectral}).
\end{enumerate}

Finally, repeating the analyses for the scaling factor in Section \ref{sec:Iterative-stage-general} justifies the choice of $\mu_0$ as suggested in the main theorems. This finishes the proof.


\section{Minimax lower bound\label{sec:Minimax-lower-bound}}

This section proves the minimax lower bound as claimed in Theorem \ref{thm:lower-bound}.
Once this is done, we can apply it to the random corruption model,
which immediately establishes Theorem \ref{thm:lower-bound-RCM} using
exactly the same calculation as in Section \ref{sub:Consequences-RCM}. 

To prove Theorem \ref{thm:lower-bound}, it suffices to analyze the maximum likelihood
(ML) rule, which minimizes the Bayesian error probability when we impose
a uniform prior over all possible inputs. Before continuing, we provide
an asymptotic estimate on the tail exponent of the likelihood ratio
test, which proves crucial in bounding the probability of error of
ML decoding. 

\begin{lem}
\label{lem:moderate-deviation}
Let $\left\{ P_{n}\right\} _{n\geq1}$
and $\left\{ Q_{n}\right\} _{n\geq1}$ be two sequences of probability
measures on a fixed finite set $\mathcal{Y}$, where  $\min_{n,y}P_{n}(y)$
and $\min_{n,y}Q_{n}(y)$ are both bounded away from 0. Let $\left\{ y_{j,n}:1\leq j\leq n\right\} _{n\geq1}$
be a triangular array of independent random variables such that $y_{j,n}\sim P_{n}$.
If we define 
\[
  \mu_{n} := \mathsf{KL}\left(P_{n}\hspace{0.3em}\|\hspace{0.3em}Q_{n}\right)
  \quad\text{and}\quad
  \sigma_{n}^{2} := \mathsf{Var}_{y\sim P_{n}}\left[\log\frac{Q_{n}\left(y\right)}{P_{n}\left(y\right)}\right],
\]
then for any given constant $\tau>0$,
\begin{align}
  \mathbb{P}\left\{ \frac{1}{n}\sum\nolimits _{j=1}^{n}\log\frac{Q_{n}(y_{j,n})}{P_{n}(y_{j,n})}+\mu_{n}>\tau\mu_{n}\right\} 
	&= \exp\left\{ -\left(1+o_{n}\left(1\right)\right)n\frac{\tau^{2}\mu_{n}^{2}}{2\sigma_{n}^{2}}\right\} \nonumber\\
	&= \exp\left\{ -\frac{1+o_{n}\left(1\right)}{4}\tau^{2}n\mu_{n}\right\} 
  \label{eq:MDP-k}
\end{align}
and
\begin{align}
  \mathbb{P}\left\{ \frac{1}{n}\sum\nolimits _{j=1}^{n} \log \frac{Q_{n}\left(y_{j,n}\right)}{P_{n}\left(y_{j,n}\right)}+\mu_{n} < -\tau\mu_{n} \right\} 
	&= \exp\left\{ -\left(1+o_{n}\left(1\right)\right)n\frac{\tau^{2}\mu_{n}^{2}}{2\sigma_{n}^{2}}\right\}  \nonumber\\
	&= \exp\left\{ -\frac{1+o_{n}\left(1\right)}{4}\tau^{2}n\mu_{n}\right\} 
  \label{eq:MDP-k-2}
\end{align}
hold as long as $\frac{\mu_{n}^{2}}{\sigma_{n}^{2}}\asymp\frac{\log n}{n}$
and $\mu_{n}\asymp\frac{\log n}{n}$.\end{lem}\begin{proof}This
lemma is a consequence of the moderate deviation theory. See Appendix
\ref{sec:Proof-of-Lemma-MDP}.\end{proof}

\begin{remark}The asymptotic limits presented in Lemma \ref{lem:moderate-deviation}
correspond to the Gaussian tail, meaning that some sort of the central limit theorem
holds in this regime. \end{remark}

We shall now freeze the input to be $x_{1}=\cdots=x_{n}=1$ and consider
the conditional error probability of the ML rule. Without loss of
generality, we assume $P_{0,n}$ and $P_{1,n}$ are minimally separated,
namely,
\begin{equation}
  \mathsf{KL}\left(P_{0,n}\hspace{0.3em}\|\hspace{0.3em}P_{1,n}\right) = \mathsf{KL}_{\min,n}.
  \label{eq:Hmin-01}
\end{equation}
In what follows, we will suppress the dependence on $n$ whenever
clear from the context, and let $\bm{x}^{\mathrm{ML}}$ represent
the ML estimate. We claim that it suffices to prove Theorem \ref{thm:lower-bound}
for the boundary regime where 
\begin{equation}
  \mathsf{KL}_{\min}
  \asymp  \frac{\log n}{np_{\mathrm{obs}}}
  \quad\text{and}\quad
  \mathsf{KL}_{\min}
  \leq  \frac{3.99\log n}{np_{\mathrm{obs}}}.
  \label{eq:boundary-regime}
\end{equation}
In fact, suppose instead that the error probability 
\[
  P_{\mathrm{e}} := \mathbb{P}\left\{ \mcr(\bm{x}^{\mathrm{ML}},\bm{x}) > 0 \right\} 
\]
tends to one in the regime (\ref{eq:boundary-regime}) but is bounded
away from one when $\mathsf{KL}_{\min}=o\left(\frac{\log n}{np_{\mathrm{obs}}}\right)$.
Then this indicates that, in the regime (\ref{eq:boundary-regime}),
one can always add extra noise\footnote{For instance, we can let $\tilde{y}_{i,j}=\begin{cases}
y_{i,j},\quad & \text{with probability }\rho_{0}\\
\mathsf{Unif}\left(m\right),\quad & \text{else}
\end{cases}$ with $\rho_{0}$ controlling the noise level.} to $y_{i,j}$ to decrease $\mathsf{KL}_{\min}$ while significantly
improving the success probability, which results in contradiction. 
Moreover, when $P_n(y)$ and $Q_n(y)$ are both bounded away from 0, it follows from Lemma \ref{lem:KL-Var-Hel} that 
\begin{equation}
	\mathsf{Var}_{y\sim P_{0,n}} \left[ \log\frac{P_{1,n}(y)}{ P_{0,n}(y) } \right] \asymp \mathsf{KL}_{\min} \asymp \frac{\log n}{np_{\mathrm{obs}}}
	\quad \text{and} \quad 
	\frac{\mathsf{KL}_{\min}^2}{	\mathsf{Var}_{y\sim P_{0,n}} \left[ \log\frac{P_{1,n}(y)}{ P_{0,n}(y) } \right] }
	\asymp \frac{\log n}{np_{\mathrm{obs}}}.
	\label{eq:Var-KL}
\end{equation}

Consider a set $\mathcal{V}_{1}=\left\{ 1,\cdots,\delta n\right\} $
for some small constant $\delta>0$. We first single out a subset
$\mathcal{V}_{2}\subseteq\mathcal{V}_{1}$ such that the local likelihood
ratio score---when restricted to samples over the subgraph induced
by $\mathcal{V}_{1}$---is sufficiently large. More precisely, we take
\[
  \mathcal{V}_{2} := \left\{ i\in\mathcal{V}_{1}~\left|~  
   \sum\nolimits _{j:j\in\mathcal{V}_{1}}\log\frac{P_{1}(y_{i,j})}{P_{0}(y_{i,j})}>-2\delta np_{\mathrm{obs}}\mathsf{KL}\left(P_{0}\hspace{0.3em}\|\hspace{0.3em}P_{1}\right)\right.\right\} ;
\]
here and throughout, we set $\frac{P_{1}\left(y_{i,j}\right)}{P_{0}\left(y_{i,j}\right)}=1$
for any $(i,j)\notin\Omega$ for notational simplicity. Recall that
for each $i$, the ML rule favors $P_{1}$ (resp.~$x_{i}=2$) against
$P_{0}$ (resp.~$x_{i}=1$) if and only if
\[
  \sum\nolimits_{j} \log \frac{P_{1}(y_{i,j})}{P_{0}(y_{i,j})} > 0,
\]
which would happen if 
\[
  \sum_{j:\text{ }j\in\mathcal{V}_{1}}\log\frac{P_{1}(y_{i,j})}{P_{0}(y_{i,j})}
  > -2 \delta np_{\mathrm{obs}} \mathsf{KL}\left( P_{0} \hspace{0.3em}\|\hspace{0.3em} P_{1}\right)
  \text{ and }
  \sum_{j:\text{ }j\notin\mathcal{V}_{1}}\log\frac{P_{1}(y_{i,j})}{P_{0}(y_{i,j})}
  > 2 \delta np_{\mathrm{obs}} \mathsf{KL}\left( P_{0} \hspace{0.3em}\|\hspace{0.3em} P_{1}\right).
\]
Thus, conditional on $\mathcal{V}_{2}$ we can lower bound the probability
of error by
\begin{eqnarray}
  P_{\mathrm{e}} & \geq & \mathbb{P}\left\{ \exists i\in\mathcal{V}_{2}:\sum_{j}\log\frac{P_{1}(y_{i,j})}{P_{0}(y_{i,j})}>0\right\} \nonumber \\
   & \geq & \mathbb{P}\left\{ \exists i\in\mathcal{V}_{2}:\sum_{j:\text{ }j\in\mathcal{V}_{1}}\log\frac{P_{1}(y_{i,j})}{P_{0}(y_{i,j})}>-2\delta np_{\mathrm{obs}}\mathsf{KL}\left(P_{0}\hspace{0.3em}\|\hspace{0.3em}P_{1}\right)  \right. \nonumber\\
	& & \qquad\qquad \left. \text{ and }\sum_{j:\text{ }j\notin\mathcal{V}_{1}}\log\frac{P_{1}(y_{i,j})}{P_{0}(y_{i,j})}>2\delta np_{\mathrm{obs}}\mathsf{KL}\left(P_{0}\hspace{0.3em}\|\hspace{0.3em}P_{1}\right)\right\} \nonumber \\
   & = & \mathbb{P}\left\{ \exists i\in\mathcal{V}_{2}:\sum_{j:\text{ }j\notin\mathcal{V}_{1}}\log\frac{P_{1}(y_{i,j})}{P_{0}(y_{i,j})}>2\delta np_{\mathrm{obs}}\mathsf{KL}\left(P_{0}\hspace{0.3em}\|\hspace{0.3em}P_{1}\right)\right\} ,
  \label{eq:Pe-LB1}
\end{eqnarray}
where the last identity comes from the definition of $\mathcal{V}_{2}$. 

We pause to remark on why (\ref{eq:Pe-LB1}) facilitates analysis.
To begin with, $\mathcal{V}_{2}$ depends only on those samples lying
within the subgraph induced by $\mathcal{V}_{1}$, and is thus independent
of $\sum_{j\notin\mathcal{V}_{1}}\log\frac{P_{1}\left(y_{i,j}\right)}{P_{0}\left(y_{i,j}\right)}$.
More importantly, the scores $\left\{ s_i := \sum_{j:\text{ }j\notin\mathcal{V}_{1}}\log\frac{P_{1}(y_{i,j})}{P_{0}(y_{i,j})}\right\} $
are statistically independent across all $i\in\mathcal{V}_{1}$ as
they rely on distinct samples. These allow us to derive that, conditional on $\mathcal{V}_{2}$,
\begin{eqnarray}
  (\ref{eq:Pe-LB1}) & = & 1-\prod_{i\in\mathcal{V}_{2}}\left(1-\mathbb{P}\left\{ \sum\nolimits _{j:j\notin\mathcal{V}_{1}}\log\frac{P_{1}(y_{i,j})}{P_{0}(y_{i,j})}>2\delta np_{\mathrm{obs}}\mathsf{KL}\left(P_{0}\hspace{0.3em}\|\hspace{0.3em}P_{1}\right)\right\} \right)\nonumber \\
   & \geq & 1-\left(1-\exp\left\{ -\left(1+o\left(1\right)\right)\left(1+2\delta\right)^{2}np_{\mathrm{obs}}\frac{\mathsf{KL}_{\min}}{4}\right\} \right)^{\left|\mathcal{V}_{2}\right|}
   \label{eq:Pe-LB2}\\
   & \geq & 1-\exp\left\{ -\left|\mathcal{V}_{2}\right|\exp\left\{ -\left(1+o\left(1\right)\right)\left(1+2\delta\right)^{2}np_{\mathrm{obs}}\frac{\mathsf{KL}_{\min}}{4}\right\} \right\} ,
  \label{eq:Pe-LB3}
\end{eqnarray}
where the last line results from the elementary inequality $1-x\leq e^{-x}$.
To see why (\ref{eq:Pe-LB2}) holds, we note that according to the Chernoff bound, the number of samples
linking each $i$ and $\overline{\mathcal{V}}_{1}$ (i.e.~$\left|\left\{ j\notin\mathcal{V}_{1}:(i,j)\in\Omega\right\} \right|$)
is at most $np_{\mathrm{obs}}$ with high probability, provided that (i) $\frac{np_{\mathrm{obs}}}{\log n}$
is sufficiently large, and (ii) $n$ is sufficiently large.
These taken collectively with Lemma
\ref{lem:moderate-deviation} and (\ref{eq:Var-KL}) yield 
\begin{align*}
  &\mathbb{P}\left\{ \sum\nolimits _{j:\text{ }j\notin\mathcal{V}_{1}}\log\frac{P_{1}(y_{i,j})}{P_{0}(y_{i,j})}>2\delta np_{\mathrm{obs}}\mathsf{KL}\left(P_{0}\hspace{0.3em}\|\hspace{0.3em}P_{1}\right)\right\} \\
  &\qquad\qquad \geq \exp\left\{ -\left(1+o\left(1\right)\right) \left(1+2\delta\right)^{2} np_{\mathrm{obs}}\frac{\mathsf{KL}_{\min}}{4}\right\} ,
\end{align*}
thus justifying (\ref{eq:Pe-LB2}). 

To establish Theorem \ref{thm:lower-bound}, we would need to show
that (\ref{eq:Pe-LB3}) (and hence $P_{\mathrm{e}}$) is lower bounded
by $1-o\left(1\right)$ or, equivalently, 
\[
  \left|\mathcal{V}_{2}\right|\exp\left\{ -\left(1+o\left(1\right)\right)\left(1+2\delta\right)^{2}np_{\mathrm{obs}}\mathsf{KL}_{\min}/4\right\} \rightarrow\infty.
\]
This condition would hold if 
\begin{equation}
  \left(1+2\delta\right)^{2}np_{\mathrm{obs}}\mathsf{KL}_{\min}/4<\left(1-\delta\right)\log n
  \label{eq:H2-UB}
\end{equation}
 and
\begin{equation}
  \left|\mathcal{V}_{2}\right|=\left(1-o\left(1\right)\right)|\mathcal{V}_{1}|=\left(1-o\left(1\right)\right)\delta n,
  \label{eq:V2-size}
\end{equation}
since under the above two hypotheses one has
\begin{eqnarray*}
  \left|\mathcal{V}_{2}\right|\exp\left\{ -\left(1+2\delta\right)^{2}np_{\mathrm{obs}}\mathsf{KL}_{\min}/4\right\}  
  & \geq & \left|\mathcal{V}_{2}\right|\exp\left\{ -\left(1-\delta\right)\log n\right\} \\
  & \geq & \delta\exp\left\{ \log n-\left(1-\delta\right)\log n\right\} \\
  & \geq & \delta n^{\delta}\rightarrow\infty.
\end{eqnarray*}
The first condition (\ref{eq:H2-UB}) is a consequence from (\ref{eq:MinimaxBound})
as long as $\delta$ is sufficiently small. It remains to verify the
second condition (\ref{eq:V2-size}). 

When $np_{\mathrm{obs}}>c_{0}\log n$ for some sufficiently large
constant $c_{0}>0$,
each $i\in\mathcal{V}_{1}$
is connected to at least $\left(1-\delta\right)p_{\mathrm{obs}}|\mathcal{V}_{1}|$
vertices in $\mathcal{V}_{1}$ with high probability, meaning that the number of random
variables involved in the sum $\text{ }\sum_{j:\text{ }j\in\mathcal{V}_{1}}\log\frac{P_{1}\left(y_{i,j}\right)}{P_{0}\left(y_{i,j}\right)}$
concentrates around $|\mathcal{V}_1| p_{\mathrm{obs}}$. Lemma \ref{lem:moderate-deviation}
thus implies that
\begin{align*}
   & \mathbb{P}\left\{ \sum\nolimits _{j:\text{ }j\in\mathcal{V}_{1}} \log\frac{P_{1}(y_{i,j})}{P_{0}(y_{i,j})} 
     \leq - 2\delta np_{\mathrm{obs}}\mathsf{KL}\left(\mathbb{P}_{0}\|\mathbb{P}_{1}\right)\right\}  \nonumber \\
   & \qquad\qquad = \mathbb{P}\left\{ \sum\nolimits _{j:\text{ }j\in\mathcal{V}_{1}} \log\frac{P_{1}(y_{i,j})}{P_{0}(y_{i,j})}
    \leq - 2|\mathcal{V}_1| p_{\mathrm{obs}} \mathsf{KL}\left(\mathbb{P}_{0}\|\mathbb{P}_{1}\right)\right\}   \nonumber \\
   & \qquad\qquad \leq  \exp\left\{ -c_4 \left( |\mathcal{V}_1| p_{\mathrm{obs}}\right)\mathsf{KL}_{\min}  \right\} 
     ~\leq~ \exp\left\{ -c_5 \left(\delta np_{\mathrm{obs}}\right)\mathsf{KL}_{\min}\right\}
\end{align*}
for some constants $c_4,c_5>0$, provided that $n$ is sufficiently large.
This gives rise to an upper bound
\begin{eqnarray*}
  \mathbb{E}\left[\left|\mathcal{V}_{1}\backslash\mathcal{V}_{2}\right|\right] 
   & = & \left|\mathcal{V}_{1}\right|\cdot\mathbb{P}\left\{ \sum\nolimits _{j:\text{ }j\in\mathcal{V}_{1}}\log\frac{P_{1}(y_{i,j})}{P_{0}(y_{i,j})}\leq-2\delta n\mathsf{KL}\left(P_{0}\|P_{1}\right)\right\} \\
   & \leq & \delta n\cdot\exp\left\{ -c_{5}\delta np_{\mathrm{obs}}\mathsf{KL}_{\min}\right\} \\
   & \overset{(\text{i})}{=} & \delta n\cdot n^{-\Theta(\delta)}
   = o\left(n\right),
\end{eqnarray*}
where (i) arises from the condition (\ref{eq:boundary-regime}). As
a result, Markov's inequality implies that with probability approaching
one, 
\[
  \left|\mathcal{V}_{1}\backslash\mathcal{V}_{2}\right| = o\left(n\right)
\]
or, equivalently, 
$  \left|\mathcal{V}_{2}\right| = \left(1-o\left(1\right)\right)\delta n.
$
This finishes the proof of Theorem \ref{thm:lower-bound}.

\section{Discussion}
\label{sec:discussion}

We have developed an efficient nonconvex paradigm 
for a class of discrete assignment problems. There are numerous
questions we leave open that might be interesting for future
investigation. For instance, it can be seen from
Fig.~\ref{fig:numerics-RCM} and Fig.~\ref{fig:numerics-Gaussian}  that
the algorithm returns reasonably good estimates even when we are below
the information limits.  A natural question is this: how can we
characterize the accuracy of the algorithm if one is satisfied with
approximate solutions? In addition, this work assumes the index set
$\Omega$ of the pairwise samples are drawn uniformly at
random. Depending on the application scenarios, we might encounter
other measurement patterns that cannot be modeled in this random
manner; for example, the samples might only come from nearby objects and hence the sampling pattern might be highly local (see, e.g.~\cite{chen2016community,globerson2015hard}). Can we determine the performance of the algorithm for more
general sampling set $\Omega$?  Moreover, the log-likelihood functions
we incorporate in the data matrix $\bm{L}$ might be imperfect. Further
study could help understand the stability of the algorithm in the
presence of model mismatch.

Returning to Assumption \ref{assumption-KL-D}, we remark that this assumption is imposed primarily out of computational concern. In fact, $\mathsf{KL}_{\max}/\mathsf{KL}_{\min}$ being exceedingly large might actually be a favorable case from an information theoretic viewpoint,  as it indicates that the hypothesis corresponding to  $\mathsf{KL}_{\max}$ is much easier to preclude compared to other hypotheses. 
It would be interesting to establish rigorously the performance of the PPM without this assumption and, in case it becomes suboptimal, how shall we modify the algorithm so as to be more adaptive to the most general class of noise models.

Moving beyond joint alignment, we are interested in seeing the potential benefits of the PPM on other discrete problems. For instance, the joint alignment problem falls
under the category of maximum \emph{a posteriori} (MAP) inference
in a discrete Markov random field, which spans numerous applications
including segmentation, object detection, error correcting codes,
and so on \cite{blake2011markov,ravikumar2006quadratic, guibas2014scalable}. 
Specifically, consider $n$ discrete variables $x_{i}\in[m]$, $1\leq i\leq n$.
We are given a set of unitary potential functions (or prior distributions) $\left\{ \psi_{i}(x_{i})\right\} _{1\leq i\leq n}$ as well
as a collection of pairwise potential functions (or likelihood functions) $\left\{ \psi_{i,j}(x_{i},x_{j})\right\} _{(i,j)\in\mathcal{G}}$
over some graph $\mathcal{G}$. The goal is to compute
the MAP assignment 
\begin{equation}
  \bm{x}_{\mathrm{MAP}} ~:=~ \arg\max_{\bm{z}}
   \prod_{i=1}^{n} \psi_{i} (z_{i}) \prod_{(i,j)\in\mathcal{G}} \psi_{i,j} (z_{i},z_{j})  
   ~=~\arg\max_{\bm{z}} 
	\left\{ \sum_{i=1}^{n} \log \psi_{i} (z_{i}) + \sum_{(i,j)\in\mathcal{G}} \log \psi_{i,j} (z_{i},z_{j}) \right\}.
\end{equation}
Similar to (\ref{eq:state-vector}) and (\ref{eq:Input-Matrix}), one can introduce the vector
$\bm{z}_{i}=\bm{e}_{j}$ to represent $z_{i}=j$, and use a matrix
$\bm{L}_{i,j}\in\mathbb{R}^{m\times m}$ to encode each pairwise log-potential
function $\log\psi_{i,j}\left(\cdot,\cdot\right)$, $(i,j) \in \mathcal{G}$. 
The unitary potential function $\psi_{i}\left(\cdot\right)$ can also
be encoded by a diagonal matrix $\bm{L}_{i,i} \in \mathbb{R}^{m\times m}$ 
\begin{equation}
	\left(\bm{L}_{i,i}\right)_{\alpha,\alpha}=\log\psi_{i}(\alpha),\qquad 1 \leq \alpha \leq m
\end{equation}
so that $\log\psi(z_i) = \bm{z}_i^{\top} \bm{L}_{i,i} \bm{z}_i$. This enables a quadratic form representation of MAP estimation:
\begin{eqnarray}
  \text{maximize}_{\bm{z}} &  & \bm{z}^{\top}\bm{L}\bm{z} \\
  \text{subject to } &  & \bm{z}_{i}\in\left\{ \bm{e}_{1},\cdots,\bm{e}_{m}\right\} ,\quad 1\leq i\leq n. \nonumber
\end{eqnarray}
As such, we expect the PPM to be effective in
solving many instances of such MAP inference problems. One of the key
questions amounts to finding an appropriate initialization that allows
efficient exploitation of the unitary prior belief
$\psi_{i}(\cdot)$. We leave this for future work.

\section*{Acknowledgements}

E.~C.~is partially supported by NSF via grant DMS-1546206, and by the
Math + X Award from the Simons Foundation.  Y.~C.~is supported by the
same award.  We thank Qixing Huang for motivating discussions about
the join image alignment problem.  Y.~Chen is grateful to Qixing
Huang, Leonidas Guibas, and Nan Hu for helpful discussions about joint
graph matching.

\appendix

\section{Proof of Theorem \ref{thm:ExactRecovery-General-2}}
\label{sec:proof-thm:ExactRecovery-General-2}

We will concentrate on proving the case where Assumption \ref{assumption-Pmin} is violated. 
Set 
$y_{\max}:=\arg\max_{y}P_{0}(y)$ and $y_{\min}:=\arg\min_{y}P_{0}(y)$.
Since $m$ is fixed, it is seen that 
$
  P_{0}(y_{\max}) \asymp 1.
$
We also have $P_{0}(y_{\min}) \rightarrow 0$. 
Denoting by $\nu:=\frac{P_{0}(y_{\max})}{P_{0}(y_{\min})}$ the
dynamic range of $P_{0}$ and introducing the metric
\[
  d\left(P \hspace{0.3em}\|\hspace{0.3em} Q\right) := \sum\nolimits_{y}P(y) \left| \log\frac{P(y)}{Q(y)} \right|,
\]
we obtain
\[
   \max_{l}d\left(P_{0}\|P_{l}\right)
   = \max_{l} \sum_{y}P_{0}(y) \left|\log \frac{P_{0}(y)}{P_{l}(y)} \right|
   \asymp P_{0}(y_{\max}) \left| \log \frac{P_{0}(y_{\max})}{P_{0}(y_{\min})} \right|
   \asymp \log\nu \gg 1.
\]
The elementary inequality $\mathsf{KL}(P \hspace{0.3em}\|\hspace{0.3em} Q) \leq d( P\hspace{0.3em}\|\hspace{0.3em} Q)$, 
together with the second Pinsker's inequality $\mathsf{KL}(P \hspace{0.3em}\| \hspace{0.3em} Q)+\sqrt{2\mathsf{KL}(P \hspace{0.3em}\| \hspace{0.3em}Q)}\geq d(P \hspace{0.3em}\|\hspace{0.3em} Q)$
\cite[Lemma 2.5]{tsybakov2008introduction}, reveals that
\[
  \mathsf{KL}_{\max}  \asymp  \max_{l} d(P_{0} \hspace{0.3em}\|\hspace{0.3em} P_{l})  \asymp \log\nu.
\]
Making use of Assumption \ref{assumption-KL-D} we get 
\begin{equation}
  d(P_{0} \hspace{0.3em}\| \hspace{0.3em} P_{l}) \asymp \mathsf{KL}_{\min} \asymp \mathsf{KL}_{\max}  \asymp \log\nu,  \qquad 1\leq l <m.
  \label{eq:dmin}
\end{equation}

Next, for each $1\leq l< m$ we single out an element  
$
  y_{l} := \arg\max_{y}  P_{0}(y) \left| \log\frac{P_{0}(y)}{P_{l}(y)} \right|.
$
This element is important because, when $m$ is fixed,  
\[
  d(P_{0} \hspace{0.3em}\|\hspace{0.3em} P_{l})\asymp P_{0}(y_{l})\left|\log\frac{P_{0}(y_{l})}{P_{l}(y_{l})}\right|
  \lesssim P_{0}(y_{l})\log\nu.
\]
As a result, (\ref{eq:dmin}) would only happen if $P_{0}(y_{l}) \asymp 1$
and $\frac{P_{0}(y_{l})}{P_{l}(y_{l})} \rightarrow \infty$. 

We are now ready to prove the theorem. With $P_{l}$ replaced by $\tilde{P}_{l}$
as defined in (\ref{eq:tilde-P}), one has
\begin{eqnarray*}
  \mathsf{KL}\big( \tilde{P}_{0} \hspace{0.3em}\|\hspace{0.3em} \tilde{P}_{l} \big) 
	& \geq &  2 \mathsf{TV}^2 \big(\tilde{P}_{0},\tilde{P}_{l} \big)
   	\geq  \frac{1}{2} \big( \tilde{P}_{0}(y_{l}) -  \tilde{P}_{l}(y_{l}) \big)^{2} \nonumber\\
	& \asymp &   \big( P_{0}(y_{l}) -  P_{l}(y_{l}) \big)^{2}
   \asymp  P_{0}^2(y_{l})
   \asymp 1,
\end{eqnarray*}
which exceeds the threshold $4.01\log n/(np_{\mathrm{obs}})$ as long
as $p_{\mathrm{obs}}\geq c_{1}\log n/n$ for some sufficiently large
constant $c_{1}$. 
In addition, since $\tilde{P}_0(y)$ is bounded away from both 0 and 1, it is easy to see that $\mathsf{KL}\big( \tilde{P}_{0} \hspace{0.3em}\|\hspace{0.3em} \tilde{P}_l) \lesssim 1$ for any $l$ and, hence, Assumption  \ref{assumption-KL-D} remains valid. 
Invoking Theorem \ref{thm:ExactRecovery-General} concludes the proof.

\section{Proof of Lemma \ref{lem:two-class}\label{sec:Proof-of-Lemma-two-class}}

(1) It suffices to prove the case where $\mathsf{KL}_{\min}=\mathsf{KL}\left(P_{0}\hspace{0.3em}\|\hspace{0.3em}P_{1}\right)$.
Suppose that 
\[
  \left|P_{0}(y_{0})-P_{l}(y_{0})\right| = \max\nolimits_{j,y}\left|P_{0}(y)-P_j(y)\right|
\]
for some $0\leq l,y_{0}<m$. In view of Pinsker's inequality \cite[Lemma 2.5]{tsybakov2008introduction},
\begin{eqnarray}
  \mathsf{KL}( P_{0} \hspace{0.3em}\|\hspace{0.3em} P_{1} ) 
  & \geq & 2\mathsf{TV}^{2}( P_{0},P_{1} )
  = 2\left(\frac{1}{2}\sum_{y=0}^{m-1}\left|P_{0}\left(y\right)-P_{1}\left(y\right)\right|\right)^{2} \nonumber \\
  & = & \frac{1}{2} \left(\sum_{y=0}^{m-1}\left|P_{0}\left(y\right)-P_{0}\left(y-1\right)\right|\right)^{2}\nonumber \\
  & \geq & 
   \frac{1}{2} \left|P_{0}\left(y_{0}\right)-P_{0}\left(y_{0}-l\right)\right|^{2}
   =  \frac{1}{2} \max\nolimits _{j,y} | P_{0}(y)-P_j(y) |^2.
  \label{eq:KL-LB}
\end{eqnarray}
In addition, for any $1\leq j <m$, 
\begin{align}
  \mathsf{KL}\left(P_{0}\hspace{0.3em}\|\hspace{0.3em}P_{j}\right)
  ~& \overset{(\text{a})}{\leq}~  \chi^{2} \left( P_{0}\hspace{0.3em}\|\hspace{0.3em}P_{j} \right)
  ~:=~ \sum_{y}\frac{ \left(P_{0}\left(y\right)-P_{j}\left(y\right)\right)^{2} }{ P_{j}\left(y\right) } \nonumber\\
  ~&\overset{(\text{b})}{\asymp}~  \sum_{y}\left(P_{0}\left(y\right)-P_{j}\left(y\right)\right)^2  \nonumber\\
  ~& \overset{(\text{c})}{\lesssim}~  \max\nolimits_{j,y} |P_{0}(y)-P_j(y)|^2,
  \label{eq:KL-UB}
\end{align}
where (a) comes from \cite[Eqn.~(5)]{sason2015f},
(b) is a consequence from Assumption \ref{assumption-Pmin}, and (c)
follows since $m$ is fixed. Combining (\ref{eq:KL-LB}) and (\ref{eq:KL-UB})
establishes $\mathsf{KL}_{\min}\asymp\mathsf{KL}_{\max}$. 

(2) Suppose that $P_{0}(y^{*})=\min_{y}P_{0}(y)$
for $y^{*}= \lfloor m/2 \rfloor$, and that $\mathsf{KL}\left(P_{0}\hspace{0.3em}\|\hspace{0.3em}P_{l}\right)=\mathsf{KL}_{\min}$
for some $0<l \leq y^{*}$. Applying Pinsker's inequality again gives
\begin{eqnarray}
  \mathsf{KL}\left( P_{0}\hspace{0.3em}\|\hspace{0.3em}P_{l} \right) 
  & \geq & 2\mathsf{TV}^{2}\left(P_{0},P_{l}\right)
  = \frac{1}{2} \left(\sum _{y=0}^{m-1}\left|P_{l}\left(y\right)-P_{0}\left(y\right)\right|\right)^{2} \nonumber\\
  & = & \frac{1}{2} \left(\sum _{y=0}^{m-1}\left|P_{0}\left(y-l\right)-P_{0}\left(y\right)\right|\right)^{2}\nonumber \\
  & \geq & \frac{1}{2} \left(\left|P_{0}(y^{*})-P_{0}\left(y^{*}-l\right)\right|+\sum _{k=1}^{\left\lceil y^{*}/l\right\rceil -1}\big|P_{0}\left(kl\right)-P_{0}\left((k-1)l\right)\big|\right)^{2}\nonumber \\
  & \geq & \frac{1}{2} \left(\left|P_{0}(y^{*})-P_{0}\left(\left(\left\lceil y^{*}/l\right\rceil -1\right)l\right)\right|+\sum _{k=1}^{\left\lceil y^{*}/l\right\rceil -1}\big| P_{0}\left(kl\right)-P_{0}\left((k-1)l\right)\big|\right)^{2}
  \label{eq:monotone}\\
 & \geq &  \frac{1}{2} \left|P_{0}(y^{*})-P_{0}(0)\right|^{2}
  = \frac{1}{2} \max_{y} |P_{0}(0)-P_0(y) |^2 \nonumber\\
 & = & \frac{1}{2} \max_{j,y} |P_{0}(y)-P_j(y) |^2,
  \label{eq:KL-LB-1}
\end{eqnarray}
where (\ref{eq:monotone}) follows from the unimodality assumption, and the last line results from the facts  $P_{0}(0) = \max_{j,y}P_j(y)$ and $P_{0}(y^*) = \min_{j,y}P_j(y)$. 
These taken collectively with (\ref{eq:KL-UB}) finish the proof.

\section{Proof of Lemma \ref{lemma:V-KL-UB} \label{sec:Proof-of-Lemma-V-KL-UB}}

(1) Since $\log(1+x)\leq x$ for any $x\geq0$, we get
\begin{eqnarray}
  \left|\log\frac{Q\left(y\right)}{P\left(y\right)}\right| & = & 
  \begin{cases}
    \log\left(1+\frac{Q\left(y\right)-P\left(y\right)}{P\left(y\right)}\right),\quad & \text{if }Q\left(y\right)\geq P\left(y\right)\\
    \log \left(1+\frac{P\left(y\right)-Q\left(y\right)}{Q\left(y\right)}\right), & \text{else}
  \end{cases}
   \\
   & \leq & \frac{\left|Q\left(y\right)-P\left(y\right)\right|}{\min\left\{ P\left(y\right),Q\left(y\right)\right\} }
  \label{eq:UB-LR-2}\\
   & \leq & ~\frac{2\mathsf{TV}\left(P,Q\right)}{\min\left\{ P\left(y\right),Q\left(y\right)\right\} }
  ~\leq~ \frac{\sqrt{2\mathsf{KL}\left(P\hspace{0.3em}\|\hspace{0.3em}Q\right)}}{\min\left\{ P\left(y\right),Q\left(y\right)\right\} },
  \label{eq:UB-LR-3}
\end{eqnarray}
where the last inequality comes from Pinsker's inequality. 

(2) Using the inequality (\ref{eq:UB-LR-2}) once again as well as
the definition of $\kappa_{0}$, we obtain
\begin{align}
  & \mathbb{E}_{y\sim P}\left[\left(\log\frac{P\left(y\right)}{Q\left(y\right)}\right)^{2}\right]
  = \sum\nolimits _{y}P(y)\left(\log\frac{P\left(y\right)}{Q\left(y\right)}\right)^{2}\nonumber \\
  & \quad\leq\sum\nolimits _{y} P(y)\frac{|P\left(y\right)-Q(y)|^{2}}{\min\left\{ P^{2}(y),Q^{2}\left(y\right)\right\} }
  = \sum\nolimits _{y} \frac{|P\left(y\right)-Q(y)|^{2}}{\min\left\{ \frac{P(y)}{Q(y)},\frac{Q(y)}{P(y)}\right\} Q(y)} \\
  & \quad\leq~ \kappa_{0} \sum\nolimits _{y}\frac{\left(Q\left(y\right)-P(y)\right)^{2}}{Q(y)}.
  \label{eq:V-chi2-1}
\end{align}
Note that $\sum_y\frac{\left(Q\left(y\right)-P\left(y\right)\right)^{2}}{Q\left(y\right)}$
is exactly the $\chi_{2}$ divergence between $P$ and $Q$,
which  satisfies  \cite[Proposition 2]{dragomir2000upper}
\begin{equation}
  \chi_{2}\left(P\hspace{0.3em}\|\hspace{0.3em}Q\right)\leq2\kappa_{0}\mathsf{KL}\left(P\hspace{0.3em}\|\hspace{0.3em}Q\right).
  \label{eq:chi2-KL-UB}
\end{equation}
Substitution into (\ref{eq:V-chi2-1}) concludes the proof.

\section{Proof of Lemma \ref{lem:KL-Var-Hel}\label{sec:Proof-of-Lemma-Chernoff-LB}}

For notational simplicity, let
\[
  \mu := \mathbb{E}_{y\sim P}\left[\log\frac{P\left(y\right)}{Q\left(y\right)}\right] 
   = \mathsf{KL}\left(P\hspace{0.3em}\|\hspace{0.3em}Q\right)\quad\text{and}\quad\sigma
   =  \sqrt{{ \mathsf{Var} }_{y\sim P}\left[\log\frac{P\left(y\right)}{Q\left(y\right)}\right]}.
\]
We first recall from our calculation in (\ref{eq:UB-LR-3}) that 
\[
  \max \left\{ \frac{\left|Q\left(y\right)-P(y)\right|}{Q(y)},\frac{\left|Q\left(y\right)-P(y)\right|}{P(y)}\right\} 
  \leq \frac{\sqrt{2\mu}}{\min\left\{ P(y),Q(y)\right\} } 
  = O ( \sqrt{\mu} ),
\]
where the last identity follows since $P(y)$ and $Q(y)$ are all bounded away from 0. 
Here and below,  the notation $f(\mu)=O(\sqrt{\mu})$ means $|f(\mu)|\leq c_{0}\sqrt{\mu}$ for some universal constant
$c_{0}>0$. 
This fact tells us that
\begin{equation}
  \frac{P(y)}{Q(y)},\frac{Q(y)}{P(y)} \in \left[1 \pm O(\sqrt{\mu}) \right],
  \label{eq:max-PQ-ratio}
\end{equation}
thus indicating that
\begin{align*}
  \log\left(\frac{ Q(y) }{ P(y)} \right)
  &= \log \left( 1+\frac{ Q(y)-P(y) }{ P(y)} \right)
  = \frac{Q(y)-P(y)}{P(y)}+O\left(\frac{\left|Q(y)-P(y)\right|^{2}}{P^{2}(y)}\right) \\
  &= \left(1 + O\left(\sqrt{\mu}\right)\right) \frac{Q\left(y\right)-P(y)}{P(y)}.
\end{align*}
All of this allows one to write
\begin{eqnarray*}
  \sigma^{2} & = & \mathbb{E}_{y\sim P}\left[\left(\log\frac{Q\left(y\right)}{P\left(y\right)}\right)^{2}\right]-\mu^{2}=\left(1+O\left(\sqrt{\mu}\right)\right)\mathbb{E}_{y\sim P}\left[\left(\frac{Q\left(y\right)-P(y)}{P\left(y\right)}\right)^{2}\right]-\mu^{2}\\
   & \overset{\text{(a)}}{=} & \left( 1 + O\left(\sqrt{\mu}\right)\right)\mathbb{E}_{y\sim Q}\left[\left(\frac{Q\left(y\right)-P(y)}{Q\left(y\right)}\right)^{2}\right]-\mu^{2} \\
   & = &\left(1+O\left(\sqrt{\mu}\right)\right)\chi^{2}\left(P\hspace{0.3em}\|\hspace{0.3em}Q\right)-\mu^{2}\\
   & = & 2\left(1+O\left(\sqrt{\mu}\right)\right)\mu. 
\end{eqnarray*}
Here, (a) arises due to (\ref{eq:max-PQ-ratio}), as the difference $Q(y)/P(y)$ can be absorbed into the prefactor $1+O(\mu)$ by adjusting the constant in $O(\mu)$ appropriately. 
%
The last line follows since, by  \cite[Proposition 2]{dragomir2000upper},
\[
  \frac{\mu}{\chi^{2}\left(P\hspace{0.3em} \| \hspace{0.3em}Q \right)}
  = \frac{\mathsf{KL}\left(P\hspace{0.3em} \| \hspace{0.3em}Q \right)}{\chi^{2}\left(P\hspace{0.3em}\|\hspace{0.3em}Q\right)}
  = \left(1+O\left(\max\left\{ \max_{y}\frac{Q(y)}{P(y)},\max_{y}\frac{P(y)}{Q(y)}\right\} \right)\right)\frac{1}{2}
  =  \left(1+O\left(\sqrt{\mu}\right)\right)\frac{1}{2}.
\]

Furthermore, it follows from \cite[Fact 1]{chen2015information} that
\[
  \left(4-\log R\right)\mathsf{H}^{2}\left(P,Q\right)
  \leq \mathsf{KL}\left(P\hspace{0.3em}\|\hspace{0.3em}Q\right)
  \leq \left(4+2\log R\right)\mathsf{H}^{2}\left(P,Q\right),
\]
where $R:=\max\left\{ \max_{y}\frac{P(y)}{Q(y)},\max_{y}\frac{Q(y)}{P(y)}\right\} $.
In view of (\ref{eq:max-PQ-ratio}), it is seen that $\log R=O\left(\sqrt{\mu}\right)$,
thus establishing (\ref{eq:lemma-H2-KL}).

\section{Proof of Lemma \ref{lem:random-matrix-norm}\label{sec:Proof-of-Lemma-Random-Matrix-Norm}}

It is tempting to invoke the matrix Bernstein inequality \cite{tropp2015introduction} to analyze random block matrices,
but it loses a logarithmic factor in comparison to the  bound
advertised in Lemma \ref{lem:random-matrix-norm}. As it turns out, it would be better
to resort to Talagrand's inequality \cite{talagrand1995concentration}. 

The starting point is to use the standard moment method  and reduce to the case with independent entries
(which has been studied in \cite{seginer2000expected,bandeira2014sharp}). Specifically,
a standard symmetrization argument \cite[Section 2.3]{tao2012topics}
gives 
\begin{equation}
  \mathbb{E}\left[\|\bm{M}\|\right]\leq\sqrt{2\pi}\mathbb{E}\left[\left\Vert \bm{B}\right\Vert \right],
  \label{eq:Gaussian-symm}
\end{equation}
where $\bm{B}=[\bm{B}_{i,j}]_{1\leq i,j\leq n}:=\left[g_{i,j}\bm{M}_{i,j}\right]_{1\leq i,j\leq n}$
is obtained by inserting i.i.d.~standard Gaussian variables
$\left\{ g_{i,j}\mid i\geq j\right\} $ in front of $\left\{ \bm{M}_{i,j}\right\} $.
In order to upper bound $\left\Vert \bm{B}\right\Vert $, we further
recognize that
$
  \|\bm{B}\|^{2p}\leq\mathsf{Tr}\left(\bm{B}^{2p}\right)
$
for all $p\in\mathbb{Z}$. 
Expanding $\bm{B}^{2p}$ as a sum over cycles of length $2p$ and
conditioning on $\bm{M}$, we have
\begin{eqnarray}
  \mathbb{E}\left[\mathsf{Tr}\left(\bm{B}^{2p}\right)\mid\bm{M}\right] & = & \sum_{1\leq i_{1},\cdots,i_{2p}\leq n}\mathbb{E}\left[\mathsf{Tr}\left(\bm{B}_{i_{1},i_{2}}\bm{B}_{i_{2},i_{3}}\cdots\bm{B}_{i_{2p-1},i_{2p}}\bm{B}_{i_{2p}i_{1}}\right)\mid\bm{M}\right]\nonumber \\
  & = & \sum_{1\leq i_{1},\cdots,i_{2p}\leq n}\mathsf{Tr}\left(\bm{M}_{i_{1},i_{2}}\cdots\bm{M}_{i_{2p},i_{1}}\right)\mathbb{E}\left[\prod\nolimits _{j=1}^{2p}g_{i_{j},i_{j+1}}\right]
  \label{eq:E-tr-B}
\end{eqnarray}
with the cyclic notation $i_{2p+1}=i_{1}$. The summands that are
non-vanishing are those in which each distinct edge is visited an
even number of times \cite{tao2012topics}, and these summands obey $\mathbb{E}\big[\prod_{j=1}^{2p}g_{i_{j},i_{j+1}}\big] \geq 0$.
As a result, 
\begin{eqnarray}
  (\ref{eq:E-tr-B}) & \leq & \sum_{1\leq i_{1},\cdots,i_{2p}\leq n}m \left( \prod\nolimits _{j=1}^{2p} \Vert \bm{M}_{i_{j},i_{j+1}} \Vert  \right)
  \mathbb{E}\left[ \prod\nolimits _{j=1}^{2p} g_{i_{j},i_{j+1}} \right].
  \label{eq:E-UB}
\end{eqnarray}

We make the observation that the right-hand side of (\ref{eq:E-UB}) is equal to $m\cdot\mathbb{E}\left[\mathsf{Tr}\left(\bm{Z}^{2p}\right) \mid \bm{M} \right]$,
 where $\bm{Z}:=\left[g_{i,j}\|\bm{M}_{i,j}\|\right]_{1\leq i,j\leq n}$.
Following the argument in \cite[Section 2]{bandeira2014sharp}
and setting $p=\log n$, one derives 
\[
  \left(\mathbb{E}\left[\mathsf{Tr}\left(\bm{Z}^{2p}\right)\mid\bm{M}\right]\right)^{\frac{1}{2p}}
  \lesssim \sigma + K\sqrt{2\log n},
\]
where $\sigma:=\sqrt{\max_{i}\sum_{j}\|\bm{M}_{i,j}\|^{2}}$, and $K$ is the upper bound on $\max_{i,j}\|\bm{M}_{i,j}\|$. Putting
all of this together, we obtain
\begin{eqnarray}
  \mathbb{E}\left[\|\bm{B}\|\mid\bm{M}\right] & \leq & \left(\mathbb{E}\left[\|\bm{B}\|^{2p}\mid\bm{M}\right]\right)^{\frac{1}{2p}}\leq\left(\mathbb{E}\left[\mathsf{Tr}\left(\bm{B}{}^{2p}\right)\mid\bm{M}\right]\right)^{\frac{1}{2p}}\nonumber \\
   & \leq & m^{\frac{1}{2p}}\cdot\left(\mathbb{E}\left[\mathsf{Tr}\left(\bm{Z}^{2p}\right)\mid\bm{M}\right]\right)^{\frac{1}{2p}}
  ~\lesssim~ \sigma+K\sqrt{\log n},
  \label{eq:E-tr-B-1}
\end{eqnarray}
where the last inequality follows since $m^{\frac{1}{\log n}}\lesssim1$
as long as $m=n^{O(1)}$. Combining (\ref{eq:Gaussian-symm}) and
(\ref{eq:E-tr-B-1}) and undoing the conditional expectation yield
\begin{equation}
  \mathbb{E}\left[\|\bm{M}\|\right] 
  ~\leq~  \sqrt{2\pi}\mathbb{E}\left[\left\Vert \bm{B}\right\Vert \right]
  ~\lesssim~ \mathbb{E} \big[\sigma+K\sqrt{\log n} \big].
\end{equation}
Furthermore, Markov's inequality gives $\mathbb{P}\left\{ \|\bm{M}\|\geq2\mathbb{E}\left[\|\bm{M}\|\right]\right\} \leq 1/2$
and hence 
\begin{equation}
  \mathsf{Median}\left[\|\bm{M}\|\right] 
  ~\leq~ 2\mathbb{E}\left[\|\bm{M}\|\right]
  ~\lesssim~ \mathbb{E}\big[ \sigma+K\sqrt{\log n} \big].
  \label{eq:Median-M}
\end{equation}

Now that we have controlled the expected spectral norm of $\bm{M}$,
we can obtain concentration results by means of Talagrand's inequality
\cite{talagrand1995concentration}. See \cite[Pages 73-75]{tao2012topics}
for an introduction.

\begin{prop}[\textbf{Talagrand's inequality}]
\label{prop:Talagrand-inequality}
Let
$\Omega=\Omega_{1}\times\cdots\times\Omega_{N}$ and $\mathbb{P}=\mu_{1}\times\cdots\times\mu_N$
form a product probability measure. Each $\Omega_{l}$ is equipped
with a norm $\|\cdot\|$, and $\sup_{x_{l}\in\Omega_{l}}\|x_{l}\|\leq K$ holds for all $1\leq l\leq N$.
Define $d\left(x,y\right):=\sqrt{\sum_{l=1}^{N}\|x_{l}-y_{l}\|^{2}}$
for any $x,y\in\Omega$, and let $f:\Omega\rightarrow\mathbb{R}$
be a 1-Lipschitz convex function with respect to $d\left(\cdot,\cdot\right)$.
Then  there exist some absolute constants $C,c>0$ such that
\begin{equation}
  \mathbb{P}\left\{ \left|f(x) - \mathsf{Median}\left[f\left(x\right)\right]\right| \geq \lambda K\right\}
  \leq C\exp\left(-c\lambda^{2}\right),\qquad \forall \lambda.
\end{equation}
\end{prop}

 Let $\Omega_{1},\cdots,\Omega_{N}$ represent the sample spaces for
 $\bm{M}_{1,1},\cdots,\bm{M}_{n,n}$, respectively, and take $\|\cdot\|$
 to be the spectral norm. 
Clearly, for any $\bm{M},\tilde{\bm{M}}\in\mathbb{R}^{nm\times nm}$,
 one has 
 \[
   \left|\|\bm{M}\|-\|\tilde{\bm{M}}\|\right|\leq\|\bm{M}-\tilde{\bm{M}}\|\leq\sqrt{\sum\nolimits _{i\leq j}\|\bm{M}_{i,j}-\tilde{\bm{M}}_{i,j}\|^{2}}
  :=d\big(\bm{M},\tilde{\bm{M}} \big).
 \]
 %
Consequently, Talagrand's inequality together with (\ref{eq:Median-M})
implies that with probability $1-O(n^{-11})$,
\begin{equation}
  \|\bm{M}\| ~\lesssim~ \mathsf{Median}\left[\|\bm{M}\|\right]+K\sqrt{\log n} ~\lesssim~ \mathbb{E}\left[\sigma\right]+K\sqrt{\log n}.
  \label{eq:M-concentrate-2}
\end{equation}

Finally, if $\mathbb{P}\left\{ \bm{M}_{i,j}=\bm{0}\right\} =p_{\mathrm{obs}}$
for some $p_{\mathrm{obs}}\gtrsim\log n/n$, then the Chernoff bound when combined with the union bound
indicates that 
\[
  \sum\nolimits _{j}\|\bm{M}_{i,j}\|^{2}\leq K^{2}\sum\nolimits _{j}\mathbb{I}\left\{ \bm{M}_{i,j} = \bm{0}\right\} 
  \lesssim K^{2}np_{\mathrm{obs}}, \qquad 1\leq i\leq n
\]
with probability $1-O\left(n^{-10}\right)$, 
%
which in turn gives
\[
  \mathbb{E}\left[\sigma\right] ~\lesssim~ \left(1-O\left(n^{-10}\right)\right) \sqrt{K^2 np_{\mathrm{obs}}}+O\left(n^{-10}\right)K\sqrt{n}
  ~\lesssim~ K\sqrt{np_{\mathrm{obs}}}.
\]
This together with (\ref{eq:M-concentrate-2}) as well as the assumption
$p_{\mathrm{obs}}\gtrsim\log n/n$ concludes the proof.

\section{Proof of Lemma \ref{lem:deviation-Lhat}\label{sec:Proof-of-Lemma-deviation-Lhat}}

The first step is to see that
\[
  \bm{L}-\mathbb{E}\left[\bm{L}\right]=\bm{L}^{\mathrm{debias}}-\mathbb{E}\left[\bm{L}^{\mathrm{debias}}\right],
\]
where $\bm{L}^{\mathrm{debias}}$ is a debiased version of $\bm{L}$
given in (\ref{eq:debiased-MLE}). This can be shown by recognizing
that 
\begin{align*}
  \bm{1}^{\top}\bm{L}_{i,j}\bm{1}
	& = \sum_{1\leq\alpha,\beta\leq m}\ell\left(z_{i}=\alpha,z_{j}=\beta;\text{ }y_{i,j}\right)
  = \sum_{1\leq\alpha,\beta\leq m}\log\mathbb{P}\left(\eta_{i,j} = y_{i,j}-(\alpha-\beta)\right) \\
	& = m \sum_{z=0}^{m-1} \log  \mathbb{P}\left(\eta_{i,j} = z \right)
\end{align*}
for any $i\neq j$, which is a fixed constant irrespective of $y_{i,j}$. 

The main advantage to work with $\bm{L}^{\mathrm{debias}}$ is that
each entry of $\bm{L}^{\mathrm{debias}}$ can be written as a linear
combination of the log-likelihood ratios. Specifically, for any $1\leq\alpha,\beta\leq m$,
\begin{eqnarray}
  \left(\bm{L}_{i,j}^{\mathrm{debias}}\right)_{\alpha,\beta} 
	& = & \ell ( z_{i} =\alpha,z_{j}=\beta; ~y_{i,j} ) - \frac{1}{m}\sum_{l=0}^{m-1} \log P_{l}\left(y_{i,j}\right) \nonumber \\
	& = & \log\mathbb{P}_{\alpha-\beta}\left(y_{i,j}\right)-\frac{1}{m}\sum_{l=0}^{m-1}\log P_{\alpha-\beta+l} (y_{i,j})  \nonumber \\
	& = & \frac{1}{m}\sum_{l=1}^{m-1}\log\frac{P_{\alpha-\beta}\left(y_{i,j}\right)}{P_{\alpha-\beta+l} (y_{i,j}) } 
  	 =  \frac{1}{m}\sum_{l=1}^{m-1}\log\frac{P_{0}\left(y_{i,j}-\alpha+\beta\right)}{P_{l}\left(y_{i,j}-\alpha+\beta\right)}.
  \label{eq:Lij-ab-UB-1}
\end{eqnarray}
Since $\bm{L}_{i,j}^{\mathrm{debias}}$ is circulant, its spectral
norm is bounded by the $\ell_{1}$ norm of any of its column: 
\begin{eqnarray}
  \left\Vert \bm{L}_{i,j}^{\mathrm{debias}}\right\Vert  
   & \leq & \sum_{\alpha=1}^{m}\left|\left(\bm{L}_{i,j}^{\mathrm{debias}}\right)_{\alpha,1}\right|
  \leq  \frac{1}{m}\sum_{l=1}^{m-1}\sum_{\alpha=1}^{m}\left|\log\frac{P_{0}\left(y_{i,j}-\alpha+1\right)}{P_{l}\left(y_{i,j}-\alpha+1\right)}\right|\\
   & = & \frac{1}{m}\sum_{l=1}^{m-1}\sum_{y=1}^{m}\left|\log\frac{P_{0}\left(y\right)}{P_{l}\left(y\right)}\right|
  \leq  \frac{1}{m}\sum_{l=1}^{m-1}\left\Vert \log\frac{P_{0}}{P_{l}}\right\Vert _{1}. 
  \label{eq:debias-block-UB}
\end{eqnarray}
To finish up, apply Lemma \ref{lem:random-matrix-norm} to arrive at
\begin{eqnarray*}
  \left\Vert \bm{L}-\mathbb{E}[\bm{L}]\right\Vert 
	&=&  \left\Vert \bm{L}^{\mathrm{debias}} - \mathbb{E}[ \bm{L}^{\mathrm{debias}} ]\right\Vert 
   \lesssim  \left(\max_{i,j}\|\bm{L}_{i,j}^{\mathrm{debias}} \|\right) \sqrt{np_{\mathrm{obs}}}  \nonumber\\
	&\lesssim& \left(\frac{1}{m}\sum_{l=1}^{m-1}\left\Vert \log\frac{P_{0}}{P_{l}}\right\Vert _{1}\right)\sqrt{np_{\mathrm{obs}}}.
\end{eqnarray*}

\section{Proofs of Lemma \ref{lemma:r-norm-KL} and Lemma \ref{lemma:r-norm-KL-q}\label{sec:Proof-of-Lemma-r-norm-KL}}

(1) \textbf{Proof of Lemma \ref{lemma:r-norm-KL}.} In view of (\ref{eq:Lz-general}),
for each $1\leq i\leq n$ one has
\begin{eqnarray}
  \|\bm{r}_{i}\|_{\infty} & \leq & \big\|\mathbb{E}\left[\bm{L}\right]\bm{h}\big\|_{\infty}+\left\Vert \bm{g}_{i}\right\Vert ,
  \label{eq:Sep_w-KL-LB}
\end{eqnarray}
where $\bm{g}:=\tilde{\bm{L}}\bm{z}$. In what follows, we will look
at each term on the right-hand side of (\ref{eq:Sep_w-KL-LB}) separately.
\begin{itemize}
\item \emph{The 1st term on the right-hand side of (\ref{eq:Sep_w-KL-LB}).} The feasibility constraint $\bm{z}_i,\bm{x}_i \in \Delta$ implies $\bm{1}^{\top}\bm{h}_{j}=0$, which enables us to express the $i$th block
of $\bm{f}:=\frac{1}{p_{\mathrm{obs}}}\mathbb{E}[\bm{L}] \bm{h}$
as
\begin{align*}
  \bm{f}_{i} & =\sum_{j:j\neq i}\bm{K}\bm{h}_{j}=\sum_{j:j\neq i}\bm{K}^{0}\bm{h}_{j} ,
\end{align*}
with $\bm{K}^{0}$ defined in (\ref{eq:defn-K-intuition}). By letting
$\bm{h}_{j,\backslash1}:=\left[h_{j,2},\cdots,h_{j,m}\right]^{\top}$
and denoting by $\bm{K}_{:,l}$ (resp.~$\bm{K}_{l,:}$) the $l$th
column (resp.~row) of $\bm{K}$, we see that
\begin{eqnarray}
  \bm{f}_{i} & = & \bm{K}_{:,1}^{0}\left(\sum _{j:j\neq i}h_{j,1}\right)+\bm{K}^{0}\sum _{j:j\neq i}\left[\begin{array}{c}
  0\\
  \bm{h}_{j,\backslash1}
  \end{array}\right].
  \label{eq:fi_1}
\end{eqnarray}
Recall from the feasibility constraint that $h{}_{j,1}\leq0$ and
$\bm{h}_{j,\backslash1}\geq\bm{0}$. Since $\bm{K}^{0}$ is a non-positive
matrix, one sees that the first term on the right-hand side of (\ref{eq:fi_1}) is non-negative,
whereas the second term  is non-positive. As a
result, the $l$th entry of $\bm{f}_{i}$---denoted by $f_{i,l}$---is
bounded in magnitude by
\begin{eqnarray*}
  \left|f_{i,l}\right| 
  & = & \left|K_{l,1}^{0}\left(\sum _{j:j\neq i}h_{j,1}\right) + \bm{K}_{l,:}^{0} \sum _{j:j\neq i}\left[\begin{array}{c}
0\\
\bm{h}_{j,\backslash1}
\end{array}\right]\right|\\
 & \leq & \max\left\{ |K_{l,1}^{0}| \sum _{j:j\neq i} |h_{j,1}|,~\left|\bm{K}_{l,:}^{0}\cdot\sum _{j:j\neq i}\left[\begin{array}{c}
0\\
\bm{h}_{j,\backslash1}
\end{array}\right]\right|\right\} \\
 & \leq & \max\left\{ \mathsf{KL}_{\max}\sum _{j:j\neq i}\left|h_{j,1}\right|,\text{ }\mathsf{KL}_{\max}\sum _{j:j\neq i}\left|\bm{1}^{\top}\bm{h}_{j,\backslash1}\right|\right\} \\
 & = & \mathsf{KL}_{\max}\sum _{j:j\neq i}\left|h_{j,1}\right|,
\end{eqnarray*}
where the last identity arises since $\bm{1}^{\top}\bm{h}_{j,\backslash1}=-h_{j,1}$.
Setting $$\overline{\bm{h}} = \big[\overline{h}_i\big]_{1\leq i\leq m} := \frac{1}{n}\sum_{i=1}^{n}\bm{h}_{i},$$
we obtain
\begin{eqnarray*}
  \left\Vert \bm{f}_{i}\right\Vert _{\infty} 
  & \leq & \mathsf{KL}_{\max}\sum\nolimits _{j=1}^{n}\left|h_{j,1}\right|
  = n\mathsf{KL}_{\max} \left| \overline{h}_{1} \right|
  \leq n\mathsf{KL}_{\max}\|\overline{\bm{h}}\|_{\infty}.
\end{eqnarray*}
If we can further prove that
\begin{equation}
  \left\Vert \overline{\bm{h}}\right\Vert _{\infty}\leq\min\left\{ \frac{k}{n}, ~ \epsilon\right\} ,
  \label{eq:hbar-bound}
\end{equation}
then we will arrive at the upper bound
\begin{equation}
  \left\Vert \mathbb{E}\left[\bm{L}\right]\bm{h}\right\Vert _{\infty} 
  = \left\Vert p_{\mathrm{obs}}\bm{f}_{i}\right\Vert _{\infty}
  \leq np_{\mathrm{obs}}\mathsf{KL}_{\max}\|\overline{\bm{h}}\|_{\infty}
  \leq np_{\mathrm{obs}}\mathsf{KL}_{\max}\min\left\{ \frac{k}{n},\epsilon\right\} .
  \label{eq:Khbar}
\end{equation}
%
%
To see why (\ref{eq:hbar-bound}) holds, we observe that (i) the constraint
$\bm{z}_{i}\in\Delta$ implies $\left|h_{i,l}\right|\leq1$, revealing
that 
\[
  \left\Vert \overline{\bm{h}}\right\Vert _{\infty}
  \leq \max_{1\leq l\leq m}\frac{1}{n}\sum\nolimits _{i=1}^{n}\left|h_{i,l}\right|
  \leq \frac{1}{n}\|\bm{h}\|_{*,0}=\frac{k}{n},
\]
and (ii) by Cauchy-Schwarz, 
\begin{eqnarray*}
  \left\Vert \overline{\bm{h}}\right\Vert _{\infty} 
   =  \frac{1}{n}\left\Vert \left[\begin{array}{ccc}
  \bm{I}_{m} & \cdots & \bm{I}_{m}\end{array}\right]\bm{h}\right\Vert _{\infty}\leq\frac{1}{n}\left(\sqrt{n}\cdot\left\Vert \bm{h}\right\Vert \right)=\epsilon.
\end{eqnarray*}

\item It remains to bound \emph{the 2nd term on the right-hand side of (\ref{eq:Sep_w-KL-LB}).}
Making use of Lemma \ref{lem:deviation-Lhat} gives 
\begin{eqnarray}
  \left\Vert \bm{g}\right\Vert 
   & \leq & \|\tilde{\bm{L}}\|\left\Vert \bm{z}\right\Vert 
  \lesssim\left\{ \frac{1}{m}\sum_{l=1}^{m-1}\left\Vert \log\frac{P_{0}}{P_{l}}\right\Vert _{1}\right\} \sqrt{np_{\mathrm{obs}}}\sqrt{n}
  \label{eq:g-norm}
\end{eqnarray}
with probability $1-O\left(n^{-10}\right)$. Let $\|\bm{g}\|_{(1)}\geq\|\bm{g}\|_{(2)}\geq\cdots\geq\|\bm{g}\|_{(n)}$
denote the order statistics of $\|\bm{g}_{1}\|$, $\cdots$, $\|\bm{g}_{n}\|$.
Then, for any $0<\rho<1$, 
\begin{equation}
  \|\bm{g}\|_{(\rho k^{*})}^{2} \leq \frac{1}{\rho k^{*}} \sum_{i=1}^{\rho k^{*}}\|\bm{g}\|_{(i)}^{2}\leq\frac{1}{\rho k^{*}}\|\bm{g}\|^{2}
  \lesssim  \frac{n^{2}p_{\mathrm{obs}}\left\{ \frac{1}{m} \sum_{l=1}^{m-1}\left\Vert \log\frac{P_{0}}{P_{l}}\right\Vert _{1}\right\} ^{2}}{\rho k^{*}},
  \label{eq:g2-UB}
\end{equation}
where $k^{*}$ is defined in (\ref{eq:defn-k-star}). In addition,
we have  $k/n \geq \xi $ and $\epsilon^2 \geq \xi^2 $ for some constant $0<\xi<1$ in the large-error regime (\ref{eq:large-error-regime}), and hence $$k^{*} = \min\{k, \epsilon^2 n \} \geq \xi^2 n.$$  Substitution into (\ref{eq:g2-UB}) yields
\begin{equation}
  \|\bm{g}\|_{(\rho k^{*})}\lesssim\left\{ \frac{1}{m}\sum_{l=1}^{m-1}\left\Vert \log\frac{P_{0}}{P_{l}}\right\Vert _{1}\right\} n\sqrt{\frac{p_{\mathrm{obs}}}{\rho k^{*}}}
  \leq  \left\{ \frac{1}{m}\sum_{l=1}^{m-1}\left\Vert \log\frac{P_{0}}{P_{l}}\right\Vert _{1}\right\} \sqrt{\frac{p_{\mathrm{obs}}n}{ \xi^2 \rho}}.
  \label{eq:g_rho_general}
\end{equation}
Consequently, if we denote by $\mathcal{I}$ the index set of those blocks
$\bm{g}_{i}$ satisfying $$\|\bm{g}_{i}\|\lesssim\left\{ \frac{1}{m}\sum_{l=1}^{m-1}\left\Vert \log\frac{P_{0}}{P_{l}}\right\Vert _{1}\right\} \sqrt{\frac{p_{\mathrm{obs}}n}{\xi^2\rho}},$$
then one has 
\[
  |\mathcal{I}|\geq n-\rho k^{*}.
\]

\end{itemize}
We are now ready to upper bound $\|\bm{r}_{i}\|_{\infty}$. For each
$i\in\mathcal{I}$ as defined above, 
\begin{eqnarray}
  (\ref{eq:Sep_w-KL-LB}) & \leq & np_{\mathrm{obs}}\mathsf{KL}_{\max}\min\left\{ \frac{k}{n},\epsilon\right\} +O\left(\left\{ \frac{1}{m}\sum_{l=1}^{m-1}\left\Vert \log\frac{P_{0}}{P_{l}}\right\Vert _{1}\right\} \sqrt{\frac{np_{\mathrm{obs}}}{\xi^2\rho}}\right)\nonumber \\
   & \leq & np_{\mathrm{obs}}\mathsf{KL}_{\max}\min\left\{ \frac{k}{n},\epsilon\right\} 
    + \alpha np_{\mathrm{obs}}\mathsf{KL}_{\max}
  \label{eq:Sep-w-UB}
\end{eqnarray}
for some arbitrarily small constant $\alpha>0$, with the proviso that
\begin{equation}
  np_{\mathrm{obs}}\mathsf{KL}_{\max}
  \geq 
  c_9 \left\{ \frac{1}{m}\sum_{l=1}^{m-1}\left\Vert \log\frac{P_{0}}{P_{l}}\right\Vert _{1}\right\} \sqrt{\frac{p_{\mathrm{obs}}n}{\xi^2 \rho}}
  \label{eq:condition-KL}
\end{equation}
for some sufficiently large constant $c_9>0$.
Since $\xi>0$ is assumed to be a fixed positive constant, the condition (\ref{eq:condition-KL}) can be satisfied if  we pick $$\rho=c_{5}\frac{\left\{ \frac{1}{m}\sum_{l=1}^{m-1}\left\Vert \log\frac{P_{0}}{P_{l}}\right\Vert _{1}\right\} ^{2}}{np_{\mathrm{obs}}\mathsf{KL}_{\max}^{2}}$$
for some sufficiently large constant $c_{5}>0$. Furthermore, in order to guarantee
 $\rho<1$,
one would need 
\begin{equation}
  \frac{\mathsf{KL}_{\max}^{2}}{\left\{ \frac{1}{m}\sum_{l=1}^{m-1}\left\Vert \log\frac{P_{0}}{P_{l}}\right\Vert _{1}\right\} ^{2}}
  \geq
  \frac{c_{10}}{p_{\mathrm{obs}}n}
  \label{eq:condition-KL-2}
\end{equation}
for some sufficiently large constant $c_{10}>0$.

It is noteworthy that if $m$ is fixed and if $\min_{l,y}P_{l}(y)$
is bounded away from 0, then
\begin{eqnarray}
  \left\Vert \log\frac{P_{0}}{P_{l}}\right\Vert _{1} & \leq & \sum_{y:P_{0}(y)\geq P_{l}(y)}\log\left(1+\frac{P_{0}\left(y\right)-P_{l}\left(y\right)}{P_{l}\left(y\right)}\right)+\sum_{y:P_{0}(y)<P_{l}(y)}\log\left(1+\frac{P_{l}\left(y\right)-P_{0}\left(y\right)}{P_{0}\left(y\right)}\right)\\
   & \leq & \frac{1}{\min_{l,y}P_{l}(y)}\sum_{y}\left|P_{0}(y)-P_{l}(y)\right|\text{ }\asymp\text{ }\mathsf{TV}\left(P_{0},P_{l}\right)\\
   & \overset{(\text{a})}{\lesssim} & \sqrt{\mathsf{KL}\left(P_{0}\|P_{l}\right)} \leq \sqrt{\mathsf{KL}_{\max}},  \label{eq:log1-KL}
\end{eqnarray}
where (a) comes from Pinsker's inequality. Thus, in this case (\ref{eq:condition-KL-2})
would follow if $\mathsf{KL}_{\max}\gg1/(np_{\mathrm{obs}})$. 

\vspace{0.5em}

(2) \textbf{Proof of Lemma \ref{lemma:r-norm-KL-q}.} 
This part can be shown using
similar argument as in the proof of Lemma \ref{lemma:r-norm-KL}.
Specifically, from the definition (\ref{eq:decompose-w-s}) we have
\begin{eqnarray}
  \|\bm{q}_{i}\|_{\infty} & \leq & \left\Vert \mathbb{E}\left[\bm{L}\right]\bm{h}\right\Vert _{\infty}+\left\Vert \hat{\bm{g}}_{i}\right\Vert
  \leq np_{\mathrm{obs}}\mathsf{KL}_{\max}\min\left\{ \frac{k}{n},\epsilon\right\} +\left\Vert \hat{\bm{g}}_{i}\right\Vert ,
\end{eqnarray}
where $\hat{\bm{g}}:=\tilde{\bm{L}}\bm{h}$, and the last inequality
is due to (\ref{eq:Khbar}). Similar to (\ref{eq:g-norm}) and (\ref{eq:g2-UB}),
we get 
\begin{eqnarray*}
  \Vert \hat{\bm{g}} \Vert  
  & \leq & \|\tilde{\bm{L}}\| \Vert \bm{h} \Vert 
  \lesssim \left\{ \frac{1}{m}\sum_{l=1}^{m-1}  \left\Vert \log\frac{P_{0}}{P_{l}}\right\Vert _{1}  \right\} \sqrt{np_{\mathrm{obs}}} (\epsilon\sqrt{n}) , \\
  \| \hat{\bm{g}} \|_{(\rho k^{*})} & \leq & \frac{1}{\sqrt{\rho k^{*}}}\|\hat{\bm{g}}\|
  \asymp \frac{\left\{ \frac{1}{m}\sum_{l=1}^{m-1}\left\Vert \log\frac{P_{0}}{P_{l}}\right\Vert _{1}\right\} ( \epsilon\sqrt{n} )\sqrt{np_{\mathrm{obs}}}}{\sqrt{\rho k^{*}}} \\
	&~\overset{(\text{b})}{\lesssim}& 
  \left\{ \frac{1}{m}\sum_{l=1}^{m-1}  \left\Vert \log\frac{P_{0}}{P_{l}}  \right\Vert _{1}\right\} \sqrt{\frac{np_{\mathrm{obs}}}{\rho}},
\end{eqnarray*}
where (b) arises since
\begin{eqnarray*}
  \frac{\epsilon\sqrt{n}}{\sqrt{k^{*}}} & = & \frac{\|\bm{h}\|}{\sqrt{k}} ~\overset{\text{(c)}}{\leq} \frac{\sqrt{2\|\bm{h}\|_{*,0}}}{\sqrt{k}} \leq \sqrt{2},
  \quad\text{if } k ~\leq \epsilon^{2}n;  \\
  \frac{\epsilon\sqrt{n}}{\sqrt{k^{*}}} & = & \frac{\epsilon\sqrt{n}}{\sqrt{\epsilon^{2}n}} ~= 1,
  \qquad\qquad\qquad\quad\text{if } k > \epsilon^{2}n.
\end{eqnarray*}
Here, (c) results from Fact \ref{fact:H}. One can thus find an index set $\mathcal{I}$ with cardinality
$|\mathcal{I}|\geq n-\rho k^{*}$ such that
\begin{equation}
  \| \hat{\bm{g}}_{i} \| ~\lesssim~ 
  \left\{ \frac{1}{m}\sum_{l=1}^{m-1}  \left\Vert \log\frac{P_{0}}{P_{l}}\right\Vert _{1}\right\} \sqrt{\frac{np_{\mathrm{obs}}}{\rho}},
  \qquad i\in\mathcal{I},
  \label{eq:gi-hat}
\end{equation}
where the right-hand side of (\ref{eq:gi-hat}) is identical to that
of (\ref{eq:g_rho_general}). Putting these bounds together and repeating
the same argument as in (\ref{eq:Sep-w-UB})-(\ref{eq:log1-KL}) complete the proof.

\section{Proof of Lemma \ref{lem:sep-s}\label{sec:Proof-of-Lemma-Chernoff-UB}}

By definition, for each $1\leq i\leq n$ one has
\[
  s_{i,1} - s_{i,l} = \sum_{j: ~ (i,j)\in\Omega }  \log\frac{ P_{0}(y_{i,j}) }{ P_{l-1}(y_{i,j}) }
  , \qquad 2 \leq l \leq m,
\]
which is a sum of independent log-likelihood ratio statistics. The main
ingredient to control $s_{i,1}-s_{i,l}$ is to establish
the following lemma. 

\begin{lem}
\label{lem:Hellinger-UB}
Consider two sequences of probability
distributions $\left\{ P_{i}\right\} $ and $\left\{ Q_{i}\right\} $
on a finite set $\mathcal{Y}$. Generate $n$ independent random variables
$y_{i}\sim P_{i}$. 

(1) For any $\gamma\geq0$,
\begin{equation}
  \mathbb{P}\left\{ \sum\nolimits _{i=1}^{n}\log\frac{P_{i}\left(y_{i}\right)}{Q_{i}\left(y_{i}\right)}\leq\gamma\right\} \leq\exp\left\{ -\sum\nolimits _{i=1}^{n} \mathsf{H}^{2}\left(P_{i},Q_{i}\right) + \frac{1}{2}\gamma\right\} ,
  \label{eq:Tail-Hel}
\end{equation}
where  $\mathsf{H}^{2}(P_i,Q_i) := \frac{1}{2} \sum_{y} (\sqrt{P_i(y)} - \sqrt{Q_i(y)})^2$. 

(2) Suppose that $\left|\log\frac{Q_{i}\left(y_{i}\right)}{P_{i}\left(y_{i}\right)}\right|\leq K$.
If
\[
  \min_{1\leq i\leq n}\mathsf{KL}\left(P_{i}\hspace{0.3em}\|\hspace{0.3em}Q_{i}\right)
  \geq c_{2}\left\{ \frac{1}{n}\sqrt{\sum\nolimits _{i=1}^{n}{ \mathsf{Var} }_{y_{i}\sim P_{i}}\left[\log\frac{Q_{i}\left(y_{i}\right)}{P_{i}\left(y_{i}\right)}\right]\log\left(mn\right)}
  + \frac{K\log\left(mn\right)}{n}\right\} 
\]
for some sufficiently large constant $c_{2}>0$, then 
\begin{eqnarray}
  \sum\nolimits _{i=1}^{n}\log\frac{Q_{i}\left(y_{i}\right)}{P_{i}\left(y_{i}\right)} & \leq & -\frac{1}{2}n\min_{1\leq i\leq n}\mathsf{KL}\left(P_{i}\hspace{0.3em}\|\hspace{0.3em}Q_{i}\right)
  \label{eq:LB-LLR}
\end{eqnarray}
with probability at least $1-O\left(m^{-11}n^{-11}\right)$.\end{lem}

We start from the case where $m$ is fixed. When $p_{\mathrm{obs}}>c_{0}\log n/n$
for some sufficiently large $c_{0}>0$, it follows from the Chernoff
bound that 
\begin{equation}
  |\left\{ j:(i,j)\in\Omega\right\} |\geq(1-\zeta)np_{\mathrm{obs}},\qquad1\leq i\leq n
  \label{eq:concentration-deg}
\end{equation}
with probability at least $1-O(n^{-10})$, where $\zeta > 0$ is some small constant. Taken together, Lemma \ref{lem:Hellinger-UB}(1)
(with $\gamma$ set to be $2\zeta np_{\mathrm{obs}}\mathsf{H}_{\min}^{2}$),
(\ref{eq:concentration-deg}), and the union bound give 
\[
  s_{i,1} - s_{i,l}>2\zeta np_{\mathrm{obs}}\mathsf{H}_{\min}^{2}, \quad1\leq i\leq n, ~2\leq l\leq m
\]
or, equivalently,
\begin{equation}
  \mathscr{S}\left(\bm{s}_{i}\right)>2\zeta np_{\mathrm{obs}}\mathsf{H}_{\min}^{2},\qquad1\leq i\leq n
  \label{eq:sep-Hel}
\end{equation}
with probability exceeding $1-mn\exp\left\{ -\left(1-2\zeta\right)np_{\mathrm{obs}}\mathsf{H}_{\min}^{2}\right\} $.
As a result, (\ref{eq:sep-Hel}) would follow with probability at
least $1-\exp\left\{ -\zeta\log(mn)\right\} -O(n^{-10})$, as long
as
\begin{equation}
  \mathsf{H}_{\min}^{2}\geq\frac{1+\zeta}{1-2\zeta}\cdot\frac{\log n+\log m}{np_{\mathrm{obs}}}.
  \label{eq:Hel-condition-1-1}
\end{equation}

It remains to translate these results into a version based on the KL divergence.
Under Assumption \ref{assumption-Pmin}, it comes from \cite[Fact 1]{chen2015information}
that $\mathsf{H}^{2}\left(P_{0},P_{l}\right)$ and $\mathsf{KL}\left(P_{0}\hspace{0.3em}\|\hspace{0.3em}P_{l}\right)$
are orderwise equivalent. This allows to rewrite (\ref{eq:sep-Hel})
as
\begin{equation}
  \mathscr{S}\left(\bm{s}_{i}\right)>c_{5}\zeta np_{\mathrm{obs}}\mathsf{KL}_{\min},\qquad1\leq i\leq n
\end{equation}
for some constant $c_{5}>0$. In addition, Lemma \ref{lem:KL-Var-Hel} and \cite[Fact 1]{chen2015information} reveal that 
\[
  \mathsf{H}_{\min}^{2} \geq \frac{1}{4} {\left(1-c_{6}\sqrt{\mathsf{KL}_{\min}}\right)} {\mathsf{KL}_{\min}} \\
  \quad \text{and} \quad 
  \mathsf{H}_{\min}^{2} \geq c_8  {\mathsf{KL}_{\min}} 
\]
for some constants $c_{6}, c_8>0$. As a result,  (\ref{eq:Hel-condition-1-1})
would hold if
\begin{eqnarray}
  \frac{ \left(1-c_{6}\sqrt{\mathsf{KL}_{\min}}\right) \mathsf{KL}_{\min}}{4}
  \geq \frac{1+\zeta}{1-2\zeta}\cdot\frac{\log n+\log m}{np_{\mathrm{obs}}} 
  \label{eq:KL-LB-2}  \\
  \text{or} \quad
  c_8 \mathsf{KL}_{\min}
  \geq \frac{1+\zeta}{1-2\zeta}\cdot\frac{\log n+\log m}{np_{\mathrm{obs}}} 
\end{eqnarray}
When both $\frac{\log n}{np_{\mathrm{obs}}}$ and $\zeta$ are sufficiently
small, it is not hard to show that (\ref{eq:KL-LB-2}) is a consequence
of $\mathsf{KL}_{\min}\geq4.01\log n/(np_{\mathrm{obs}})$. 

Finally, the second part (\ref{eq:sep-si-1-2}) of Lemma \ref{lem:sep-s}
is straightforward by combining (\ref{eq:concentration-deg}), Lemma
\ref{lem:Hellinger-UB}(2), and the union bound. 

\begin{proof}[\textbf{Proof of Lemma \ref{lem:Hellinger-UB}}]

(1) For any $\gamma\geq0$, taking the Chernoff bound we obtain
\begin{eqnarray*}
  \mathbb{P}\left\{ \sum_{i=1}^{n}\log\frac{Q_{i}\left(y_{i}\right)}{P_{i}\left(y_{i}\right)}\geq-\gamma\right\}  
  & \leq & \frac{\prod_{i=1}^{n}\mathbb{E}_{y\sim P_{i}}\left[\exp\left(\frac{1}{2}\log\frac{Q_{i}(y)}{P_{i}\left(y\right)}\right)\right]}{\exp\left(-\frac{1}{2}\gamma\right)} \\
	&=& \exp\left(\frac{1}{2}\gamma\right)\prod_{i=1}^{n}\left(1-\mathsf{H}^{2}\left(P_{i},Q_{i}\right)\right),
\end{eqnarray*}
where the last identity follows since
\begin{align*}
 & \mathbb{E}_{y\sim P}\left[\exp\left(\frac{1}{2}\log\frac{Q(y)}{ P (y) } \right) \right] 
 = \mathbb{E}_{y\sim P}\left[\sqrt{\frac{Q(y)}{P(y)}}\right]=\sum_{y}\sqrt{P(y)Q(y)}\\
 & \quad=\sum_{y}\frac{1}{2}\left\{ P(y)+Q(y)-\left(\sqrt{P(y)}-\sqrt{Q(y)}\right)^{2}\right\} = 1-\mathsf{H}^{2}\left(P,Q\right).
\end{align*}
 The claim (\ref{eq:sep-si-1}) then follows by observing that $1-\mathsf{H}^{2}\left(P_{i},Q_{i}\right)\leq\exp\left\{ -\mathsf{H}^{2}\left(P_{i},Q_{i}\right)\right\} $. 

(2) Taking expectation gives 
\[
  \sum_{i=1}^{n}\mathbb{E}_{y_{i}\sim P_{i}}\left[ \log \frac{Q_{i} \left(y_{i}\right) }{ P_{i} (y_{i}) }\right]
  = - \mathsf{KL} ( P_{i}\hspace{0.3em}\|\hspace{0.3em}Q_{i} ).
\]
From our assumption $\left|\log\frac{Q_{i}\left(y_{i}\right)}{P_{i}\left(y_{i}\right)}\right|\leq K$,
the Bernstein inequality ensures the existence of some constants $c_{0},c_{1}>0$
such that
\begin{align}
 & \sum_{i=1}^{n}\left(\log\frac{Q_{i}\left(y_{i}\right)}{P_{i}\left(y_{i}\right)} + \mathsf{KL}\left(P_{i}\hspace{0.3em}\|\hspace{0.3em}Q_{i}\right)\right)\nonumber \\
 & \quad\leq\text{ }c_{0}\sqrt{\sum_{i=1}^{n}{\mathsf{Var}}_{y_{i}\sim P_{i}} 
 \left[ \log\frac{Q_{i}\left(y_{i}\right)}{P_{i}\left(y_{i}\right)} \right]\log\left(mn\right)} 
 + c_{1}\max_{i}\left|\log\frac{Q_{i}\left(y_{i}\right)}{P_{i}\left(y_{i}\right)}\right| \log (mn)
  \label{eq:UB4}
\end{align}
 with probability at least $1-O\left(m^{-11}n^{-11}\right)$. This
taken collectively with the assumption
\[
  n\min_{1\leq i\leq n}\mathsf{KL}\left(P_{i}\hspace{0.3em}\|\hspace{0.3em}Q_{i}\right)
  \gg
  \sqrt{\sum_{i=1}^{n}{\mathsf{Var}}_{y_{i}\sim P_{i}} \left[\log\frac{Q_{i}\left(y_{i}\right)}{P_{i}\left(y_{i}\right)}\right] \log\left(mn\right)} + K\log\left(mn\right)
\]
establishes (\ref{eq:LB-LLR}).
\end{proof}

\section{Proof of Lemma \ref{lem:simple-condition-general}\label{sec:Proof-of-Lemma-simple-condition}}

To begin with, (\ref{eq:KLmax-general-m}) is an immediate consequence
from (\ref{eq:SNR-condition-general-m}) and Assumption \ref{assumption-KL-D}. Next, 
\[
  \frac{\mathsf{KL}_{\max}}{\max_{l,y}\left|\log\frac{P_{0}\left(y\right)}{P_{l}\left(y\right)}\right|}
  \geq  \frac{\mathsf{KL}_{\max}}{\max_{1\leq l<m}\left\Vert \log\frac{P_{0}}{P_{l}}\right\Vert _{1}}\overset{(\text{a})}{\gg}\frac{1}{\sqrt{np_{\mathrm{obs}}}}
  \overset{(\text{b})}{\gtrsim}\frac{\log\left(mn\right)}{np_{\mathrm{obs}}}.
\]
where (a) arises from (\ref{eq:SNR-condition-general-m}) together with Assumption \ref{assumption-KL-D}, and (b)
follows as soon as $p_{\mathrm{obs}}\gtrsim\frac{\log^{2}(mn)}{n}$.
This establishes the second property of (\ref{eq:sep-si-1-2}).

Next, we turn to the first condition of (\ref{eq:sep-si-1-2}). If 
$
  \max_{y}P_{0}\left(y\right)\lesssim1/\log\left(mn\right)
$ 
holds, then we can derive
\begin{eqnarray*}
 {\mathsf{Var}}_{y\sim P_{0}}\left[\log\frac{P_{0}\left(y\right)}{P_{l} (y)}\right] 
  & \leq & \sum_{y}P_{0}(y)\left(\log\frac{P_{0}\left(y\right)}{P_{l}\left(y\right)}\right)^{2}
  \lesssim\frac{1}{\log\left(mn\right)}\sum_{y}\left(\log\frac{P_{0}\left(y\right)}{P_{l}\left(y\right)}\right)^2 \\
 & \leq & \frac{1}{\log\left(mn\right)}\left\Vert \log\frac{P_{0}}{P_{l}}\right\Vert ^{2},
\end{eqnarray*}
which together with (\ref{eq:SNR-condition-general-m}) implies that
\begin{equation}
  \frac{\mathsf{KL}_{\min}^{2}}{\max_{0\leq l<m}{\mathsf{Var}}_{y\sim P_{0}} \left[\log\frac{P_{0}\left(y\right)}{P_{l}\left(y\right)}\right]}
  \gtrsim  \log\left(mn\right)\cdot\frac{ \mathsf{KL}_{\min}^{2} }{ \max_{1\leq l<m}\left\Vert \log\frac{P_{0}}{P_{l}}\right\Vert ^{2} }
  \geq  \frac{c_{4}\log\left(mn\right)}{np_{\mathrm{obs}}}.
  \label{eq:general-2-1}
\end{equation}

Finally, consider the complement regime in which $\max_{y}P_{0}\left(y\right) \gg \frac{1}{\log\left(mn\right)}$, 
which obeys $\max_{y}P_{0}\left(y\right) \gg 1/m$ as long as $m \gtrsim \log(n)$. Suppose
 $P_{0}(y_{0}) =\max_{y}P_{0}(y)$ and $\frac{P_{0}(y_{0})}{P_{l}(y_{0})}=\max_{j,y}\frac{P_{0}(y)}{P_{j}(y)}$ hold for some $l$ and $y_0$, then it follows from the preceding inequality $\max_{y}P_{0}\left(y\right) \gg 1/m$ that $\frac{P_{0}(y_{0})}{P_{l}(y_{0})} \geq 2$.
Under Assumption \ref{assumption-KL-D}, the KL divergence is lower bounded by  
\[
  \mathsf{KL}_{\min}\asymp\mathsf{KL}_{l} 
  \gtrsim \mathsf{H}^{2}\left(P_{0},P_{l}\right) 
  \geq \frac{1}{2}\left(\sqrt{P_{0}(y_{0})}-\sqrt{P_{l}(y_{0})}\right)^{2}
  \asymp  P_{0}(y_{0})  \gg  \frac{1}{\log(mn)}.
\]
%
Moreover, the assumption $m=n^{O(1)}$ taken collectively with Assumption
\ref{assumption-Pmin} ensures 
\[
  \max_{j,y}\left|\log\frac{P_{0}(y)}{P_{j}(y)}\right| \lesssim \log n,
\]
allowing one to bound 
\[
  {\mathsf{Var}}_{y\sim P_{0}}\left[\log\frac{P_{0}\left(y\right)}{P_{l}\left(y\right)}\right] 
  \leq \max_{j,y}\left|\log\frac{P_{0}(y)}{P_{j}(y)}\right|^{2}
  \lesssim \log^{2}n.
\]
Combining the above inequalities we obtain
\[
  \frac{\mathsf{KL}_{\min}^{2}}{\max_{0\leq l<m}{ \mathsf{Var} }_{y\sim P_{0}}\left[\log\frac{P_{0}\left(y\right)}{P_{l}\left(y\right)}\right]}
  \gg \frac{\frac{1}{\log^{2}\left(mn\right)}}{\log^{2}n}
  \gtrsim \frac{1}{\log^{4}(mn)}
  \gtrsim \frac{\log\left(mn\right)}{np_{\mathrm{obs}}},
\]
as long as $p_{\mathrm{obs}}\gtrsim \frac{\log^{5}(mn)}{n}$, as claimed.

\section{Proof of Lemma \ref{lem:moderate-deviation}\label{sec:Proof-of-Lemma-MDP}}

For the sake of conciseness, we will only prove (\ref{eq:MDP-k}), as (\ref{eq:MDP-k-2}) can be shown using the same argument. 
We recognize that (\ref{eq:MDP-k}) can be established by demonstrating
the moderate deviation principle with respect to 
\[
  S_{n} = \sum_{i=1}^{n}\left(\zeta_{i,n}-\mathbb{E}\left[\zeta_{i,n}\right]\right)\qquad\text{where }\zeta_{i,n}
  :=  \frac{ \log\frac{Q_{n}\left(y_{i,n}\right) } { P_{n}\left(y_{i,n}\right)}}{\sigma_{n}}.
\]
To be precise, the main step is to invoke \cite[Theorem 6]{merlevede2009bernstein}
to deduce that
\begin{eqnarray}
  a_{n}\log\mathbb{P}\left\{ \sqrt{a_{n}}\frac{S_{n}}{\sqrt{n}}>\tau\right\}   =  -\left(1+o_{n}\left(1\right)\right)\frac{\tau^{2}}{2}
  \label{eq:MDP-Sk}
\end{eqnarray}
for any constant $\tau>0$, where
\begin{equation}
	a_{n}:= \frac{\sigma_{n}^{2}}{n\mu_{n}^{2}}.
  \label{eq:defn-ak}
\end{equation}
In fact, one can connect the event $\left\{ \sqrt{a_{n}}\frac{S_{n}}{\sqrt{n}}>\tau\right\} $
with the likelihood ratio test because
\begin{eqnarray*}
  \mathbb{P}\left\{ \sqrt{a_{n}}\frac{S_{n}}{\sqrt{n}}>\tau\right\}  
  & = & \mathbb{P}\left\{ \sum_{i=1}^{n}\frac{\log\frac{Q_{n}(y_{i,n})}{P_{n}(y_{i,n})}+\mu_{n}}{\sigma_{n}}>n\frac{\mu_{n}}{\sigma_{n}}\tau\right\}  \\
	& = & \mathbb{P} \left\{ \sum_{i=1}^{n}\log\frac{Q_{n}(y_{i,n})}{P_{n}(y_{i,n})}>\left(\tau-1\right)n\mu_{n}  \right\} .
\end{eqnarray*}
This reveals that (\ref{eq:MDP-Sk}) is equivalent to 
\begin{eqnarray}
  \mathbb{P}\left\{ \sum\nolimits _{i=1}^{n} \log \frac{Q_{n}(y_{i,n})}{P_{n}(y_{i,n})} > \left( \tau-1 \right)n\mu_{n} \right\}  
  & = & \exp \left( -\left(1+o_{n}( 1) \right)n\frac{\tau^{2}\mu_{n}^{2}}{2\sigma_{n}^{2}} \right)
  \label{eq:MDP-proof}
\end{eqnarray}
as claimed in the first identity of (\ref{eq:MDP-k}). Moreover, by Lemma \ref{lem:KL-Var-Hel}, it is seen that $2\mu_{n}=\left(1+O\left(\sqrt{\mu_{n}}\right)\right)\sigma_{n}^{2}$
 in the regime considered herein, leading to the second identity of (\ref{eq:MDP-k}). Hence, it suffices to prove (\ref{eq:MDP-Sk}). 

In order to apply \cite[Theorem 6]{merlevede2009bernstein}, we need to check that the double indexed
sequence $$\left\{ \zeta_{i,n}-\mathbb{E}\left[\zeta_{i,n}\right]:\text{ }i\leq n\right\} _{n\geq1}$$
satisfies the conditions required therein. First of all, the independence
assumption gives 
\[
  v^{2} := \sup_{i,n}\left\{ \mathsf{Var}  \left[\zeta_{i,n}\right] + 2 \sum\nolimits _{j>i}\left|{\mathsf{Cov}} \left(\zeta_{i,n},\zeta_{j,n}\right)\right|\right\} =\sup_{i,n}\left\{ { \mathsf{Var} }\left[\zeta_{i,n}\right]\right\} = 1 < \infty.
\]
Second, 
\[
  \frac{{\mathsf{Var}}\left[S_{n}\right]}{n}=\frac{\sum_{i=1}^{n}{ \mathsf{Var} } \left[ \zeta_{i,n} \right]}{n} = 1 > 0.
\]
Third, it follows from Lemma \ref{lemma:V-KL-UB} that
\begin{eqnarray}
  M_{n}:=\sup_{i}\left|\zeta_{i,n}\right| = \frac{\sup_{i}\left|\log\frac{Q_{n}\left(y_{i,n}\right)}{P_{n}\left(y_{i,n}\right)}\right|}{\sigma_{n}} & \leq & \frac{1}{\min\left\{ P_{n}\left(y_{i,n}\right),Q_{n}\left(y_{i,n}\right)\right\} } \frac{\sqrt{2\mu_{n}}}{\sigma_{n}}
  \asymp   \frac{\sqrt{\mu_{n}}}{\sigma_{n}}.
  \label{eq:bound-Mk}
\end{eqnarray}
The assumption $\frac{\mu_{n}^{2}}{\sigma_{n}^{2}}\asymp\frac{\log n}{n}$
further gives
\[
  a_{n} = \frac{ \sigma_{n}^{2} }{ n\mu_{n}^{2} } \asymp \frac{1}{\log n} = o_{n} (1).
\]
Moreover, making use of the bound (\ref{eq:bound-Mk}) as well as
the assumptions $\frac{\mu_{n}^{2}}{\sigma_{n}^{2}}\asymp\frac{\log n}{n}$
and $\mu_{n}\gtrsim\frac{\log n}{n}$, we derive 
\[
  \frac{na_{n}}{M_{n}^{2}\log^{4}n}  \gtrsim  \frac{n\frac{\sigma_{n}^{2}}{n\mu_{n}^{2}}}{\frac{\mu_{n}}{\sigma_{n}^{2}}\cdot\log^{4}n} = \frac{\sigma_{n}^{4}}{\mu_{n}^{3}\log^{4}n}=\frac{\sigma_{n}^{4}}{\mu_{n}^{4}}\cdot\frac{\mu_{n}}{\log^{4}n} \gtrsim \frac{n}{\log^{5}n} \rightarrow \infty.
\]
With these conditions in place, we can invoke \cite[Theorem 6]{merlevede2009bernstein}
to establish (\ref{eq:MDP-Sk}).

\bibliographystyle{IEEEtran}
\bibliography{bib_alignment}

\end{document}